\documentclass[preprint,12pt]{elsarticle}
\usepackage[utf8]{inputenc}
\usepackage{parskip}
\usepackage[left=2cm,right=2cm,top=2cm,bottom=2 cm]{geometry}
\usepackage{graphicx}
\usepackage{wrapfig}
\usepackage{amsmath}
\usepackage{caption}
\usepackage{comment}
\usepackage{ragged2e}
\usepackage{enumitem}
\usepackage[table]{xcolor}
\usepackage{multirow}
\usepackage{array}
\usepackage{subcaption}
\usepackage[numbers]{natbib}
\usepackage[colorlinks=true, linkcolor=blue, citecolor=blue]{hyperref}
\usepackage[noabbrev, nosort]{cleveref}
\usepackage{amsthm}

\theoremstyle{plain}  
\newtheorem{theorem}{Theorem}[section]
\newtheorem{lemma}{Lemma}[section]

\newtheorem{corollary}{Corollary}[section]

\theoremstyle{definition}  
\newtheorem{definition}{Definition}[section]

\newtheorem{remark}{Remark}[section]
\newtheorem{proposition}{Proposition}[section]

\newcolumntype{s}{>{\columncolor{yellow!20}}c}

\usepackage{amssymb}
\usepackage{mathrsfs}

\journal{Journal of the Mechanics and Physics of Solids}

\begin{document}

\begin{frontmatter}

\title{Revisiting the cofactor conditions: Elimination of transition layers in compound domains}

\author[inst1]{Mohd Tahseen}

\affiliation[inst1]{organization={Department of Aerospace Engineering},
            addressline={Indian Institute of Science},
            city={Bangalore},
            postcode={560012}, 
            country={India}}

\author[inst1]{Vivekanand Dabade}

\begin{abstract}

This paper investigates the conditions necessary for the elimination of transition layers at interfaces involving compound domains, extending the classical framework of cofactor conditions. Although cofactor conditions enable stress-free phase boundaries between Type I/II domains and austenite, their applicability to compound domains has remained limited. Here, we present a comprehensive theoretical framework to characterize all compatible interfaces, highlighting the fundamental importance of the commutation property among martensitic variants. By establishing necessary and sufficient algebraic conditions, referred to as \textit{extreme compatibility conditions}, we demonstrate the simultaneous elimination of transition layers at phase interfaces for both Type I/II and compound laminates, across all volume fractions of the martensitic variants. We also investigate the possibility of achieving supercompatibility in non-conventional twins, recently observed in the NiMnGa system \cite{hanus}. The focus of our work is on cubic-to-orthorhombic and cubic-to-monoclinic~II phase transformations, for which the extreme compatibility conditions are explicitly derived and systematically analyzed. The theory predicts novel zero-elastic-energy microstructures, including an increased number of triple clusters, spearhead-shaped martensitic nuclei, stress-free inclusions of austenite within martensite, and distinctive four-fold martensitic clusters. This significantly expands the possible modes of forming stress-free interfaces between phases and reveals new energy-minimizing microstructures that can facilitate the nucleation of martensite within austenite and vice versa. These configurations highlight significant enhancements in transformation reversibility and material durability, guiding the rational design of next-generation shape memory alloys with optimized functional properties.

\end{abstract}


\end{frontmatter}

\section{Introduction}
\label{sec:Intro}

The \textit{cofactor conditions} are conditions of enhanced geometric compatibility between the austenite and martensite phases of a material undergoing a reversible martensitic phase transformation. Such transformations are essential for realizing functional properties like the shape memory effect, superelasticity, and multiferroic behavior \cite{origami_inspired}, \cite{kumar2019shape}, \cite{srivastava2010direct}, \cite{haga2005medical}, \cite{mohdjani2014review}. A key challenge in leveraging the full potential of these properties is ensuring that the material can undergo highly reversible phase transformations over a large number of cycles without functional degradation. This is especially critical in demanding applications such as artificial heart valves, where components must withstand millions of cycles under physiological conditions; microactuators in precision devices, which rely on repeatable motion at small scales; and energy conversion systems, where repeated thermal or mechanical cycling is required to maintain efficiency over long operating lifetimes. The reversibility of these transformations is governed by the degree of compatibility between the two phases, i.e. the ease with which one phase can transform back and forth into the other.\\[5 pt]
An important factor determining the extent of this ease is the geometric compatibility between the two phases. This compatibility dictates how easily the two lattice structures can transition into one another, without inducing significant strain or plastic deformation. By manipulating and engineering this compatibility, the material behavior can be fine-tuned to enhance its reversibility and improve its durability.\\[5 pt]
A systematic way to engineer this compatibility is offered by the \textit{geometrically nonlinear theory of martensite}, which provides a rigorous framework for understanding structural phase transformations \cite{BALL200461}, \cite{JAMES2000197}, \cite{James2019}. According to this theory, when the middle eigenvalue ($\lambda_2$) of the transformation stretch tensor equals one, it is possible to form a perfectly compatible interface between austenite and a single variant of martensite, which carries zero elastic energy. Such interfaces are believed to move through the material with minimal resistance, enabling phase transformations to occur with low energetic cost. This mechanism is expected to enhance the transformation efficiency of the material.  \\[5 pt]
The condition $\lambda_2=1$ marked the beginning of what is now known as phase engineering, where the material composition is tuned to achieve a specified level of geometric compatibility  \cite{James2019}, \cite{gu2018phaseengineering}. Several materials have been successfully engineered to satisfy the condition $\lambda_2=1$, resulting in near-zero thermal hysteresis and improved reversibility under repeated thermal cycling \cite{Cui2006},\cite{Zhang2009}, \cite{Zarnetta2010}. Experiments have revealed a strong correlation between the closeness of $\lambda_2$ to $1$ and the improved functional performance of these materials, thereby supporting the hypothesis of zero elastic energy interfaces. This success not only lends credibility to the theoretical framework but also paves the way for the systematic discovery of new materials with superior functional properties. \\[5pt]
The cofactor conditions represent a more advanced approach of phase engineering as compared $\lambda_2 = 1$ criterion. These conditions were introduced by James and Zhang \cite{JamesZhang2005}, following the success of tuning $\lambda_2$ to $1$ to achieve low hysteresis materials. The key idea behind the cofactor conditions is to enable a martensitic laminate composed of two variants to be compatible with the austenite phase for all volume fractions, $\mu\in[0,1]$, including the limiting cases of $0$ and $1$. This generalization subsumes the special case where a single martensitic variant is compatible with austenite, making $\lambda_2 = 1$ a subset of the cofactor conditions.\\[5 pt]
In this sense, cofactor conditions extend the concept of compatibility beyond individual martensitic variants to the entire set of deformation gradients attainable through homogeneous mixtures of two variants. This enhanced compatibility between the austenite and martensite phases is known as supercompatibility. From a theoretical point of view, cofactor conditions emerge as the conditions of degeneracy for the equations of the \textit{crystallographic theory of martensite}, see equation \eqref{eq:17}. A consequence of this degeneracy is that the phase interface can form in infinitely many distinct ways, providing greater geometric flexibility for the transformation to proceed.\\[5pt]
A seminal contribution to the study of cofactor conditions is the work of Chen et al. \cite{CHEN20132566}, which provided an understanding of how these conditions lead to the formation of zero elastic energy interfaces in some cases. The study demonstrated that when the cofactor conditions hold for Type I or Type II domains (\textit{see \autoref{sec:kin_comp} for definition}), the transition layer between their laminates and austenite can be eliminated entirely by the formation of triple junctions and parallel domain walls, respectively, at the phase interface, see \autoref{fig:triple_1_2}. The implications of this result are significant. It means that when the cofactor conditions are satisfied for Type I or Type II domains, there exist an infinite number of deformation gradients formed by these domains, which form a zero elastic energy phase interface with austenite. This flexibility allows the advancing phase interface to accommodate defects, precipitates, and non-transforming inclusions, thereby mitigating the formation of regions with high localized strain.\\[5 pt]
Prompted by the promising theoretical insights, a select few materials have been engineered to satisfy the cofactor conditions, aiming to enhance their functional properties. The first such material was Zn$_{45}$Au$_{30}$Cu$_{25}$, which satisfies the cofactor conditions for both Type-I and Type-II twin systems \cite{Song2013}. This alloy exhibited reduced thermal hysteresis and demonstrated exceptional cyclic stability, exceeding $16,000$ transformation cycles. Interestingly, repeated thermal cycling in this alloy led to the emergence of a non-reproducible microstructure, the origin of which remains not fully understood to this date. Another important example is that of Ti$_{54}$Ni$_{34}$Cu$_{12}$, a thin-film alloy which closely satisfies the cofactor conditions for Type-I twins \cite{chluba}. This alloy demonstrated exceptional repeatability of superelastic behavior, even after $10^7$ loading cycles. These examples highlight the vital role of cofactor conditions in material design via phase engineering.\\[5pt]
\begin{figure}[hbt!]
    \centering
    \includegraphics[scale=0.55]{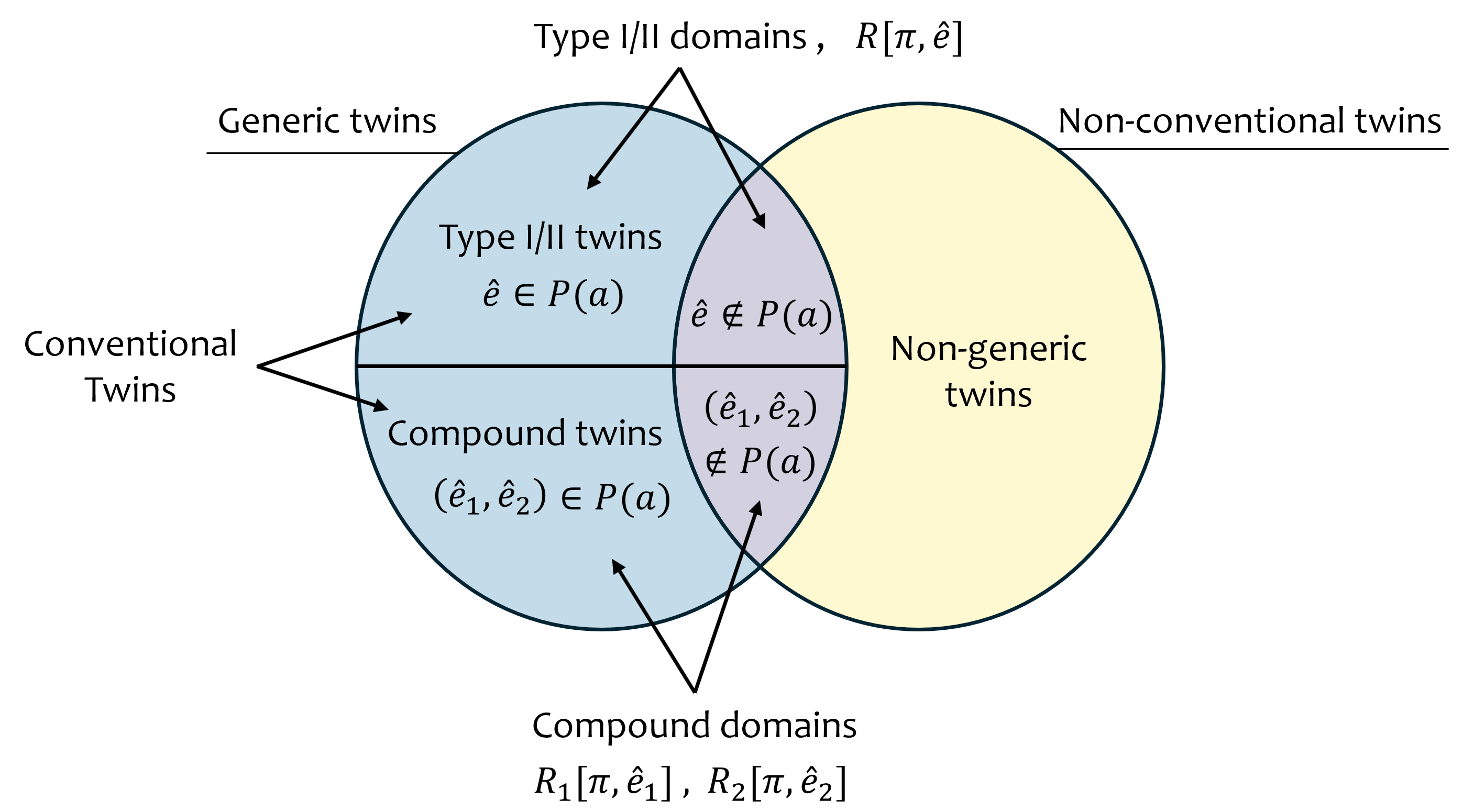}
    \caption{Classification of all transformation twins arising in martensitic phase transformations. }
    \label{fig:venn}
\end{figure}

\begin{figure}[hbt!]
    \centering
    \begin{minipage}[b]{0.50\textwidth}
    \centering
    \includegraphics[scale=0.39]{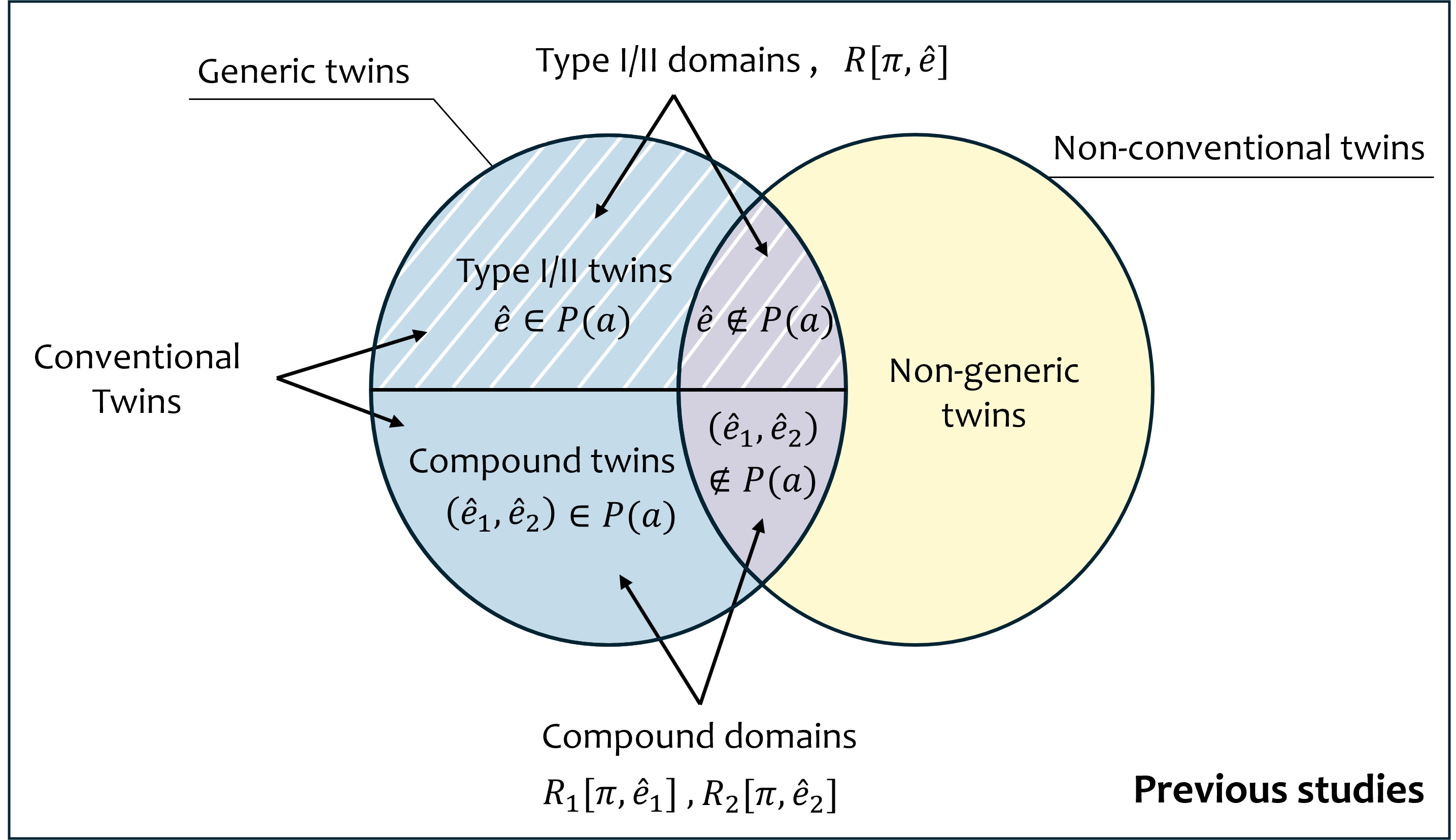}
    \subcaption{}
    \label{fig:venn_scope_a}
    \end{minipage}%
    \begin{minipage}[b]{0.50\textwidth}
    \centering
    \includegraphics[scale=0.39]{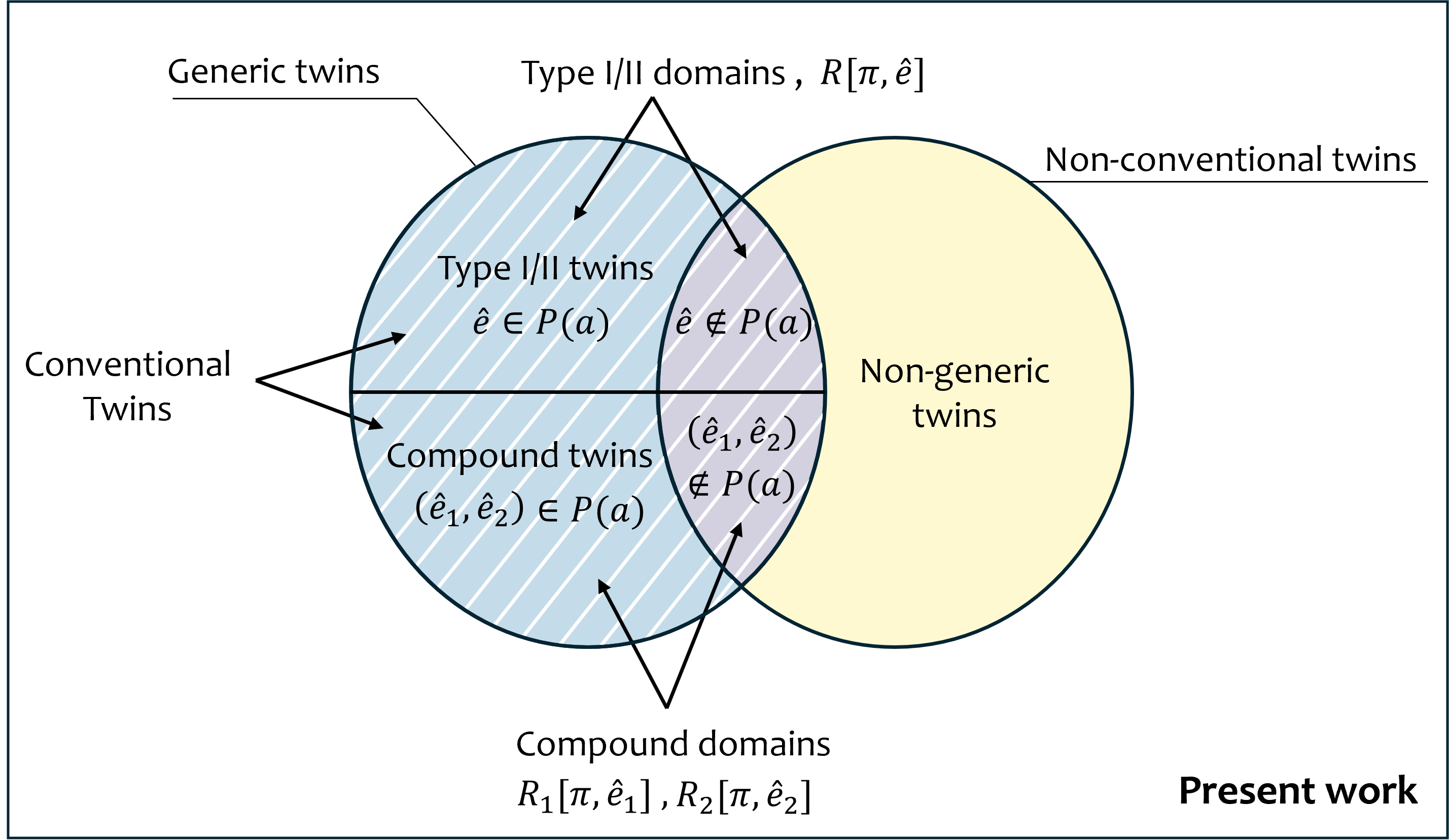}
    \subcaption{}
    \label{fig:venn_scope_b}
    \end{minipage}
    \caption{(a) Shaded region is where the elimination of transition layer is possible, when the cofactor conditions are satisfied (b) Shaded region shows the extension of this elimination to the remaining cases of generic twins. This highlights the scope of present work.}
    \label{fig:venn_scope}
\end{figure}

In contrast to Type I/II domains, cofactor conditions satisfied for compound domains do not lead to the elimination of the transition layer between compound laminates and austenite, see \Cref{fig:venn_scope_a}. The classification shown in \Cref{fig:venn_scope_a} is based on the categorization of solutions to the twinning equation employed by Chen et al.\cite{CHEN20132566} and illustrated in \Cref{fig:venn}. This limitation of cofactor conditions prevents the material from achieving a state in which all twinning systems can form phase interfaces without requiring intermediate transition layers. Achieving this level of crystallographic compatibility would broaden the admissible set of stress-free interfaces available in the material, thereby offering increased flexibility in accommodating microstructural heterogeneities and potentially leading to enhanced functional performance of the material.\\[5 pt]
The cofactor conditions have been further analyzed by Della Porta \cite{DELLAPORTA201927} to group the variant pairs based on their ability to simultaneously satisfy these conditions, see \Cref{tab:DP_CO,tab:DP_CM}. A detailed discussion on this is presented in \autoref{sec:extreme_comp}. Della Porta et al. have also put forth another idea of supercompatibility within the martensitic phase, known as \textit{triplet condition} \cite{DellaPorta2022Triplet}. Triplet conditions are conceptually similar to the cofactor conditions. They represent the conditions of degeneracy of the equations of compatibility between a martensitic laminate and another variant. In a real material, these conditions are easier to achieve than cofactor conditions and lead to the formation of martensitic triplets, which are three-fold martensitic microstructures, see \autoref{fig:triplets_27}. These triplets combine with the usual four-fold crossing twins \cite{Bhattacharya1997} to form intricate colonies of variants (see Figure 10, \cite{DellaPorta2022Triplet}) which is believed to suppress the plastic deformation and improve the mechanical properties of the martensitic phase.\\[5pt]
The primary goal of this paper is to extend the notion of zero elastic-energy phase interfaces to the laminates of compound domains, see \autoref{fig:venn_scope_b}. The ability to accommodate the laminates of compound domains within the zero-energy framework expands the set of deformation gradients which can make exact interfaces with austenite, without requiring any transition layer. More deformation gradients mean more ways in which the growing interface can accommodate defects and heterogeneities present inside the material. As a result, the reversibility of the material is expected to improve significantly. \\[5pt]
To achieve this goal, we adopt an inverse approach compared to that employed by Chen et al.\cite{CHEN20132566}: we first construct compatible \textit{triple clusters} involving a pair of compound domains and austenite, and then derive the sufficient conditions under which a laminate, formed by stacking these clusters, maintains compatibility across the phase interface for all volume fractions. Here, the term \textit{triple clusters} collectively refers to two types of configurations—triple junctions and parallel domain walls—that can lead to the elimination of the transition layer. Precise definitions are presented in \autoref{sec:commutation}. An important result of this paper is the discovery that compatible variants which commute necessarily form compound domains. This commutation property can be exploited to construct triple clusters between compound domains and austenite. Moreover, when two compatible variants commute, both solutions of the twinning equation for compound domains can independently form distinct triple clusters with austenite. This offers a distinct advantage over Type-I/II domains, where at most a single triple cluster can exist for each pair \cite{DELLAPORTA201927}. Thus, the commutation property not only reflects the underlying symmetry between martensitic variants but, under certain conditions, also enhances the compatibility across the two phases.\\[5pt]
 A notable observation in this context is that the non-conventional twins observed in the cubic to monoclinic II transformation also commute. This observation allows us to analyze non-conventional twin systems under the commutation framework. In particular, this framework enables us to (1) classify these twins under compound domains and (2) facilitate the formation of stress-free interfaces between their laminates and austenite. This way, the notion of zero-elastic energy interfaces can be extended to non-conventional twins.\\[5pt]
We also identify cases of compound domains that do not commute, see \autoref{rem:different_compound_tendancy}. For these non-commuting cases, we explicitly derive the general conditions under which any pair of compound domains forms triple clusters with austenite, see \autoref{thm:compound_triple}. In particular, for such clusters to occur, at least one of the \cref{eq:58,eq:59,eq:61,eq:63} must be satisfied. The results obtained previously by virtue of commutation, arise as a special cases under these conditions. This way, we complete the characterization of triple clusters for all twin systems. In doing so, we push the concept of stress-free interfaces to its full extent, where in all Type I/II as well as compound laminates are simultaneously compatible with austenite for all volume fractions without the need of any transition layer. For lack of a more precise term, we refer to this enhanced level of compatibility as \textit{extreme compatibility}. Finally, we report some novel microstructures which are possible under \textit{extreme compatibility conditions} in cubic to orthorhombic transformations. These include a flexible spearhead nucleus of martensite, austenite inclusions embedded in martensite, and a special type of four-fold martensitic microstructure. The presence of these nuclei is expected to unlock new transformation pathways in the material. \\[5 pt]
{\textit{Organization of the paper:}} In \autoref{sec:gnlm_CC}, we briefly summarize key results from the literature that are essential for understanding this work. In \autoref{sec:triple_beyond}, we investigate the commutation properties of martensitic variants and establish their connection to compound domains in \autoref{thm:commutation}. We then use this commutation property to construct triple clusters between compound domains and austenite in  \Cref{thm:firstkind,thm:secondkind}. The remaining cases of non-commuting compound domains are addressed in \autoref{thm:compound_triple}. In \autoref{sec:extreme_comp}, we formulate the extreme compatibility conditions for cubic to orthorhombic and cubic to monoclinic-II transformations. The novel microstructures enabled by the extreme compatibility in cubic to orthorhombic transformations are presented in \autoref{sec:new_micro}. Finally, we present our conclusions and suggest directions for future research in \autoref{sec:conclusion}.\\[5 pt]
{\textit{Notation:}} $\mathbb{R}^3$ is the three-dimensional real vector space. $\mathbb{R}^{3\times 3}$ is the set of all 3 $\times$ 3 real matrices. $\mathbb{R}^{3\times 3}_+$ is the subset of $\mathbb{R}^{3\times 3}$ with positive determinant. $\mathbb{R}^{3\times 3}_{+sym}$ is the set of all positive definite, symmetric, 3 $\times$ 3 matrices. $SO(3)$ is the special orthogonal group, i.e. the set of all 3 $\times$ 3 orthogonal matrices with unit determinant. $\mathbb{S}^2$ represents the 2-sphere. $\mathbb{C}^{n\times n}$ is the set of all square matrices of size $n$ with complex entries. $\hat{u}_i,\:\hat{e}_i$ and $\hat{\underline{e}}_i$ represent the eigenvectors, symmetry axes and the standard basis of $\mathbb{R}^3$ respectively. The symbol $cof\;U$ denotes the cofactor of the matrix U. $\mathscr{P}(a)$ and $\mathscr{P}(m)$ denote the point groups of austenite and martensite, respectively.

\section{Geometrically Non-Linear Theory of Martensite and Cofactor Conditions}
\label{sec:gnlm_CC}

In this section, we summarize the essential elements of the crystallographic theory of martensite. More detailed accounts can be found in \cite{bhattacharya2003microstructure},\cite{Ericksen1991},\cite{BALL200461},\cite{JAMES2000197}. We also import some useful results from Chen at al \cite{CHEN20132566}, which is the first formal study of the cofactor conditions and from Della Porta \cite{DELLAPORTA201927}, which further improved the understanding of supercompatibility.\\[4pt]
Let $\Omega \in \mathbb{R}^3 $ be a region occupied by a crystal in a reference configuration, when the crystal structure is austenite. As the crystal deforms upon phase transformation, a point $x \in \Omega$ is displaced to $y(x) \in \mathbb{R}^3$. The entire deformation of the crystal is given by a mapping $y: \Omega\rightarrow\mathbb{R}^3$.\\ We assume that the helmholtz free energy density of the material $\Phi$ depends on the gradient of the deformation, $\nabla y(x) \in \mathbb{R}^{3\times3}_{+}$ as well as the absolute temperature $\theta$. The total thermoelastic energy of the crystal is given by

\begin{equation}
\label{eq:1}
    \int_{\Omega} \Phi(\nabla y(x), \theta) \,dx.
\end{equation}

The free energy density $\Phi$ has two important invariance characteristics. First is that of the principle of frame indifference, i.e. for any rigid rotation $ Q \in SO(3)$, we have 

\begin{equation}
\label{eq:2}
\Phi(Q\nabla y(x), \theta) = \Phi(\nabla y(x), \theta).
\end{equation}

The second invariance comes from the symmetry of the crystalline material. We restrict ourselves to rotational symmetries only i.e. the symmetries which belong to the Ericksen-Pitteri neighborhood \cite{pitteri_zanzotto}. If $ R \in SO(3)$ is a symmetry operation for crystal, then

\begin{equation}
\label{eq:3}
\Phi(\nabla y(x) R, \theta) = \Phi(\nabla y(x), \theta).
\end{equation}

We can combine \eqref{eq:2} and \eqref{eq:3} to write

\begin{equation}
\label{eq:4}
\Phi(R^T\nabla y(x) R, \theta) = \Phi(\nabla y(x), \theta).
\end{equation}

The free energy density $\Phi$ defined on $\Omega$ can be related to the lattice level deformations using the Cauchy-Born rule \cite{ericksen1984cauchy}. Therefore, we can understand $\nabla y(x)$ as the gradient of the deformation which maps the lattice of austenite to that of the martensite. The symmetry of austenite and martensite lattices then requires that there is a set of transformation stretch matrices $U_1, U_2, ...., U_N$ $\in \mathbb{R}^{3\times3}_{+sym}$ that are related by symmetry $U_i = R_i^TU_1R_i,\;i=1,2,....N,$ where $R_i \in \mathscr{P}(a)$, the point group of austenite. $U_1, U_2, ...., U_N$ are called the variants of martensite and $N$ is given by the ratio of orders of the point groups of austenite and martensite $\mathscr{P}(a)/\mathscr{P}(m)$. We assume that the dependence of $U_i$ on temperature is small enough to neglected. $U_i$ are then the energy wells of the martensite i.e. if $\theta_c$ is the transformation temperature, we have

\begin{equation}
\label{eq:5}
\Phi(U_1, \theta) = \cdots = \Phi(U_N, \theta) \;\leq\; \Phi(F, \theta),\quad \theta \leq \theta_c\:,
\end{equation}

F is any element of $\mathbb{R}^{3\times3}_{+sym}$. At the transformation temperature, $\theta = \theta_c$, austenite and martensite equally minimize the energy i.e. $I$ is also a minimizer of free energy. For temperatures above the transformation temperature, only austenite minimizes the energy.

\begin{equation}
\label{eq:6}
\Phi(U_1, \theta) = \cdots = \Phi(U_N, \theta)=\Phi(I,\theta) \;\leq\; \Phi(F, \theta),\quad \theta = \theta_c\:,
\end{equation}

\begin{equation}
\label{eq:7}
\Phi(I, \theta)\;\leq\; \Phi(F, \theta),\quad \theta \geq \theta_c.
\end{equation}

Using the principle of frame indifference, we can construct the full set of minimizers of free energy $\Phi$ at the transformation temperature $\theta_c$ as

\begin{equation}
\label{eq:8}
SO(3)U_1\cup SO(3)U_2 \cdots \cup SO(3)U_N \cup SO(3)I.
\end{equation}

\subsection{Kinematic Compatibility}
\label{sec:kin_comp}

The general equation of compatibility between two distinct variants of martensite $U$ and $V$ is written as 
\begin{equation}
\label{eq:9}
    QV - U = a \otimes \hat{n}.
\end{equation}

In \eqref{eq:9}, $U,V \in \mathbb{R}^{3\times3}_{+\scriptstyle{sym}} $ are known and the equation needs to be solved for  $Q \in SO(3) $, $a \in \mathbb{R}^3$ and $\hat{n} \in S^2 $ that satisfy \eqref{eq:9}.
Broadly speaking, $U$ and $V$ as above, can be classified into generic and non-generic twins, see \autoref{fig:venn}. Generic twins are further classified into conventional and non-conventional twins. Conventional twins can be either Type I/II twins or compound twins. Type I/II twins are the ones that satisfy Mallard's law, i.e. they are related by a single $180^\circ$ rotation in $\mathscr{P}(a)$. Compound twins are related by two distinct $180^\circ$ rotations in $\mathscr{P}(a)$. The term \textit{domains} is used in a general sense, when the $180^\circ$ rotation(s), relating the two variants, do(es) not necessarily belong to $\mathscr{P}(a)$. All conventional twins are domains, but domains may be conventional twins or non-conventional but generic twins. Type I/II domains and compound domains are mutually exclusive sets and collectively exhaust all cases of generic twins. Non-generic twins exist only for special values of transformation strain and perish under arbitrary perturbations of the lattice parameters. All non-generic twins are also non-conventional by definition. The above definitions are adopted from \cite{CHEN20132566}. We avoid the discussion on non-generic twins altogether in this paper and treat them as incompatible. Henceforth, the term non-conventional twins will refer to non-conventional but generic twins. We refer the reader to \cite{CHEN20132566},\cite{pitteri_zanzotto} for more details on this classification.

\begin{theorem}
    \label{thm: main_solution}
    (Ball and James \cite{Ball1987}, Bhattacharya \cite{bhattacharya2003microstructure}) Given matrices $F$ and $G$ with positive determinants, such  that $C=G^{-T}F^{T}FG^{-1} \neq I,$ and $\lambda_1 \leq1, \;\lambda_2=1,\;\lambda_3\geq1$ holds, then there are exactly two solutions to the equation $QF-G = a \otimes \hat{n}$, given by 

\begin{equation}
    a = \rho\left(\sqrt{\frac{\lambda_3(1-\lambda_1)}{\lambda_3 - \lambda_1}}\hat{u}_1\; +\; \kappa \sqrt{\frac{\lambda_1(\lambda_3 -1)}{\lambda_3 - \lambda_1}}\hat{u}_3 \right),
    \nonumber
\end{equation}

\begin{equation}
\label{eq:10}
    \hat{n} = \frac{\sqrt{\lambda_3} - \sqrt{\lambda_1}}{\rho\sqrt{\lambda_3 - \lambda_1}}\left(-\sqrt{1-\lambda_1}\:G^T\hat{u}_1\; +\; \kappa \sqrt{\lambda_3 -1}\:G^T\hat{u}_3 \right).
\end{equation}

where $\kappa = \pm1$, $\rho\neq0$ is chosen to make $\left|{\hat{n}}\right| =1$  and $\hat{u}_i$ are the eigenvectors of $C$ corresponding to the eigenvalues $\lambda_i$.
\end{theorem} 

\begin{corollary}
\label{corollary:identity_solutions}
Let $F = F^T = U$ be a variant of martensite and $G=I$ in \autoref{thm: main_solution}. If the eigenvalues of $U$ satisfy $0<\lambda_1 \leq1, \;\lambda_2=1,\;\lambda_3\geq1$, then the solutions to $QU-I = a\otimes \hat{n}$ are given by  

\begin{equation}
    a = \frac{\rho}{\sqrt{\lambda_3^2 - \lambda_1^2}}\left(\lambda_3\sqrt{1-\lambda_1^2}\;\hat{u}_1\; +\; \kappa\; \lambda_1\sqrt{\lambda_3^2 -1}\;\hat{u}_3 \right)
    \nonumber,
\end{equation}

\begin{equation}
\label{eq:11}
    \hat{n} = \frac{\lambda_3 - \lambda_1}{\rho\sqrt{\lambda_3^2 - \lambda_1^2}}\left(-\sqrt{1-\lambda_1^2}\;\hat{u}_1\; +\; \kappa \sqrt{\lambda_3^2 -1}\;\hat{u}_3 \right).
\end{equation}

where $\hat{u}_i$ are the normalized eigenvectors of $U$ corresponding to the eigenvalues $\lambda_i$.\eqref{eq:11} gives the solutions of two possible interfaces between a single variant of martensite and austenite.
\end{corollary} 

\begin{theorem}
\label{thm:Mallards}
    (Gurtin \cite{Gurtin1983}, Bhattacharya \cite{bhattacharya2003microstructure}) Let $R$ be a $180^{\circ}$ rotation about some axis $\hat{e}$ and the matrices $F$ and $G$ as above satisfy
\begin{enumerate}
    \item $F = QGR$ for some rotation $Q \in SO(3),$
    \item $F^T F \neq G^TG.$
\end{enumerate}
Then there are two solutions to the equation $QF-G=a\otimes \hat{n}$, given by

\begin{equation}
    a_I = 2\left(\frac{G^{-T}\hat{e}}{\left|G^{-T}\hat{e}\right|^2} - G\hat{e}\right), \quad\quad \hat{n}_I = \hat{e}.
    \nonumber
\end{equation}

\begin{equation}
\label{eq:12}
    a_{II} = \rho G\hat{e}, \quad\quad  \hat{n}_{II}=\frac{2}{\rho}\left(\hat{e} - \frac{G^TG\hat{e}}{\left|G\hat{e}\right|^2}\right).
\end{equation}
\end{theorem} 

\begin{corollary}
\label{corollary:type12_solutions}
If $U$ and $V$ in \eqref{eq:9} satisfy $V = (-I + 2\;\hat{e}\otimes \hat{e})U(- I + 2\;\hat{e}\otimes \hat{e})$, then by \autoref{thm:Mallards}, we have

\begin{equation}
     a_I = 2\left(\frac{U^{-1}\hat{e}}{\left|U^{-1}\hat{e}\right|^2} - U\hat{e}\right), \quad\quad \hat{n}_I = \hat{e}.
    \nonumber
\end{equation}

\begin{equation}
\label{eq:13}
    a_{II} = \rho U\hat{e}, \quad\quad  \hat{n}_{II}=\frac{2}{\rho}\left(\hat{e} - \frac{U^2\hat{e}}{\left|U\hat{e}\right|^2}\right).
\end{equation}

where $\rho\neq0$ is chosen to make $\left|{\hat{n}_{II}}\right| =1$. In this case $(a_I,\hat{n}_1)$ generate a Type I solution and $(a_{II},\hat{n}_{II})$ generate a Type II solution.
\end{corollary}

\begin{theorem}
\label{thm:compound_domains}
    (Chen et al \cite{CHEN20132566}, Proposition 1)  Let $U \in \mathbb{R}^{3\times3}_{\scriptstyle{+sym}}$ be a variant of martensite and define $V = (-I + 2\;\hat{e}_1\otimes \hat{e}_1)U(- I + 2\;\hat{e}_1\otimes \hat{e}_1)$ such that $V \neq U$. There is a second unit vector $\hat{e}_2 \nparallel \hat{e}_1$, satisfying $V = (-I + 2\;\hat{e}_2\otimes \hat{e}_2)U(- I + 2\;\hat{e}_2\otimes \hat{e}_2)$ if and only if $\hat{e}_1$ is perpendicular to an eigenvector of $U$. Furthermore, $\hat{e}_1,\hat{e}_2$ and that eigenvector(normalized) form an orthonormal basis of $\mathbb{R}^3$. In this case $U$ and $V$ are said to from compound \textit{domains} and the solutions to \eqref{eq:9} are given by
\begin{equation}
    a^c_I = \zeta\: U\hat{e}_2,    \quad\quad  \hat{n}^c_1 = \hat{e}_1, \quad\quad \text{where}\;\; \zeta = 2\frac{\hat{e}_2.{U_1}^{-2}\hat{e}_1}{\hat{e}_1.{U_1}^{-2}\hat{e}_1} \:,
    \nonumber
\end{equation}
\begin{equation}
\label{eq:14}
    a^c_{II} = \eta\: U\hat{e}_1,    \quad\quad  \hat{n}^c_1 = \hat{e}_2, \quad\quad \text{where}\;\; \eta= -2\frac{\hat{e}_2.{U_1}^{2}\hat{e}_1}{\hat{e}_1.{U_1}^{2}\hat{e}_1}.
\end{equation}
\end{theorem}

\subsection{Crystallographic theory of martensite}
\label{sec:ptmc}
The crystallographic theory of martensite explains the formation of the interface between austenite and a finely twinned laminate of martensite. The twinned laminate is understood as a sequence of deformations $y^{(k)}, k = 1,2,3 ,....., $ of the two twin variants, such that
\begin{equation}
\label{eq:15}
    \int_{\Omega}\Phi(\nabla y^{(k)}(x),\theta)dx \to 0 \;\;\text{as}\;\; k\to\infty,
\end{equation}
where $k$ is related to the inverse of the width of the interpolation layer between austenite and martensite. Ball and James \cite{Ball1987} show that as $k\to\infty$, the sequence $y^{(k)}(x)$ converges to a deformation $y(x)$ such that
\begin{equation}
\label{eq:16}
    \nabla y = \mu (QV) + (1-\mu)U = \mu (U + a\otimes \hat{n}) + (1-\mu)U.
\end{equation}

Then assuming \eqref{eq:9} holds, the equation of the crystallographic theory of martensite is written as
\begin{equation}
\label{eq:17}
    \tilde{Q}[\mu (U + a\otimes \hat{n}) + (1-\mu)U] - I = b\otimes \hat{m},
\end{equation}

which is to be solved for $0\leq\mu\leq1$, the volume fraction of the laminate, $\tilde{Q}\in SO(3)$ and vectors $b$ and $\hat{m}$.

\subsection{Cofactor conditions}
\label{sec:CC}
The cofactor conditions are the necessary and sufficient conditions that the crystallographic theory of martensite \eqref{eq:17} has a solution for all volume fractions, $\mu\in[0,1]$. The cofactor conditions are given as

\begin{equation}
\label{eq:CC1}
    \lambda_2 =1,\quad \textit{where $\lambda_2$ is the middle eigenvalue of $U,$}
    \tag{CC1}
\end{equation}
\begin{equation}
\label{eq:CC2}
    a.\;U\:cof\:(U^2 - I)\hat{n}=0,
    \tag{CC2}
\end{equation}
\begin{equation}
\label{eq:CC3}
    tr(U^2) - det(U^2) - \frac{\left|a\right|^2}{4} - 2 \geq 0.
    \tag{CC3}
\end{equation}

The satisfaction of cofactor conditions gives rise to special microstructures which can be summarized by the following theorems.

\begin{theorem}
\label{thm:type1triple}
    (Chen et al \cite{CHEN20132566}, Theorem 7) Suppose that the cofactor conditions are satisfied for Type I domains, then there exist rotations $R^{I}_u$ and $R^{I}_v$ in $SO(3)$ and vectors $(b_u^{I},b^{I}_v,a^{I}) \in \mathbb{R}^3$ and vectors $(\hat{m}_u^{I},\hat{m}_v^{I},\hat{n}^{I})$ in $\mathbb{S}^2$ such that the following set of equations hold.
\begin{equation}
    R_u^{I}U - I = b^{I}_u \otimes \hat{m}^{I}_u,
    \nonumber
\end{equation}
\begin{equation}
    R_v^{I}V - I = b_v^{I} \otimes \hat{m}_v^{I},
    \nonumber
\end{equation}
\begin{equation}
    R_u^{I}U - R_v^{I}V = a^{I} \otimes \hat{n}^{I},
    \nonumber
\end{equation}

where $b^{I}_u\;||\;b^{I}_v\;||\;a^{I}$ and $(\hat{m}_u^{I}, \hat{m}_v^{I} $ and $\hat{n}^{I})$ lie in a plane \autoref{fig:triple_1_2_a}.
\end{theorem} 

\begin{theorem}
\label{thm:type2triple}
    (Chen et al \cite{CHEN20132566}, Theorem 8) Suppose that the cofactor conditions are satisfied for Type II domains, then there exist rotations $R_u^{II}$ and $R_v^{II}$ in $SO(3)$ and vectors $(b_u^{II},b_v^{II},a^{II})\in$ $\mathbb{R}^3$ and vector $\hat{m}^{II}$ $\in$ $\mathbb{S}^2$ such that the following set of equations hold.

\begin{equation}
    R_u^{II}U - I = b_u^{II} \otimes \hat{m}^{II},
    \nonumber
\end{equation}
\begin{equation}
    R_v^{II}V - I = b_v^{II} \otimes \hat{m}^{II},
    \nonumber
\end{equation}
\begin{equation}
    R_u^{II}U - R_v^{II}V = a^{II} \otimes \hat{m}^{II},
    \nonumber
\end{equation}

where $(b_u^{II},\;b_v^{II}$ and $a^{II})$ lie in a plane, \autoref{fig:triple_1_2_b}.
\end{theorem}

\begin{figure}
    \centering

    \begin{minipage}[b]{0.5\textwidth}
        \centering
        \includegraphics[scale=0.75]{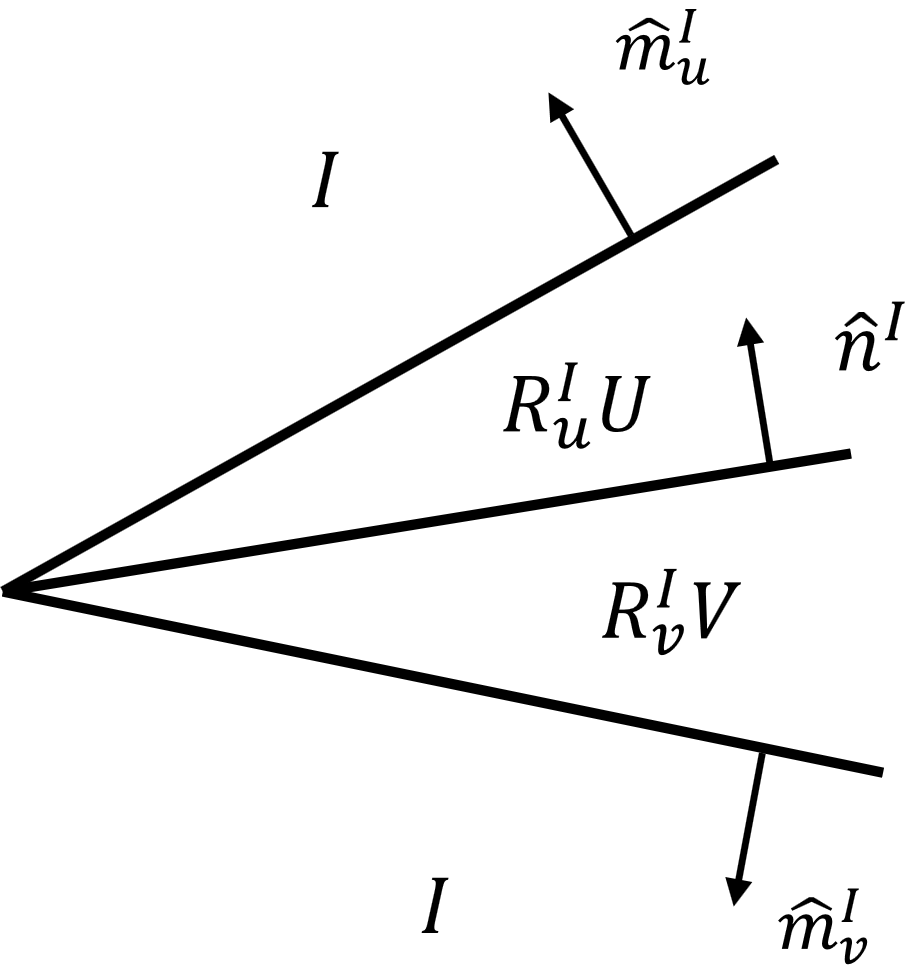}
        \subcaption{}
        \label{fig:triple_1_2_a}
    \end{minipage}%
    \begin{minipage}[b]{0.5\textwidth}
        \centering
        \includegraphics[scale=0.75]{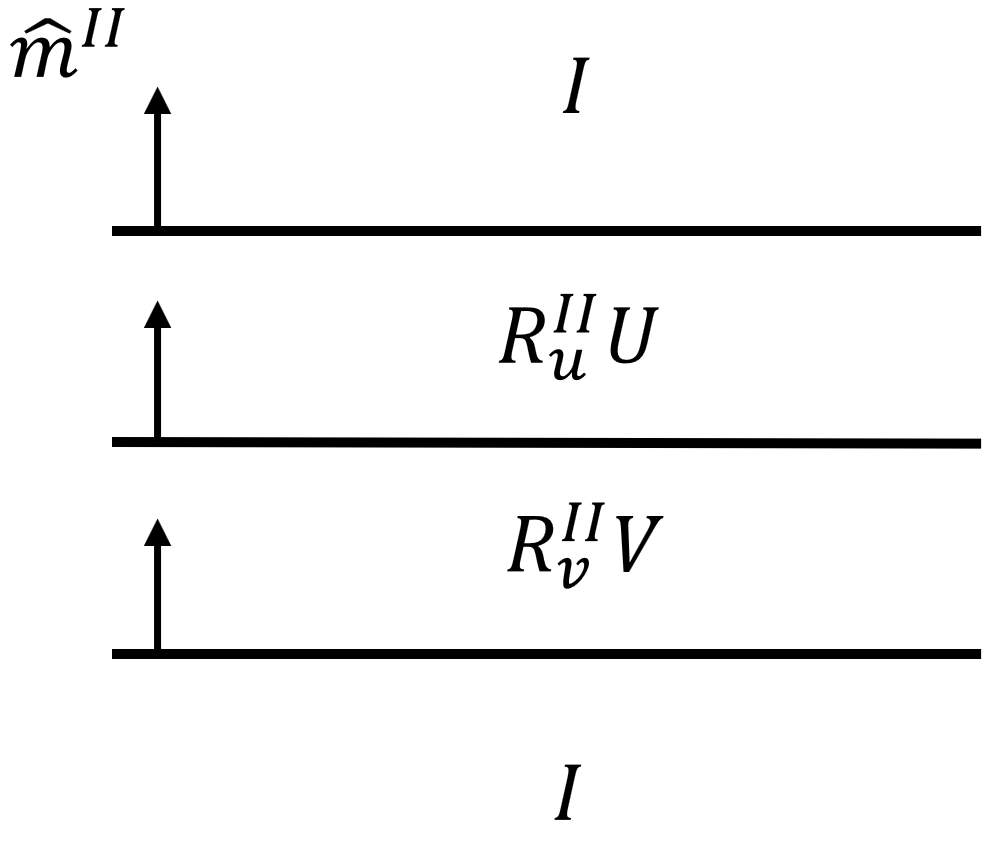}
        \subcaption{}
        \label{fig:triple_1_2_b}
    \end{minipage}

    \caption{(a) \textit{Type-I triple cluster}: A triple junction configuration composed of Type-I domains and austenite as implied by \autoref{thm:type1triple} (b)  \textit{Type-II triple cluster}: A parallel domain wall configuration composed by Type-II domains and austenite as implied by \autoref{thm:type2triple}.}
    \label{fig:triple_1_2}
\end{figure}

\autoref{thm:type1triple} and \autoref{thm:type2triple} allow us to construct stress-free interfaces between a laminate of Type I/II domains and austenite by eliminating the transition layer at the phase interface. However, we do not have an equivalent result for the case of a laminate of compound domains.\\[5 pt]
Finally, we state the following results from Della Porta, which allow us to find all the twin systems that satisfy the cofactor conditions, starting from a pair of variants, as well as, imposes a restriction on the satisfaction of cofactor conditions for Type-I and Type-II twins simultaneously, for that pair.

\begin{theorem}
\label{thm:dellaportaQ}
     (Della Porta \cite{DELLAPORTA201927}, Proposition 3.1) Let $(Q,a,\hat{n})$ be a solution of \eqref{eq:9}. Suppose $(Q,a,\hat{n})$ satisfies the cofactor conditions. Then for every $R\in SO(3)$, the matrices $RUR^T$, $RVR^T \in \mathbb{R}^{3\times 3}_{+sym}$ are such that $RQR^T \in SO(3)$, $Ra \in \mathbb{R}^3$ and $R\hat{n}\in \mathbb{S}^2$ is the same type of solution (Type I/II or compound) of \eqref{eq:9} for stretch tensors $RUR^T$ and $RVR^T$, as is $(Q,a,\hat{n})$. Furthermore, $(RQR^T, Ra,R\hat{n})$ also satisfies the cofactor conditions.
\end{theorem}

\begin{theorem}
\label{thm:dellaportaR}
    (Della Porta \cite{DELLAPORTA201927}, Proposition 3.2) Let $U$ and $V$ as in \eqref{eq:9} be such that $U$ and $V$ form Type-I/II twins. Then the Type-I and Type-II solutions of the twinning equation \eqref{eq:9} cannot satisfy \eqref{eq:CC1} and \eqref{eq:CC2} at the same time.
\end{theorem} 
With these results in place, we now proceed to investigate the possibility of triple clusters beyond Type I/II domains in the following section.

\section{Triple Clusters Beyond Type I/II Domains}
\label{sec:triple_beyond}
In this section, we extend the idea of triple clusters with austenite, as implied in \autoref{thm:type1triple} and Theorem \autoref{thm:type2triple} for Type-I and Type-II domains, respectively, to the case of compound domains, see \autoref{fig:venn_scope}. These compound domains may be either (conventional) compound twins or non-conventional twins. We first study the commutation properties of the stretch tensors of the martensitic variants. This reveals an interesting connection between the commutation of variants and their symmetry. One of our main findings is that when the stretch tensors of two compatible martensitic variants commute, they always form compound domains. However, the converse is not always true. A pair of compound domains may or may not commute. This commutation property can be exploited for two key applications. \\[5pt]
Firstly, the observation that the non-conventional twins in cubic to monoclinic-II transformations commute, this allows us to view them as compound domains, see \autoref{fig:venn_structure}. Then, the solutions obtained by Chen et al. \cite{CHEN20132566} for compound domains, can be easily extended to the case of non-conventional twins, using \autoref{thm:compound_domains}. This allows us to represent the twinning elements of non-conventional twins using relatively simple formulae, without the need to resort to the general solution given by \eqref{eq:10}. Pitteri and Zanzotto have discussed this issue in their work (section 4.8, \cite{PITTERI1998225}).\\[5pt]
Secondly, under certain assumptions, we can utilize the commutation property, to form planar triple clusters between a pair of commuting variants and austenite. These triple clusters enable the elimination of transition layer between austenite and some compound domains (compound twins or non-conventional twins, whose stretch tensors commute). This addresses the constraint of the cofactor conditions, which are insufficient to eliminate the transition layer for compound domains, see \autoref{fig:venn_scope_b}. However, since compound domains do not always commute, there are instances where this strategy does not work. This is the case for (conventional) compound twins in cubic to monoclinic (I \& II) transformations. We explicitly derive formulae for such cases in \autoref{thm:compound_triple}. For improved readability and interpretation, the results of \autoref{sec:triple_beyond} are summarized in \Cref{fig:3ortho,fig:3mono}. We encourage the reader to refer to this figure while reading \autoref{sec:triple_compound}.
\\[5pt] Our approach covers all the modes in which the transition layer can be eliminated between laminates of generic twins and austenite. It naturally leads to the question, `How many of these modes are simultaneously possible in a given transformation, and under what conditions?'. We answer this question in \autoref{sec:extreme_comp}.
\subsection{Commutation properties of stretch tensors}
\label{sec:commutation}
We begin by stating a definition and a lemma, which will be used later to establish some results.

\begin{definition}
\label{def:commutation}
    Two matrices $A$, $B$ $\in \mathbb{C}^{n\times n}$ are said to be simultaneously diagonalizable if there exists a single non-singular $S$ $\in \mathbb{C}^{n\times n}$ such that $S^{-1}AS$ and $S^{-1}BS$ are both diagonal. A family $\mathscr{F} \subset \mathbb{C}^{n\times n}$  is said to be simultaneously diagonalizable if there exists a single non-singular $S$ $\in \mathbb{C}^{n\times n}$ such that $S^{-1}AS$ is diagonal $\forall A\in \mathscr{F}$.
\end{definition} 

\begin{figure}[hbt!]
    \centering
    \includegraphics[scale=0.50]{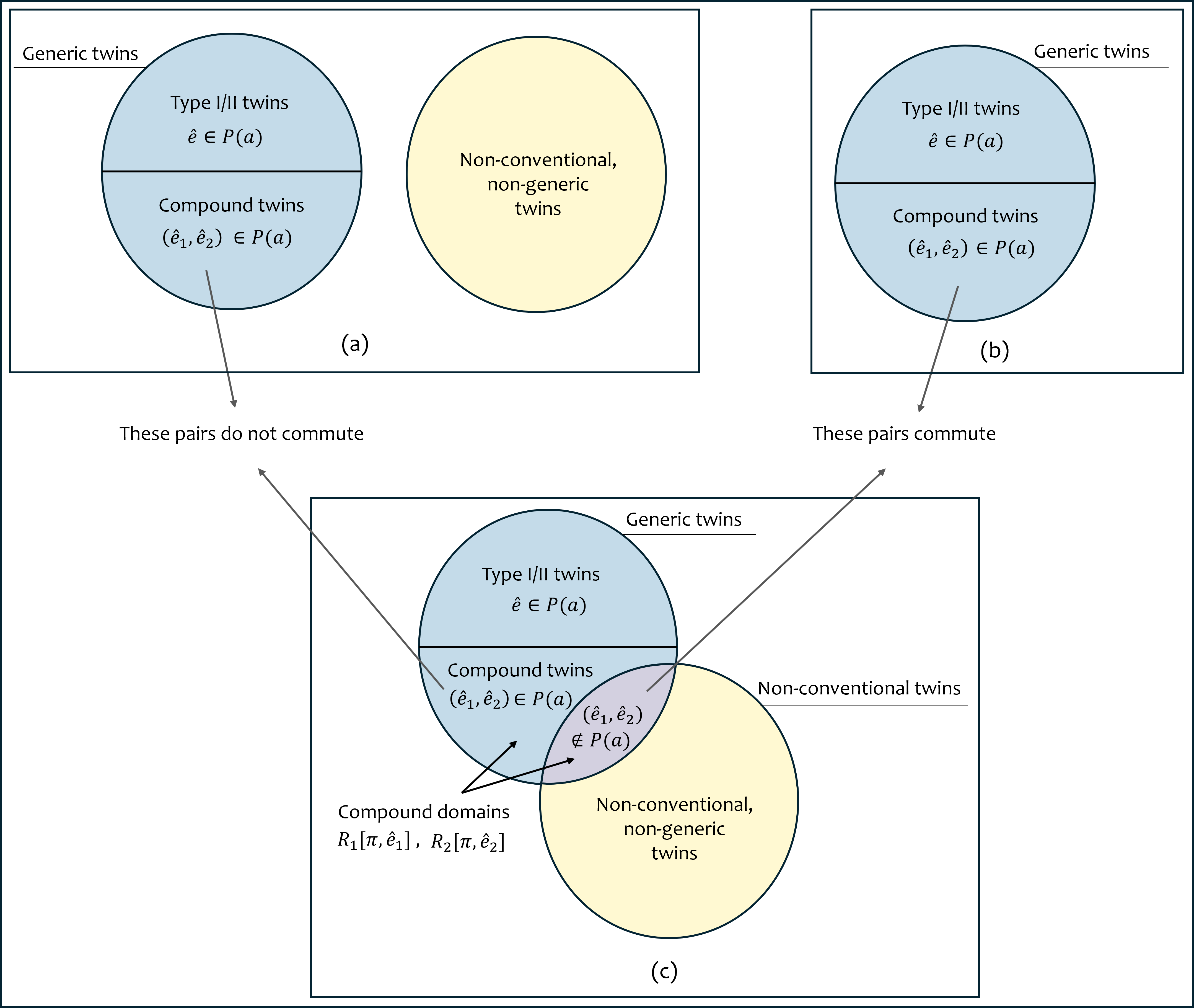}
    \caption{Structure of twinning in (a) Cubic to monoclinic-I transformation (b) Cubic to orthorhombic transformation, $\langle110\rangle$\textsubscript{cubic} type (c) Cubic to monoclinic-II transformation.}
    \label{fig:venn_structure}
\end{figure}

\begin{lemma}
\label{lem:commutation}
    Let $A$, $B$ $\in \mathbb{C}^{n\times n}$ be diagonalizable. Then $A$ and $B$ are simultaneously diagonalizable if and only if they commute. A family of diagonalizable matrices, $\mathscr{F}\subset \mathbb{C}^{n\times n}$  is a simultaneously diagonalizable family if and only if it is a commuting family \cite{Horn_Johnson_1985}.
\end{lemma} 

With these ideas in mind, let us examine, case by case, the structure of stretch tensors for various transformations.

\subsection*{Cubic to tetragonal transformation}
For the case of cubic to tetragonal transformations, we have three variants in the martensitic phase, which have the following form
\begin{gather}
\label{eq:18}
U_1 = 
    \begin{bmatrix}
        a & 0& 0 \\ 0 & a & 0 \\ 0 & 0 & b
    \end{bmatrix},\qquad
    U_2 = 
    \begin{bmatrix}
        a & 0& 0 \\ 0 & b & 0 \\ 0 & 0 & a
    \end{bmatrix},\qquad
    U_3 = 
    \begin{bmatrix}
        b & 0& 0 \\ 0 & a & 0 \\ 0 & 0 & a
    \end{bmatrix}.
\end{gather}

Since the matrices are diagonal, their product is entry-wise, which is commutative. Matrices $U_1$, $U_2$, and $U_3$ satisfy the premise of \autoref{lem:commutation} and form a family of commuting and hence simultaneously diagonalizable matrices. All the pairs of variants $(U_1$, $U_2)$, $(U_1$, $U_3)$ and $(U_2$, $U_3)$ form compound twins.

\subsection*{Cubic to orthorhombic transformation}
In cubic to orthorhombic transformations, we have six variants in the martensitic phase. Depending on the lattice correspondence, two different transformation pathways are possible. In the first type of transformation, two of the three axes of orthorhombic symmetry are obtained from $\langle$110$\rangle$\textsubscript{cubic} axes, then the stretch tensors are of the form

\begin{gather}
    U_1 = 
    \begin{bmatrix}
        d & 0 & 0 \\ 0 & a & b \\ 0 & b & a
    \end{bmatrix},\qquad
    U_2 = 
    \begin{bmatrix}
        d & 0 & 0 \\ 0 & a & -b \\ 0 & -b & a
    \end{bmatrix},\qquad
    U_3 = 
    \begin{bmatrix}
        a & 0& -b \\ 0 & d & 0 \\ -b & 0 & a
    \end{bmatrix},\nonumber
    \\\nonumber\\
    U_4 = 
    \begin{bmatrix}
        a & 0& b \\ 0 & d & 0 \\ b & 0 & a
    \end{bmatrix},\qquad
    U_5 = 
    \begin{bmatrix}
        a & -b& 0 \\ -b & a & 0 \\ 0 & 0 & d
    \end{bmatrix},\qquad
    U_6 = 
    \begin{bmatrix}
        a & b& 0 \\ b & a & 0 \\ 0 & 0 & d
    \end{bmatrix}.\label{eq:19}
\end{gather}

It can be verified that $U_1\:U_2 = U_2\:U_1$, $U_3\:U_4 = U_4\:U_3$ and $U_5\:U_6 = U_6\:U_5$. The pairs $(U_1$,$U_2)$, $(U_3$,$U_4)$, and $(U_5$,$U_6)$ commute and form compound twins. All other pairs form Type-I/II twins, and do not commute. Unless otherwise stated, all mentions of cubic to orthorhombic transformations in this paper refer to this pathway of transformation.

The second type of cubic to orthorhombic transformation is the one in which the axes of orthorhombic symmetry are along the $\langle$100$\rangle$\textsubscript{cubic} axes and the stretch tensors are of the form

\begin{gather}
\label{eq:20}
    U_1 = 
    \begin{bmatrix}
        a & 0 & 0 \\ 0 & b & 0 \\ 0 & 0 & c
    \end{bmatrix},\qquad
    U_2 = 
    \begin{bmatrix}
        b & 0 & 0 \\ 0 & a & 0 \\ 0 & 0 & c
    \end{bmatrix},\qquad
    U_3 = 
    \begin{bmatrix}
        a & 0 & 0 \\ 0 & c & 0 \\ 0 & 0 & b
    \end{bmatrix},\nonumber
    \\\nonumber\\
    U_4 = 
    \begin{bmatrix}
        c & 0 & 0 \\ 0 & b & 0 \\ 0 & 0 & a
    \end{bmatrix},\qquad
    U_5 = 
    \begin{bmatrix}
        b & 0 & 0 \\ 0 & c & 0 \\ 0 & 0 & a
    \end{bmatrix},\qquad
    U_6 = 
    \begin{bmatrix}
        c & 0 & 0 \\ 0 & a & 0 \\ 0 & 0 & b
    \end{bmatrix}.
\end{gather}

In this case, all variants commute with each other. Each variant forms compound twins with exactly three variants and is incompatible with the remaining two variants, e.g. $U_1$ forms compound twins with $U_2$, $U_3$, $U_4$ and is incompatible with $U_5$ and $U_6$. This shows that if two variants commute, they need not always be compatible.\\[5pt]
In both types of cubic to orthorhombic transformations, we see that the stretch tensors of compound twins commute. 

\subsection*{Cubic to monoclinic transformation}
For this transformation, we have twelve variants in the martensitic phase. Similar to the previous case, two different transformation pathways are possible here. In the first type of transformation, the axis of the monoclinic symmetry corresponds to $\langle110\rangle$\textsubscript{cubic} direction. This transformation is also called cubic to monoclinic-I transformation. The stretch tensors are of the form

\begin{gather}
    U_1 = 
    \begin{bmatrix}
        d & -c & -c \\ -c & a & b \\ -c & b & a
    \end{bmatrix},\quad
    U_2 = 
    \begin{bmatrix}
        d & -c & c \\ -c & a & -b \\ c & -b & a
    \end{bmatrix},\quad
    U_3 = 
    \begin{bmatrix}
        d & c & -c \\ c & a & -b \\ -c & -b & a
    \end{bmatrix},\quad
    U_4 = 
    \begin{bmatrix}
        a & c& -b \\ c & d & -c \\ -b & -c & a
    \end{bmatrix},\nonumber
    \\\nonumber\\
    U_5 = 
    \begin{bmatrix}
        a & c& b \\ c & d & c \\ b & c & a
    \end{bmatrix},\quad
    U_6 = 
    \begin{bmatrix}
        d & c& c \\ c & a & b \\ c & b & a
    \end{bmatrix},\quad
    U_7 = 
    \begin{bmatrix}
        a & -b& c \\ -b & a & -c \\ c & -c & d
    \end{bmatrix},\quad
    U_8 = 
    \begin{bmatrix}
        a & b& c \\ b & a & c \\ c & c & d
    \end{bmatrix}\label{eq:21}, 
    \\\nonumber\\
    U_9 = 
    \begin{bmatrix}
        a & -c& b \\ -c & d &-c \\ b & -c & a
    \end{bmatrix},\quad
    U_{10} = 
    \begin{bmatrix}
        a & -c& -b \\ -c & d & c \\ -b & c & a
    \end{bmatrix},\quad
    U_{11} = 
    \begin{bmatrix}
        a & -b& -c \\ -b & a & c \\ -c & c & d
    \end{bmatrix},\quad
    U_{12} = 
    \begin{bmatrix}
        a & b& -c \\ b & a & -c \\ -c & -c & d
    \end{bmatrix}.\nonumber
\end{gather}%
No pair of matrices in \eqref{eq:21} satisfies the commutation property. There are no non-conventional, generic twins in this type of transformation. All twins are either Type I/II or compound. This case (and the next one) is an example of a scenario in which compound twins do not commute.\\[5pt]

For the second kind of cubic to monoclinic transformations, the axis of the monoclinic symmetry corresponds to $\langle100\rangle$\textsubscript{cubic} direction. This transformation is also called cubic to monoclinic-II transformation. The stretch tensors are of the form

\begin{gather}
    U_1 = 
    \begin{bmatrix}
        d & 0 & 0 \\ 0 & a & b \\ 0 & b & c
    \end{bmatrix},\quad
    U_2 = 
    \begin{bmatrix}
        d & 0 & 0 \\ 0 & a & -b \\ 0 & -b & c
    \end{bmatrix},\quad
    U_3 = 
    \begin{bmatrix}
        a & 0& -b \\ 0 & d & 0 \\ -b & 0 & c
    \end{bmatrix},\quad
    U_4 = 
    \begin{bmatrix}
        a & 0& b \\ 0 & d & 0 \\ b & 0 & c
    \end{bmatrix},\nonumber
    \\\nonumber\\
    U_5 = 
    \begin{bmatrix}
        d & 0& 0 \\ 0 & c & b \\ 0 & b & a
    \end{bmatrix},\quad
    U_6 = 
    \begin{bmatrix}
        c & -b& 0 \\ -b & a & 0 \\ 0 & 0 & d
    \end{bmatrix},\quad
    U_7 = 
    \begin{bmatrix}
        c & b& 0 \\ b & a & 0 \\ 0 & 0 & d
    \end{bmatrix},\quad
    U_8 = 
    \begin{bmatrix}
        a & b& 0 \\ b & c & 0 \\ 0 & 0 & d
    \end{bmatrix}\label{eq:22},
    \\\nonumber\\
    U_9 = 
    \begin{bmatrix}
        c & 0& b \\ 0 & d & 0 \\ b & 0 & a
    \end{bmatrix},\quad
    U_{10} = 
    \begin{bmatrix}
        c & 0& -b \\ 0 & d & 0 \\ -b & 0 & a
    \end{bmatrix},\quad
    U_{11} = 
    \begin{bmatrix}
        a & -b& 0 \\ -b & c & 0 \\ 0 & 0 & d
    \end{bmatrix},\quad
    U_{12} = 
    \begin{bmatrix}
        d & 0& 0 \\ 0 & c & -b \\ 0 & -b & a
    \end{bmatrix}.\nonumber
\end{gather}\\[10pt]
 In this case, the pairs $(U_1,U_{12})$, $(U_2,U_{5})$, $(U_3,U_{9})$, $(U_4,U_{10})$, $(U_6,U_{8})$ and $(U_7,U_{11})$ are the ones which commute. These pairs form non-conventional compound domains and constitute column 5 of \autoref{tab:CM}. There are 24 pairs of Type I/II twins, listed in columns 1 and 2 of \autoref{tab:CM}, and 12 pairs of compound twins, listed in columns 3 and 4 of \autoref{tab:CM}, none of which commute.\\[5pt]
Therefore, it can be concluded that in case of cubic austenite transforming to monoclinic martensite, commuting stretch tensors arise only in cubic-to-monoclinic-II transformations, where they form non-conventional twins.

\subsection*{Tetragonal to orthorhombic transformation}
In cubic to orthorhombic transformations, we have just two variants in the martensitic phase. Again, two different transformation pathways are possible in this transformation. If the first type of transformation, the fourfold $c$ axis and the basal face diagonals of the tetragonal unit cell transform to the primitive unit cell of the orthorhombic lattice \cite{intermediatetwinningspontaneous}. The stretch tensors are of the form,

\begin{gather}
\label{eq:23}
    U_1 = 
    \begin{bmatrix}
        a & b & 0 \\ b & a & 0 \\ 0 & 0 & c
    \end{bmatrix},\qquad
    U_2 = 
    \begin{bmatrix}
        a & -b & 0 \\ -b & a & 0 \\ 0 & 0 & c
    \end{bmatrix}.\qquad
\end{gather}

Here, $U_1$ and $U_2$ commute and form compound twins. In the second type of transformation, the axes of orthorhombic unit cell are obtained from $\langle100\rangle$\textsubscript{tetragonal} directions \cite{bhattacharya2003microstructure}. The stretch tensors are given by

\begin{gather}
\label{eq:24}
    U_1 = 
    \begin{bmatrix}
        a & 0 & 0 \\ 0 & b & 0 \\ 0 & 0 & c
    \end{bmatrix},\qquad
    U_2 = 
    \begin{bmatrix}
        b & 0 & 0 \\ 0 & a & 0 \\ 0 & 0 & c
    \end{bmatrix}.\qquad
\end{gather}

Again, $U_1$ and $U_2$ commute and form compound twins.

\subsection*{Tetragonal to monoclinic transformation}
For tetragonal to monoclinic transformations, we have four variants in the martensitic phase. Three different transformation pathways have been reported for this transformation \cite{zirconia_monoclinic}. We discuss only one of these pathways, which is the only case where we get commuting matrices. This is the case when $[110]_t\rightarrow[100]_m\;$, $[001]_t\rightarrow[010]_m\;$ and $\;[1\bar{1}0]_t\rightarrow[001]_m$. The stretch tensors are of the form

\begin{gather}
\label{eq:25}
    U_1 = 
    \begin{bmatrix}
        a & b & 0 \\ b & c & 0 \\ 0 & 0 & d
    \end{bmatrix},\qquad
    U_2 = 
    \begin{bmatrix}
        a & -b & 0 \\ -b & c & 0 \\ 0 & 0 & d
    \end{bmatrix},\qquad
     U_3 = 
    \begin{bmatrix}
        c & b & 0 \\ b & a & 0 \\ 0 & 0 & d
    \end{bmatrix},\qquad
    U_4 = 
    \begin{bmatrix}
        c & -b & 0 \\ -b & a & 0 \\ 0 & 0 & d
    \end{bmatrix}.\qquad
\end{gather}

Here, the pairs $(U_1,U_4)$ and $(U_2,U_3)$ are the only ones that commute with each other. These pairs form non-conventional compound domains. All other pairs, $(U_1,U_2)$, $(U_1,U_3)$, $(U_2,U_4)$ and $(U_3,U_4)$, form compound twins and do not commute.

\subsection*{Orthorhombic to monoclinic transformation}
For orthorhombic to monoclinic transformations, we have two variants in the martensitic phase. This transformation is used as a prototypical example for studying the two well problem, along with the case of tetragonal to orthorhombic transformation \cite{Luskin_1996}. The stretch tensors are of the form

\begin{gather}
\label{eq:26}
    U_1 = 
    \begin{bmatrix}
        a & 0 & 0 \\ b & c & 0 \\ 0 & 0 & d
    \end{bmatrix},\qquad
    U_2 = 
    \begin{bmatrix}
        a & 0 & 0 \\ -b & c & 0 \\ 0 & 0 & d
    \end{bmatrix}.\qquad
\end{gather}

$U_1$ and $U_2$ form compound twins with each other and do not commute.

\begin{remark}
\label{rem:comm_and_compound}
    The property of commutation of stretch tensors bears an interesting relation to the symmetry of corresponding variants and depends only on the symmetry of the martensitic phase. For tetragonal martensite, all variants that commute form compound twins, and all compound twins commute with each other. For orthorhombic martensite, all compound twins commute with each other, but all variants that commute are not necessarily compatible with each other. For monoclinic martensite, the only variants that commute form non-conventional twins. The (conventional) compound twins do not commute in cubic to monoclinic transformation, see \autoref{fig:venn_structure}. This pattern seems to hold irrespective of the symmetry of austenite phase.
\end{remark}
\begin{theorem}
\label{thm:commutation}
    Let $U,V \in \mathbb{R}^{3\times3}_{\scriptstyle{+sym}} $, $U\neq V$, be two compatible variants of martensite such that  $UV=VU$. If the eigenvalues of $U$, $\lambda_i, \;i=1,2,3,$ are all distinct, then $U$ and $V$ form compound domains.
\end{theorem}
\begin{proof}
    
$U$ and $V$ commute, so they have a joint eigenbasis, see \autoref{lem:commutation}. Let $\lambda_i,\;i=1,2,3,$ be the eigenvalues of $U$ and $\hat{u}_i$ be the corresponding eigenvectors. Since $\lambda_i$ are all distinct, $\hat{u}_i$ form an orthonormal basis of $\mathbb{R}^3$. Then, writing the spectral decomposition of $U$ and $V$, we have
\begin{gather}
U= \lambda_1\;\hat{u}_1\otimes \hat{u}_1 + \lambda_2\;\hat{u}_2\otimes \hat{u}_2 + \lambda_3\;\hat{u}_3\otimes \hat{u}_3,\qquad V= \lambda_1\;\hat{u}_i\otimes \hat{u}_i + \lambda_2\;\hat{u}_j\otimes \hat{u}_j + \lambda_3\;\hat{u}_k\otimes \hat{u}_k,\nonumber
\end{gather}
where $\{i,j,k\}$ is a some permutation of $\{1,2,3\}$, other than $\{1,2,3\}$. Since $U$ and $V$ are compatible, there exists a unit vector $\hat{e}$ such that $Q=-I +2\:\hat{e}\otimes\hat{e}$ and $V=QUQ^{T}$ \cite{CHEN20132566}, we have,
\begin{gather}
V= \lambda_1\;Q\hat{u}_1\otimes Q\hat{u}_1 + \lambda_2\;Q\hat{u}_2\otimes Q\hat{u}_2 + \lambda_3\;Q\hat{u}_3\otimes Q\hat{u}_3.\qquad \nonumber
\end{gather}
Comparing the two expressions of $V$, we have,  $Q\hat{u}_1=\pm\hat{u}_i,\quad Q\hat{u}_2=\pm\hat{u}_j, \quad Q\hat{u}_3=\pm\hat{u}_k$.\\[5 pt]
Without loss of generality, suppose $i=1$, then
\begin{gather}
\nonumber
 Q\hat{u}_1 = \pm\;\hat{u}_1\; \implies (-I + 2\;\hat{e}\otimes\hat{e})\:\hat{u_1}=\pm\hat{u}_1 \implies -\hat{u}_1 \; +\; 2(\hat{u}_1\:.\:\hat{e})\:\hat{e} = \pm \:\hat{u}_1 \\[3pt]
 \nonumber\implies \text{either }\:\hat{e} \perp \:\hat{u}_1\:\:\text{or}\;\;
 \hat{e} = \hat{u}_1\;\perp\hat{u}_2.
 \end{gather}
If $i=2$, then
\begin{gather}
\nonumber
 Q\hat{u}_1 = \pm\;\hat{u}_2\; \implies (-I + 2\;\hat{e}\otimes\hat{e})\:\hat{u_1}=\pm\hat{u}_2 \implies -\hat{u}_1 \; +\; 2(\hat{u}_1\:.\:\hat{e})\:\hat{e} = \pm \:\hat{u}_2 \\[3pt]
 \nonumber\implies \:\hat{e} \;= \:\frac{1}{\sqrt{2}} (\pm\hat{u}_2+\hat{u}_1) \perp \hat{u}_3.
 \end{gather}
If $i=3$, then
 \begin{gather}
\nonumber
 Q\hat{u}_1 = \pm\;\hat{u}_3\; \implies (-I + 2\;\hat{e}\otimes\hat{e})\:\hat{u_1}=\pm\hat{u}_3 \implies -\hat{u}_1 \; +\; 2(\hat{u}_1\:.\:\hat{e})\:\hat{e} = \pm \:\hat{u}_3 \\[3pt]
 \nonumber\implies \:\hat{e} \;= \:\frac{1}{\sqrt{2}} (\pm\hat{u}_3+\hat{u}_1)\;\perp \hat{u}_2.
 \end{gather}
In all cases, $\hat{e}$ is perpendicular to an eigenvector of $U$. Then by \autoref{thm:compound_domains}, $U$ and $V$ form compound domains.
\end{proof}
\begin{corollary}
\label{cor:non_conv_comm}
     Non-conventional twins arising in a cubic to monoclinic-II transformation and a tetragonal to monoclinic transformation are compound domains.
\end{corollary}
\begin{proof}
    It is easy to check that in a cubic to monoclinic-II transformation, the pairs $(U_1,U_{12})$, $(U_2,U_{5})$, $(U_3,U_{9})$, $(U_4,U_{10})$, $(U_6,U_{8})$ and $(U_7,U_{11})$ as defined in \eqref{eq:22} are compatible and commute, therefore by \autoref{thm:commutation}, form compound domains. Similarly, the non-conventional twins that arise in a tetragonal to monoclinic transformation, i.e. the pairs $(U_1,U_{12})$, $(U_2,U_{5})$ as defined in \eqref{eq:25} are also compatible and commute, therefore, form compound domains.
\end{proof}

The first application of the commutation property of martensitic variants is that we can write the solutions to the twinning equations of the non-conventional twin pairs mentioned in \autoref{cor:non_conv_comm}, using simple formulae, by considering them as compound domains. According to \autoref{thm:compound_domains}, Proposition 12 and Corollary 13 of \cite{CHEN20132566}, each pair listed in \autoref{cor:non_conv_comm} must have two $180^\circ$ rotations relating them, the axes of which are orthogonal to an eigenvector of the stretch tensor, as well as to each other. Once the two axes relating a given pair are known, the corresponding twinning solutions can be explicitly computed using \eqref{eq:14}. The general formula for determining the rotation axes that relate compound domains is provided in equation A.27 of \cite{CHEN20132566}. Substituting the axes obtained from this expression into \eqref{eq:14} yields the twinning solutions for non-conventional twins.\\[5pt]
A tedious computation using equation $A.27$ of \cite{CHEN20132566} shows that for the pairs listed in \autoref{cor:non_conv_comm}, the symmetry axes lie in a plane orthogonal to the axis of monoclinic symmetry. In this setting, the symmetry axes can be obtained as follows.\\[5 pt]
Consider the pair $(U_1,U_{12})$ from \eqref{eq:22}. The axis of monoclinic symmetry for $U_1$ is ${\{1,0,0}\}$, which is also an eigenvector of $U_1$. Let $q\in \mathbb{R}$ be such that there is a $180^\circ$ rotation about the unit vector, $\hat{e}=\{0,q,\sqrt{1-q^2}\}$, given by $R_1=-I+2\;\hat{e}\otimes\hat{e}$. Then, we must have have,
\begin{gather}
\nonumber
    (-I+2\;\hat{e}\otimes\hat{e})U_1(I+2\;\hat{e}\otimes\hat{e})=U_{12}
\end{gather}

Upon solving, we have $q= \pm\sqrt{\frac{1}{2}\pm\frac{|b|}{\sqrt{D}}}$, where $D=4b^2 +(a-c)^2$. This expression for $q$ contains four solutions, out of which only two are admissible, depending on the sign of $(a-c)b$. For example, if $(a-c)b>0$, then the two axes obtained from this expression of $q$, are orthogonal to each other, and are given by,
\begin{equation}
    \hat{e}_1=\begin{pmatrix}
0 \\
\sqrt{\frac{1}{2}+\frac{|b|}{\sqrt{D}}} \\
\sqrt{\frac{1}{2}-\frac{|b|}{\sqrt{D}}}
\end{pmatrix}, \qquad  \hat{e}_2=\begin{pmatrix}
0 \\
-\sqrt{\frac{1}{2}-\frac{|b|}{\sqrt{D}}} \\
\sqrt{\frac{1}{2}+\frac{|b|}{\sqrt{D}}}
\end{pmatrix}
\notag
\end{equation}
\\[5 pt]
Similarly, we can obtain the symmetry axes for the remaining cases of non-conventional compound domains in cubic to monoclinic-II transformation. The exact same procedure is followed to obtain the symmetry axes for the non-conventional compound domains arising in tetragonal to monoclinic transformation. \autoref{tab:non_conv_MC} and \autoref{tab:non_conv_TM} contain the symmetry axes for all non-conventional compound domains arising in these two transformations.

\begin{table}
    \centering
\begin{tabular}{ |c|c|c|} 
 \hline
 \multirow{2}{*}{Variant pairs} &  \multicolumn{2}{|c|}{Symmetry axes}\\ 
 \cline{2-3}
  &  $\hat{e}_1$ & $\hat{e}_2$\\ 
 \hline
 & & \\
 $(U_{1},U_{12})$ & $\{0,\;\sqrt{\frac{1}{2}+\frac{|b|}{\sqrt{D}}},\;
\sqrt{\frac{1}{2}-\frac{|b|}{\sqrt{D}}}\}$ & $\{0,\;-\sqrt{\frac{1}{2}-\frac{|b|}{\sqrt{D}}},\;
\sqrt{\frac{1}{2}+\frac{|b|}{\sqrt{D}}}\}$\\
 \hline
 & & \\
$(U_{2},U_{5})$ & $\{0,\;\sqrt{\frac{1}{2}-\frac{|b|}{\sqrt{D}}},\;
\sqrt{\frac{1}{2}+\frac{|b|}{\sqrt{D}}}\}$ & $\{0,\;-\sqrt{\frac{1}{2}+\frac{|b|}{\sqrt{D}}},\;
\sqrt{\frac{1}{2}-\frac{|b|}{\sqrt{D}}}\}$\\
 \hline
 & & \\
 $(U_{3},U_{9})$ & $\{\sqrt{\frac{1}{2}-\frac{|b|}{\sqrt{D}}},\;0,\;
\sqrt{\frac{1}{2}+\frac{|b|}{\sqrt{D}}}\}$ & $\{-\sqrt{\frac{1}{2}+\frac{|b|}{\sqrt{D}}},\;0,\;
\sqrt{\frac{1}{2}-\frac{|b|}{\sqrt{D}}}\}$\\
 \hline
 & & \\
 $(U_{4},U_{10})$  & $\{\sqrt{\frac{1}{2}+\frac{|b|}{\sqrt{D}}},\;0,\;
\sqrt{\frac{1}{2}-\frac{|b|}{\sqrt{D}}}\}$ & $\{-\sqrt{\frac{1}{2}-\frac{|b|}{\sqrt{D}}},\;0,\;
\sqrt{\frac{1}{2}+\frac{|b|}{\sqrt{D}}}\}$\\
 \hline
 & & \\
 $(U_{6},U_{8})$ & $\{\sqrt{\frac{1}{2}-\frac{|b|}{\sqrt{D}}},\;\sqrt{\frac{1}{2}+\frac{|b|}{\sqrt{D}}},\;
0\}$ & $\{-\sqrt{\frac{1}{2}+\frac{|b|}{\sqrt{D}}},\;\sqrt{\frac{1}{2}-\frac{|b|}{\sqrt{D}}},\;
0\}$\\
 \hline
 & & \\
 $(U_{7},U_{11})$ & $\{\sqrt{\frac{1}{2}+\frac{|b|}{\sqrt{D}}},\;\sqrt{\frac{1}{2}-\frac{|b|}{\sqrt{D}}},\;
0\}$ & $\{-\sqrt{\frac{1}{2}-\frac{|b|}{\sqrt{D}}},\;\sqrt{\frac{1}{2}+\frac{|b|}{\sqrt{D}}},\;
0\}$\\
 \hline
\end{tabular}
\caption{The symmetry axes for the non-conventional twins arising in cubic to monoclinic-II transformations (22).}
\label{tab:non_conv_MC}
\end{table}
\begin{table}
    \centering
\begin{tabular}{ |c|c|c|} 
 \hline
 \multirow{2}{*}{Variant pairs} &  \multicolumn{2}{|c|}{Symmetry axes}\\ 
 \cline{2-3}
  &  $\hat{e}_1$ & $\hat{e}_2$\\ 
 \hline
 & & \\
 $(U_{1},U_{4})$ & $\{\sqrt{\frac{1}{2}-\frac{|b|}{\sqrt{D}}},\;\sqrt{\frac{1}{2}+\frac{|b|}{\sqrt{D}}},\;
0\}$ & $\{-\sqrt{\frac{1}{2}+\frac{|b|}{\sqrt{D}}},\;\sqrt{\frac{1}{2}-\frac{|b|}{\sqrt{D}}},\;
0\}$\\
 \hline
 & & \\
$(U_{2},U_{3})$ & $\{\sqrt{\frac{1}{2}+\frac{|b|}{\sqrt{D}}},\;\sqrt{\frac{1}{2}-\frac{|b|}{\sqrt{D}}},\;
0\}$ & $\{-\sqrt{\frac{1}{2}-\frac{|b|}{\sqrt{D}}},\;\sqrt{\frac{1}{2}+\frac{|b|}{\sqrt{D}}},\;
0\}$\\
 \hline
\end{tabular}
\caption{The symmetry axes for the non-conventional twins arising in tetragonal to monoclinic transformations (25).}
\label{tab:non_conv_TM}
\end{table}
We now turn to the second application of commutation property of the martensitic variants. We will show that when two compatible martensitic variants commute, then under certain assumptions, we can form planar triple clusters between the two martensitic variants and austenite. These triple clusters are similar (\textit{in geometry}) to the ones implied by the cofactor conditions for Type I/II domains, see \Cref{fig:triple_1_2}. However, there is an important difference between the two cases. For Type I/II domains under the cofactor conditions, either a Type I solution or a Type II solution of the twinning equation can form a planar triple cluster with austenite. Both the solutions cannot simultaneously form the triple clusters with austenite, see \autoref{thm:dellaportaR}. In contrast to this, the triple clusters obtained by virtue of the commutation property allows both the solutions of the twinning equation to form two distinct planar triple clusters with austenite. Both these triple clusters are of the same type, i.e. either both the triple clusters have parallel shear vectors and coplanar normal vectors, see \autoref{fig:firstkind} (\textit{similar to Type-I triple cluster}) or both have coplanar shear vectors and same interface normals, see \autoref{fig:secondkind} (\textit{similar to Type-II triple cluster}). To avoid confusion, we call the planar triple clusters of parallel shear type for compound domains as `triple clusters of first kind' and the planar triple clusters of same normal type for compound domains as `triple clusters of second kind'.\\[5 pt]
This ability of the commuting variants to form two distinct triple clusters with austenite, again implies that commuting variants cannot form Type I/II domains, see \autoref{thm:dellaportaR}. Then, commuting (compatible) variants must form compound domains, which was proved in \autoref{thm:commutation} using a different approach. We now proceed towards constructing these triple clusters with austenite using the commutation property. 

\subsection{Commutation and triple clusters with austenite}
\label{sec:comm_triple}

\begin{theorem}
    \label{thm:firstkind}
Let $U,V \in \mathbb{R}^{3\times3}_{\scriptstyle{+sym}} $, $U\neq V$, be two variants of martensite such that $UV=VU$ and the eigenvalues of $U$ satisfy, $\frac{1}{\sqrt{2}}<\lambda_1<\lambda_2 = 1<\lambda_3$. Suppose that the eigenvector corresponding to $\lambda_2$ is same for both $U$ and $V$. If the other two eigenvalues of $U$ satisfy 
\begin{equation}
\label{eq:27}
    \lambda_3=\frac{\lambda_1}{\sqrt{2\lambda_1^2 -1}},
\end{equation}

then there exist rotations $R_u^{I+}$, $R_v^{I+}$, $R_u^{I-}$, $R_v^{I-}$ $\in SO(3)$, vectors $b^{I+}$, $b^{I-}\;\in \mathbb{R}^3$ and normals $\hat{m}_u^{I+}$, $\hat{m}_v^{I+}$, $\hat{m}_u^{I-}$, $\hat{m}_v^{I-}\in \mathbb{S}^2$ such that the following sets of equations hold simultaneously.

\begin{gather}\nonumber
    R_u^{I+}U - I = b^{I+} \otimes \hat{m}_u^{I+},\\[5 pt]\nonumber
    R_v^{I+}V - I = b^{I+} \otimes \hat{m}_v^{I+},\\[5 pt]
    R_u^{I+}U - R_v^{I+}V = b^{I+} \otimes (\hat{m}_u^{I+} -\hat{m}_v^{I+}).
    \label{eq:28}
\end{gather}
\begin{gather}\nonumber
    R_u^{I-}U - I = b^{I-} \otimes \hat{m}_u^{I-},\\[5pt]\nonumber
    R_v^{I-}V - I = b^{I-} \otimes -\hat{m}_v^{I-},\\[5pt]
    R_u^{I-}U - R_v^{I-}V = b^{I-} \otimes (\hat{m}_u^{I-} +\hat{m}_v^{I-}).
    \label{eq:29}
\end{gather}

\eqref{eq:28} and \eqref{eq:29} imply that there exist two distinct planar triple clusters between the two variants, $U, V$ and austenite, $I$ \textit{(triple clusters of first kind, see \autoref{fig:firstkind})}.
\end{theorem}
\begin{proof} $U$ and $V$ are symmetry related domains, so there is a rotation $R \in SO(3)$ such that $V=RUR^T$, implying $U$ is similar to $V$. Then $U$ and $V$ have the same spectrum of eigenvalues.\\[3pt]
So, $U$ and $V$ both have their middle eigenvalue equal to 1. Hence each can form a compatible interface with the austenite. Then there are two solutions to the equation, $QU-I=b_u\;\otimes \;\hat{m}_u$, given by $R^{I\pm}_uU = b^{I\pm}_u\;\otimes \;\hat{m}^{I\pm}_u $, where $b^{I\pm}_u$ and $\hat{m}^{I\pm}_u$ are computed using \autoref{corollary:identity_solutions}.
\begin{equation}
    b_u^{I\pm} = \frac{\rho}{\sqrt{\lambda_3^2 - \lambda_1^2}}\left(\lambda_3\sqrt{1-\lambda_1^2}\;\hat{u}_1\; \; \pm\; \lambda_1\sqrt{\lambda_3^2 -1}\;\hat{u}_3 \right),
    \nonumber
\end{equation}

\begin{equation}
    \hat{m}_u^{I\pm} = \frac{\lambda_3 - \lambda_1}{\rho\sqrt{\lambda_3^2 - \lambda_1^2}}\left(-\sqrt{1-\lambda_1^2}\;\hat{u}_1\;  \pm\;\sqrt{\lambda_3^2 -1}\;\hat{u}_3 \right).
    \label{eq:30}
\end{equation}

Similarly, we can write the solutions to the equation, $QV-I=b_v\;\otimes \;\hat{m}_v$, as $R^{I\pm}_vV = b^{I\pm}_v\;\otimes \;\hat{m}^{I\pm}_v $ as
\begin{equation}
    b_v^{I\pm} = \frac{\rho}{\sqrt{\lambda_3^2 - \lambda_1^2}}\left(\lambda_3\sqrt{1-\lambda_1^2}\;\hat{v}_1\; \; \pm\; \lambda_1\sqrt{\lambda_3^2 -1}\;\hat{v}_3 \right),
    \nonumber
\end{equation}

\begin{equation}
\label{eq:31}
    \hat{m}_v^{I\pm} = \frac{\lambda_3 - \lambda_1}{\rho\sqrt{\lambda_3^2 - \lambda_1^2}}\left(-\sqrt{1-\lambda_1^2}\;\hat{v}_1\;  \pm\;\sqrt{\lambda_3^2 -1}\;\hat{v}_3 \right).
\end{equation}\\[10pt]

\begin{figure}[hbt!]
    \centering
    \begin{minipage}[b]{0.32\textwidth}
    \centering
    \includegraphics[scale=0.59]{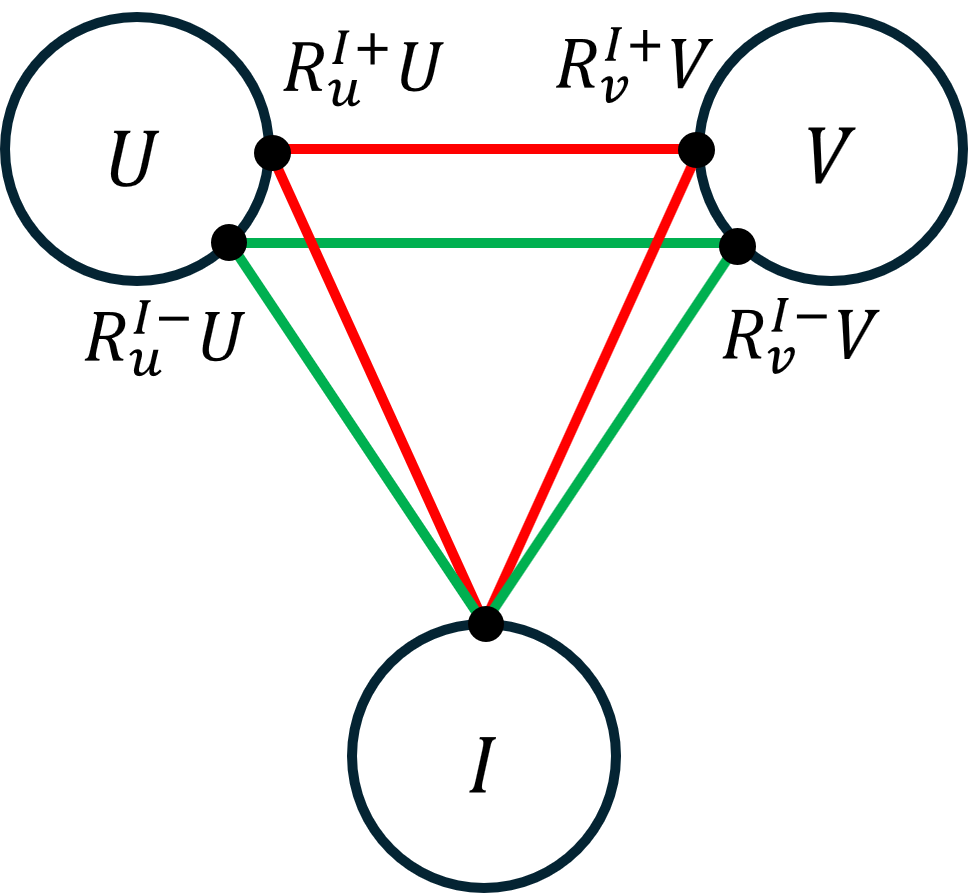}
    \subcaption{}
    \label{fig:firstkind_a}
    \end{minipage}%
    \begin{minipage}[b]{0.32\textwidth}
    \centering
    \includegraphics[scale=0.59]{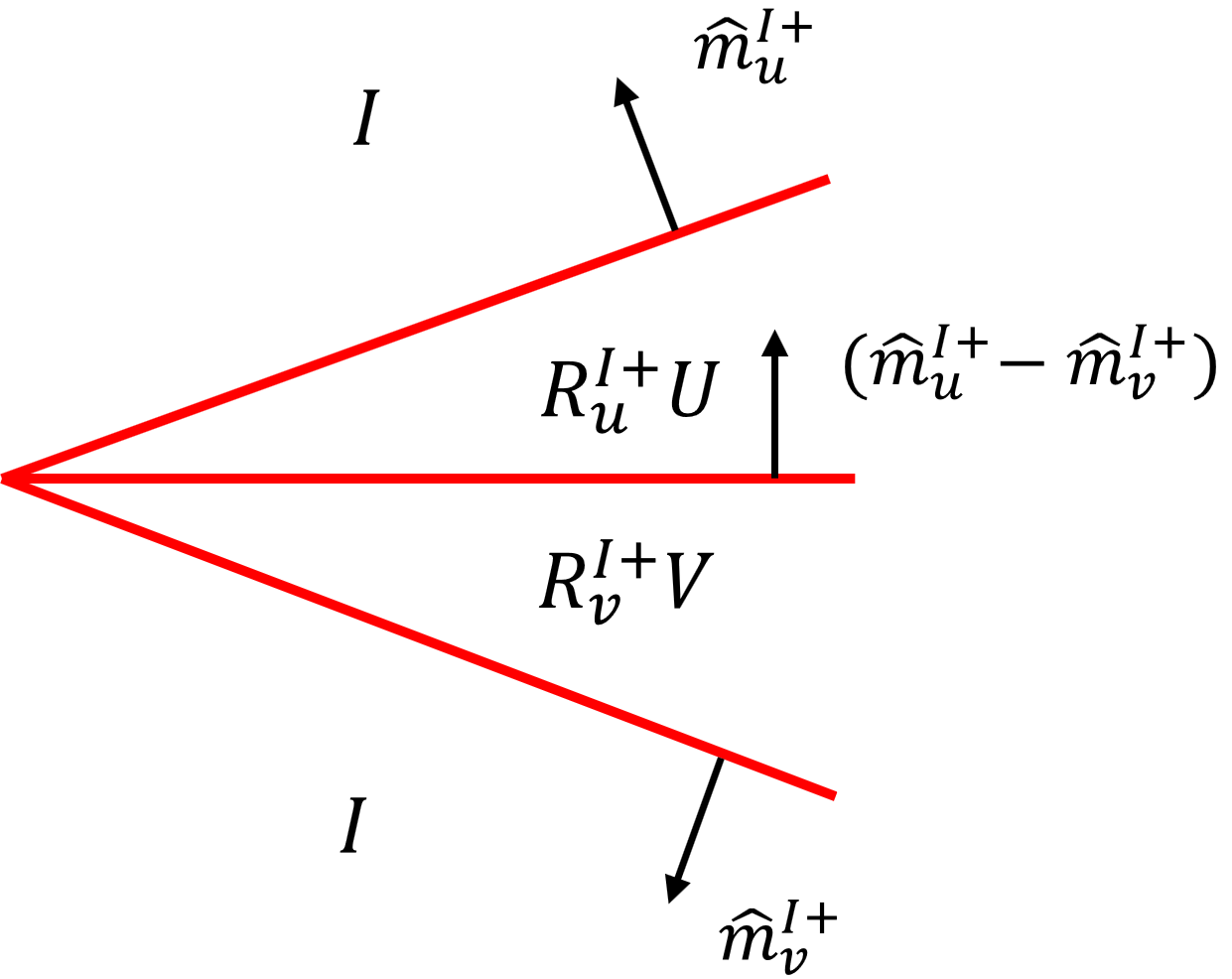}
    \subcaption{}
    \label{fig:firstkind_b}
    \end{minipage}
    \begin{minipage}[b]{0.32\textwidth}
    \centering
    \includegraphics[scale=0.59]{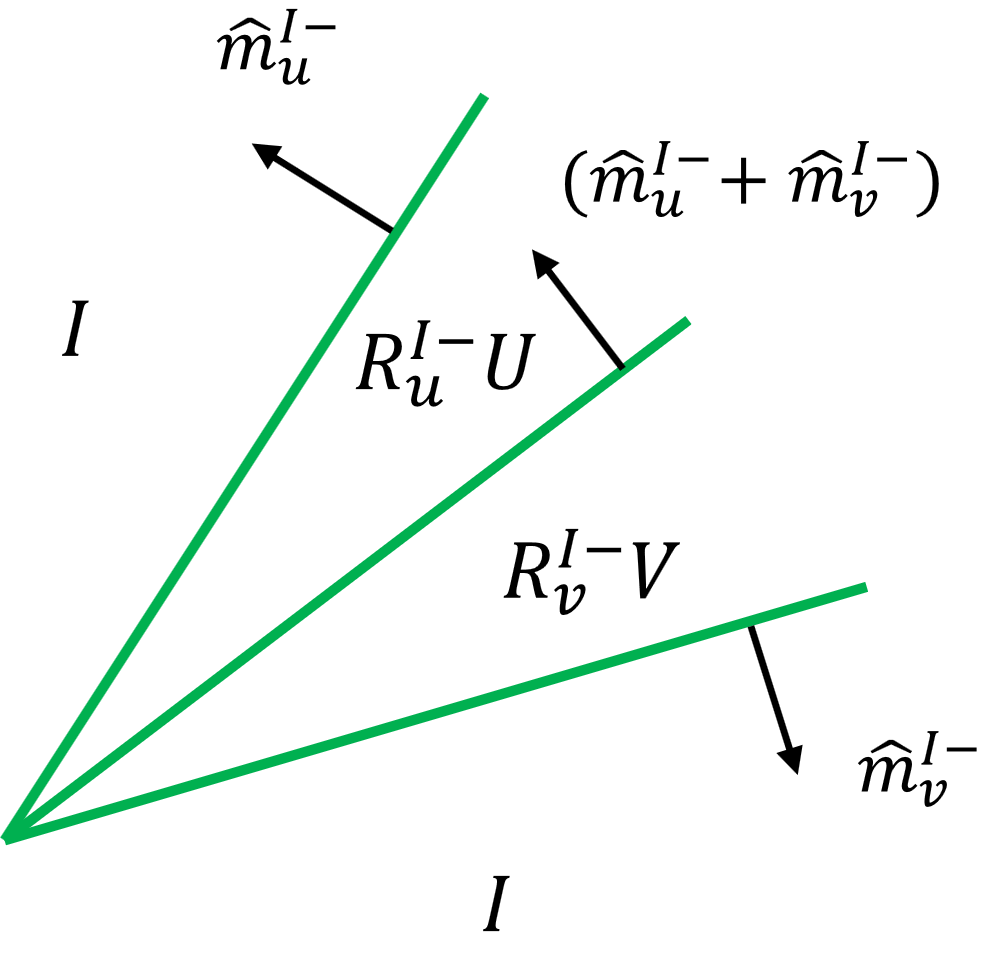}
    \subcaption{}
    \label{fig:firstkind_c}
    \end{minipage}
    \caption{Two distinct planar triple clusters between two commuting variants of martensite and austenite, represented by two different colours (red and green)  (a) Compatibility diagram in strain space (b) Triple cluster implied by \eqref{eq:35} and \eqref{eq:36}  (c) Triple cluster implied by \eqref{eq:37}. (b) and (c) lie in different planes. We call (b) and (c), triple clusters of first kind}
    \label{fig:firstkind}
\end{figure}

Writing the spectral decomposition of $U$ and $V$,
\begin{gather}
U= \lambda_1\;\hat{u}_1\otimes \hat{u}_1 + \hat{u}_2\otimes \hat{u}_2 + \lambda_3\;\hat{u}_3\otimes \hat{u}_3,\qquad V= \lambda_1\;\hat{v}_1\otimes \hat{v}_1 + \hat{u}_2\otimes \hat{u}_2 + \lambda_3\;\hat{v}_3\otimes \hat{v}_3.\nonumber
\end{gather}
$U$ and $V$ commute, so there exists a joint eigenbasis for $U$ and $V$, see \autoref{lem:commutation}. Since $\lambda_i$ are all distinct, eigenvectors are mutually orthogonal, therefore, we must have $\hat{v}_1 = \hat{u}_3$ and $\hat{v}_3 = \hat{u}_1$, or else $U=V$, which is prohibited. \\[5pt]
Rewriting \eqref{eq:31}, we have the solutions of $QV-I=b_v\;\otimes \;\hat{m}_v$ as
\begin{equation}
    b_v^{I\pm} = \frac{\rho}{\sqrt{\lambda_3^2 - \lambda_1^2}}\left(\lambda_3\sqrt{1-\lambda_1^2}\;\hat{u}_3\; \; \pm\; \lambda_1\sqrt{\lambda_3^2 -1}\;\hat{u}_1 \right),
    \nonumber
\end{equation}

\begin{equation}
\label{eq:32}
    \hat{m}_v^{I\pm} = \frac{\lambda_3 - \lambda_1}{\rho\sqrt{\lambda_3^2 - \lambda_1^2}}\left(-\sqrt{1-\lambda_1^2}\;\hat{u}_3\;  \pm\;\sqrt{\lambda_3^2 -1}\;\hat{u}_1 \right).
\end{equation}

Using the relation, $\lambda_3=\frac{\lambda_1}{\sqrt{2\lambda_1^2 -1}} \implies\lambda_3\sqrt{1-\lambda_1^2}=\lambda_1\sqrt{\lambda_3^2 -1}$, in \eqref{eq:30} and \eqref{eq:32}, we have

\begin{gather}\nonumber
    b_u^{I\pm} = C(\hat{u}_1 \pm \hat{u}_3),\\[3 pt]\nonumber
    b_v^{I\pm} = C(\hat{u}_3 \pm \hat{u}_1),\\[3 pt]
    \nonumber
    \hat{m}_u^{I\pm} = D(-\sqrt{ 2\lambda_1^2-1}\; \hat{u}_1 \pm \hat{u}_3),\\[3 pt]
    \hat{m}_v^{I\pm} = D(-\sqrt{2\lambda_1^2-1}\; \hat{u}_3 \pm \hat{u}_1).
    \label{eq:33}
\end{gather}
where,
\begin{gather}
    C =\frac{\lambda _1 \sqrt{1-\lambda _1^2} \rho }{\sqrt{2 \lambda _1^2-1} \sqrt{\lambda _1^2 \left(\frac{1}{2 \lambda _1^2-1}-1\right)}}, \;\;\;\;\;\;\;\;\;    \text{and}    \;\;\;\;\;\;\;\;\;\; D=\frac{\sqrt{\lambda _1^2 \left(\frac{1}{2 \lambda _1^2-1}-1\right)} \sqrt{\frac{\lambda _1^2}{2 \lambda _1^2-1}-1}}{\lambda _1 \left(\frac{1}{\sqrt{2 \lambda _1^2-1}}+1\right) \rho }.\label{eq:34}
\end{gather}

We can write, $b^{I+}_u = b_v^{I+} = b^{I+} (say)$ and $b^{I-}_u = -b_v^{I-} = b^{I-} (say)$. Then, we have
\begin{gather}
    \nonumber R_u^{I+}U -I = b^{I+} \otimes \;\hat{m}_u^{I+},\\[3pt]
    R_v^{I+}V -I = b^{I+} \otimes \;\hat{m}_v^{I+}.\label{eq:35}
\end{gather}
Subtracting \eqref{eq:35}\textsubscript{2} from \eqref{eq:35}\textsubscript{1}, we have 
\begin{gather}
    R_u^{I+}U - R_v^{I+}V  = b^{I+} \otimes \;(\hat{m}_u^{I+}- \hat{m}_v^{I+}).\label{eq:36}
\end{gather}

Similarly, we can write the following set of equations,
\begin{gather}
    \nonumber R_u^{I-}U -I = b^{I-} \otimes \;\hat{m}_u^{I-},\\[3pt]
    \nonumber R_v^{I-}V -I = b^{I-} \otimes \;-\hat{m}_v^{I-},\\[3pt]
    R_u^{I-}U - R_v^{I-}V  = b^{I-} \otimes \;(\hat{m}_u^{I-} + \hat{m}_v^{I-}).\label{eq:37}
\end{gather}

If $\frac{1}{\sqrt{2}}<\lambda _1<1$, then we have $C>0$,  $\;b^{I+}=C(\hat{u}_1+\hat{u}_3)$ and  $b^{I-}=C(\hat{u}_1-\hat{u}_3)$. Clearly, $\; b^{I+} \nparallel b^{I-}$, which shows that the triple cluster implied by \eqref{eq:35},\eqref{eq:36} is distinct from the one implied by \eqref{eq:37}. We refer to these two planar clusters as triple clusters of the first kind.
\end{proof}

\begin{figure}
    \centering
    \begin{minipage}[b]{0.32\textwidth}
    \centering
    \includegraphics[scale=0.59]{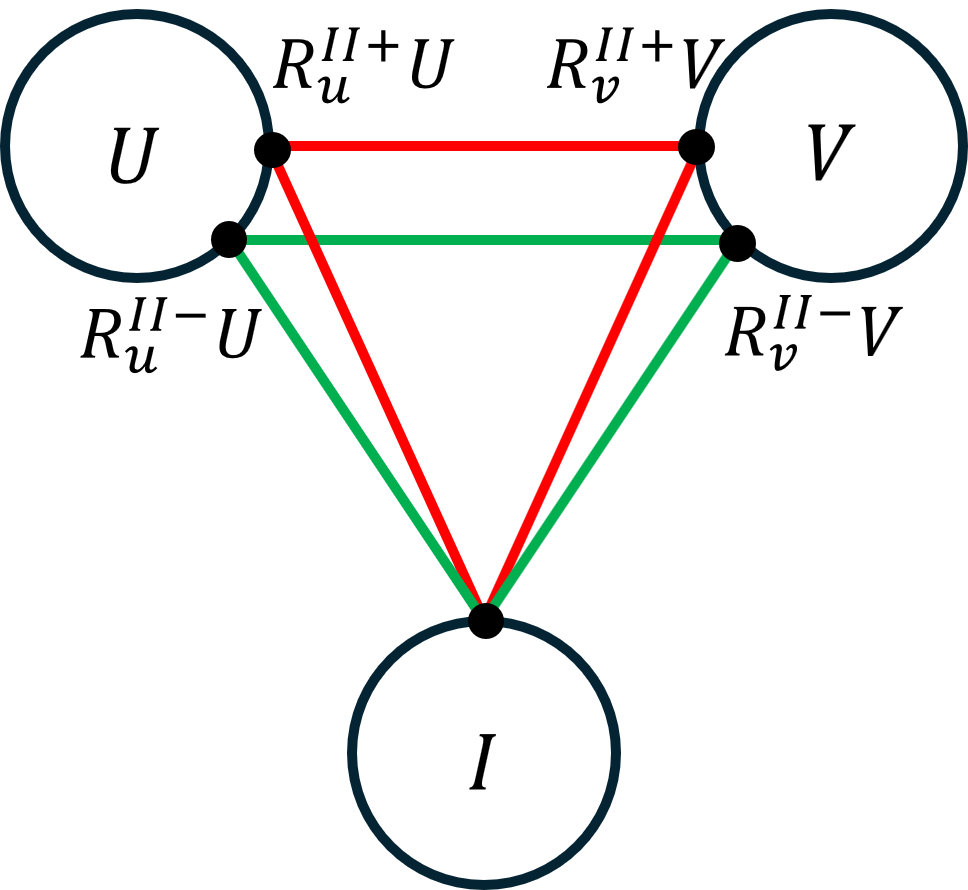}
    \subcaption{}
    \label{fig:secondkind_a}
    \end{minipage}%
    \begin{minipage}[b]{0.32\textwidth}
    \centering
    \includegraphics[scale=0.59]{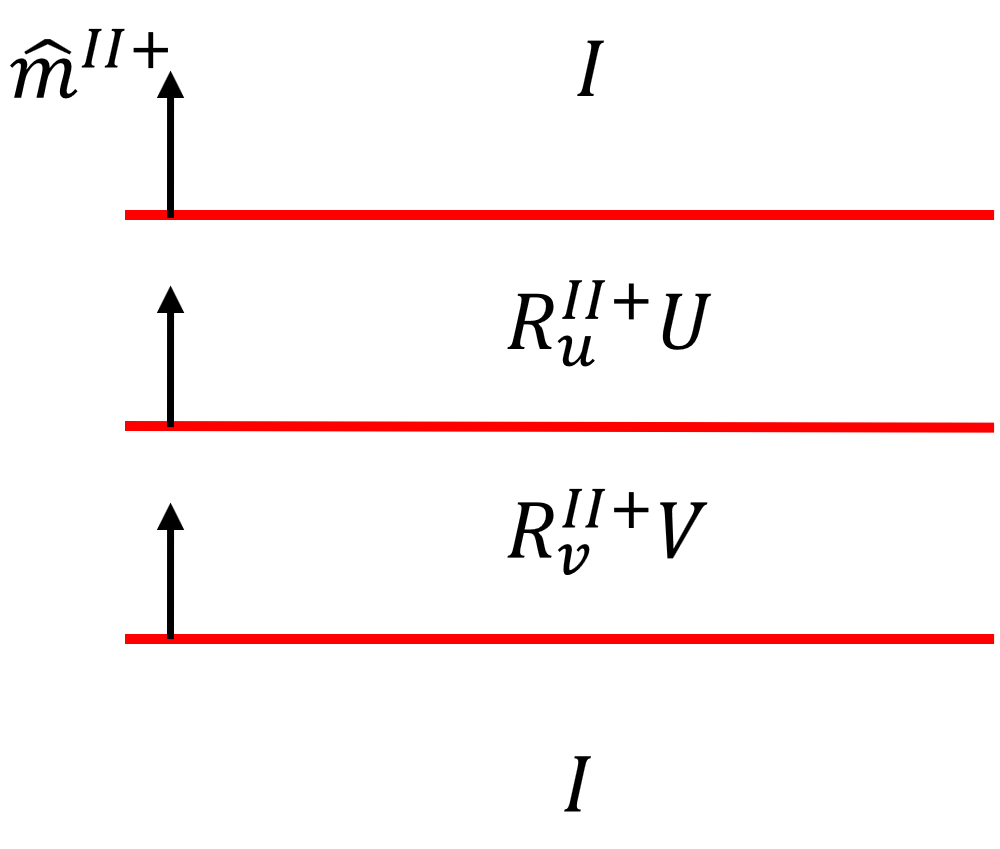}
    \subcaption{}
    \label{fig:secondkind_b}
    \end{minipage}
    \begin{minipage}[b]{0.32\textwidth}
    \centering
    \includegraphics[scale=0.59]{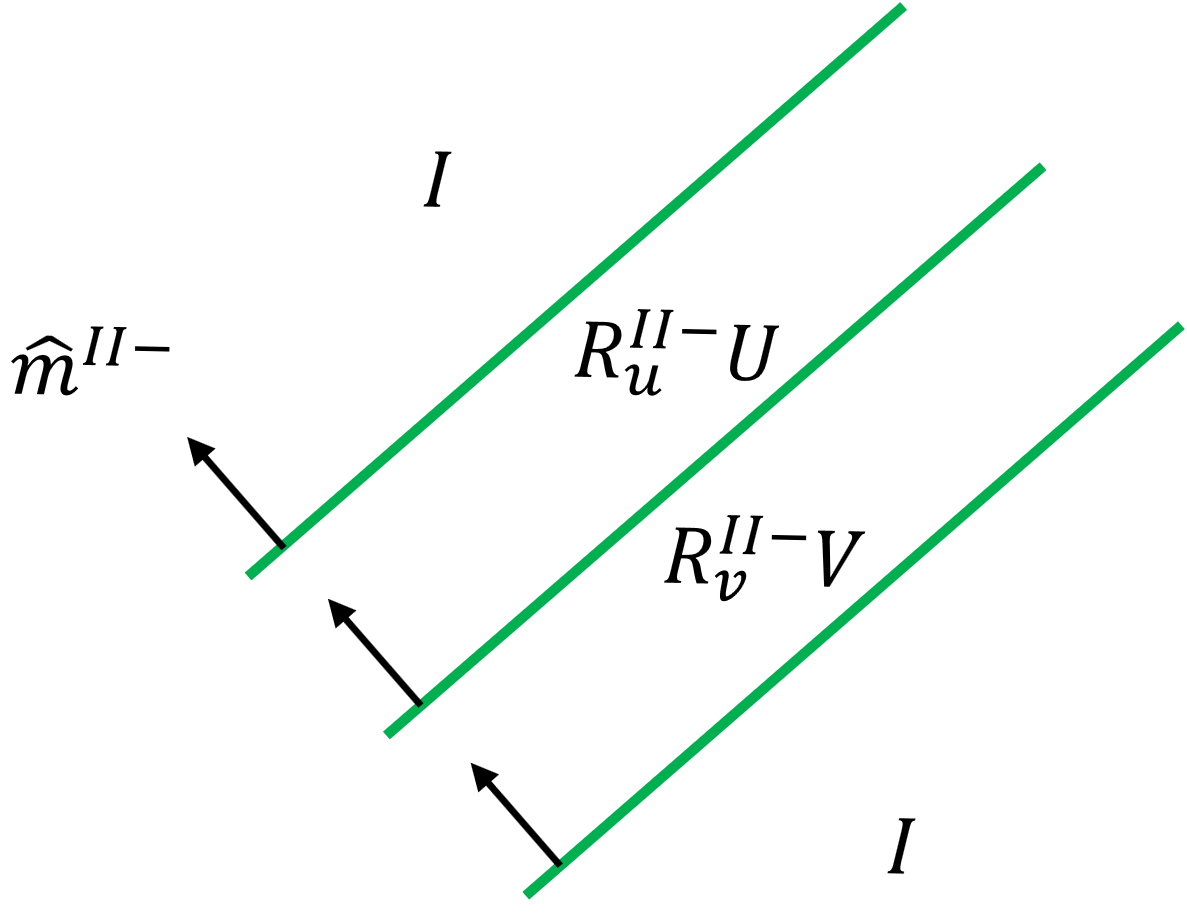}
    \subcaption{}
    \label{fig:secondkind_c}
    \end{minipage}
    \caption{Two distinct planar triple clusters between two commuting variants and austenite, represented by two different colors (red and green) (a) Compatibility diagram in strain space (b) Triple cluster implied by \eqref{eq:43} and \eqref{eq:44}  (c) Triple cluster implied by \eqref{eq:45}. Normals $\hat{m}^{II+}$ in (b) and $\hat{m}^{II-}$ in (c) are different. We call (b) and (c), triple clusters of second kind}
    \label{fig:secondkind}
\end{figure}
\begin{theorem}
    \label{thm:secondkind}
 Let $U,V \in \mathbb{R}^{3\times3}_{\scriptstyle{+sym}} $ be two variants of martensite such that $UV=VU$ and the eigenvalues of $U$ satisfy, $0<\lambda_1<\lambda_2 = 1<\lambda_3$. Suppose that the eigenvector corresponding to $\lambda_2$ is same for both $U$ and $V$. If the other two eigenvalues of $U$ satisfy
\begin{equation}
    \lambda_3 = \sqrt{2-\lambda_1^2},\label{eq:38}
\end{equation}

then there exist rotations $R_u^{II+}$, $R_v^{II+}$, $R_u^{II-}$, $R_v^{II-}$ $\in SO(3)$, vectors $b_u^{II+}$, $b_v^{II+}$, $b_u^{II-}$, $b_v^{II-}\;\in \mathbb{R}^3$ and normals $\hat{m}^{II+}$, $\hat{m}^{II-}\in \mathbb{S}^2$ such that the following sets of equations hold simultaneously.

\begin{gather}\nonumber
    R_u^{II+}U - I = b_u^{II+} \otimes \hat{m}^{II+},\\[5 pt]\nonumber
    R_v^{II+}V - I = b_v^{II+} \otimes \hat{m}^{II+},\\[5 pt]
    R_u^{II+}U - R_v^{II+}V =  (\hat{b}_u^{II+} +\hat{b}_v^{II+})\; \otimes \hat{m}^{II+}.\label{eq:39}
\end{gather}
\begin{gather}\nonumber
    R_u^{II-}U - I = b_u^{II-} \otimes \hat{m}^{II-},\\[5 pt]\nonumber
    R_v^{II-}V - I = b_v^{II-} \otimes \hat{m}^{II-},\\[5 pt]
    R_u^{II-}U - R_v^{II-}V =  (\hat{b}_u^{II-} -\hat{b}_v^{II-})\; \otimes \hat{m}^{II-}.\label{eq:40}
\end{gather}

\eqref{eq:39} and \eqref{eq:40} allow two distinct planar triple clusters to form between the two variants, $U, V$ and austenite, $I$ \textit{(triple clusters of the second kind, see \autoref{fig:secondkind})}.\\[5pt]
\end{theorem}
\begin{proof}The proof is similar to that of \autoref{thm:firstkind}. Substituting, $\lambda_3 = \sqrt{2-\lambda_1^2} \; \implies \sqrt{\lambda_3^2 -1} = \sqrt{1-\lambda_1^2}$, in equations \eqref{eq:32} and \eqref{eq:34}, we get

\begin{gather}\nonumber
    b_u^{II\pm} =E(\sqrt{2-\lambda_1^2}\;\hat{u}_1 \pm \lambda_1\hat{u}_3),\\[3 pt]
    \nonumber b_v^{II\pm} = E(\sqrt{2-\lambda_1^2}\;\hat{u}_3 \pm \lambda_1\hat{u}_1),\\[3 pt]
    \nonumber
    \hat{m}_u^{II\pm} = F(-\hat{u}_1 \pm \hat{u}_3),\\[3 pt]
    \hat{m}_v^{II\pm} = F(- \hat{u}_3 \pm \hat{u}_1).\label{eq:41}
\end{gather}
where,
\begin{gather}
    E =\frac{\rho }{\sqrt{2}}, \;\;\;\;\;\;\;\;\;    \text{and}    \;\;\;\;\;\;\;\;\;\; F=\frac{\sqrt{2-\lambda _1^2}-\lambda _1}{\sqrt{2} \rho }.\label{eq:42}
\end{gather}

We can write, $\hat{m}^{II+}_u = -\hat{m}_v^{II+} = \hat{m}^{II+} (say)$ and $m^{II-}_u = m_v^{II-} = m^{II-} (say)$. Then, we have
\begin{gather}
    \nonumber R_u^{II+}U -I = b_u^{II+} \otimes \;\hat{m}^{II+},\\[3pt]
    R_v^{II+}V -I = -b_v^{II+} \otimes \;\hat{m}^{II+}.\label{eq:43}
\end{gather}
Subtracting \eqref{eq:43}\textsubscript{2} from \eqref{eq:43}\textsubscript{1}, we have 
\begin{gather}
    R_u^{II+}U - R_v^{II+}V  = (b_u^{II+} + b_v^{II+}) \otimes \;\hat{m}^{II+}.\label{eq:44}
\end{gather}

Similarly, we can write the following set of equations,
\begin{gather}
    \nonumber R_u^{II-}U -I = b_u^{II-} \otimes \;\hat{m}^{II-},\\[3pt]
    \nonumber R_v^{II-}V -I = b_v^{II-} \otimes \;\hat{m}^{II-},\\[3pt]
    R_u^{II-}U - R_v^{II-}V  = (b_u^{II-} - b_v^{II-}) \otimes \;\hat{m}^{II-}.\label{eq:45}
\end{gather}

If $0<\lambda _1<1$, then we have $F>0,\; \hat{m}^{II+} = F(-\hat{u}_1+\hat{u}_3)$, and $\hat{m}^{II-}=F(-\hat{u}_1-\hat{u}_3)$. Clearly, $\hat{m}^{II+}\neq \hat{m}^{II-}$, which shows that the triple cluster implied by \eqref{eq:43},\eqref{eq:44} is distinct from the one implied by \eqref{eq:45}. We refer to these two planar clusters as triple clusters of the second kind. \\
\end{proof}

\begin{remark}
    \label{rem:27_38}
The requirement of \autoref{thm:firstkind} and \autoref{thm:secondkind} cannot be satisfied in cubic to tetragonal transformations, since two of the eigenvalues are equal for this transformation. For cubic to orthorhombic transformations of $\langle110\rangle_{cubic}$ type given by \eqref{eq:19}, these theorems hold only if $d=1$ and the other two eigenvalues satisfy either \eqref{eq:27} or \eqref{eq:38}. In this case, the compound twins can form two planar triple clusters with austenite. For cubic to monoclinic-II transformations, given by \eqref{eq:22}, we again require that $d=1$ and the other two eigenvalues satisfy either \eqref{eq:27} or \eqref{eq:38}. In this case, the non-conventional compound domains can form two planar triple clusters with austenite. Similarly, for tetragonal to orthorhombic transformations (\eqref{eq:23},\eqref{eq:24}), these theorems can be satisfied if $c=1$ and the other two eigenvalues satisfy either \eqref{eq:27} or \eqref{eq:38}. Here, compound twins can form two planar triple clusters with austenite. For tetragonal to monoclinic transformation \eqref{eq:25}, we require $d=1$, and either \eqref{eq:27} or \eqref{eq:38} holds. In this case, non conventional compound domains can form triple clusters with austenite. For cubic to monoclinic-I transformation \eqref{eq:21} and orthorhombic to monoclinic transformation \eqref{eq:26}, it is not possible to satisfy the premises of these theorems.
\end{remark}
\begin{remark}
  \label{rem:bothsolutions}  
The two triple clusters as shown in \autoref{fig:firstkind} \textit{(or \autoref{fig:secondkind})} imply that both the solutions of twinning equation between two variants $U$ and $V$ form triple clusters of the first \textit{(or the second kind)} with austenite. This can be shown simply by premultiplying the compatibility equations of planar clusters with appropriate rotations. To see this, premultiply \eqref{eq:35},\eqref{eq:36} by $(R_u^{I+})^T$ and \eqref{eq:37} by $(R_u^{I-})^T$ to obtain the following set of equations.
\begin{gather}
    \nonumber U -(R_u^{I+})^T = (R_u^{I+})^Tb^{I+} \otimes \;\hat{m}_u^{I+},\\[3pt]
    \nonumber(R_u^{I+})^TR_v^{I+}V -(R_u^{I+})^T = (R_u^{I+})^Tb^{I+} \otimes \;\hat{m}_v^{I+},\\[3pt]
    U - (R_u^{I+})^TR_v^{I+}V  = (R_u^{I+})^Tb^{I+} \otimes \;(\hat{m}_u^{I+}- \hat{m}_v^{I+}).\label{eq:46}\\[12pt]
    \nonumber U -(R_u^{I-})^T = (R_u^{I-})^Tb^{I-} \otimes \;\hat{m}_u^{I-},\\[3pt]
     \nonumber(R_u^{I-})^TR_v^{I-}V -(R_u^{I-})^T = (R_u^{I-})^Tb^{I-} \otimes \;-\hat{m}_v^{I-},\\[3pt]
     U - (R_u^{I-})^TR_v^{I-}V  = (R_u^{I-})^Tb^{I-} \otimes \;(\hat{m}_u^{I-} + \hat{m}_v^{I-}).\label{eq:47}
\end{gather}\\[10 pt]

\begin{figure}[hbt!]
    \centering
    \begin{minipage}[b]{0.5\textwidth}
    \centering
    \includegraphics[scale=0.7]{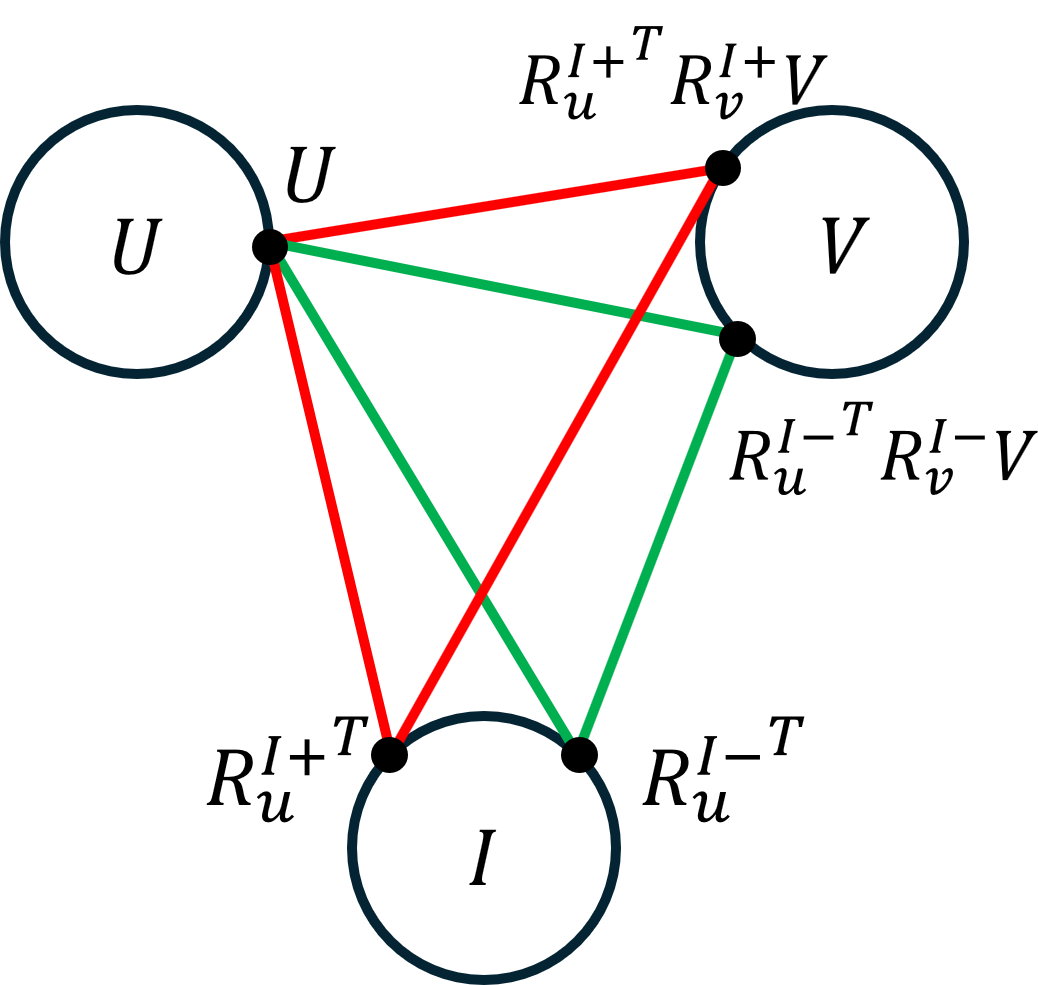}
    \subcaption{}
    \label{fig:bothsolutions_a}
    \end{minipage}%
    \begin{minipage}[b]{0.5\textwidth}
    \centering
    \includegraphics[scale=0.7]{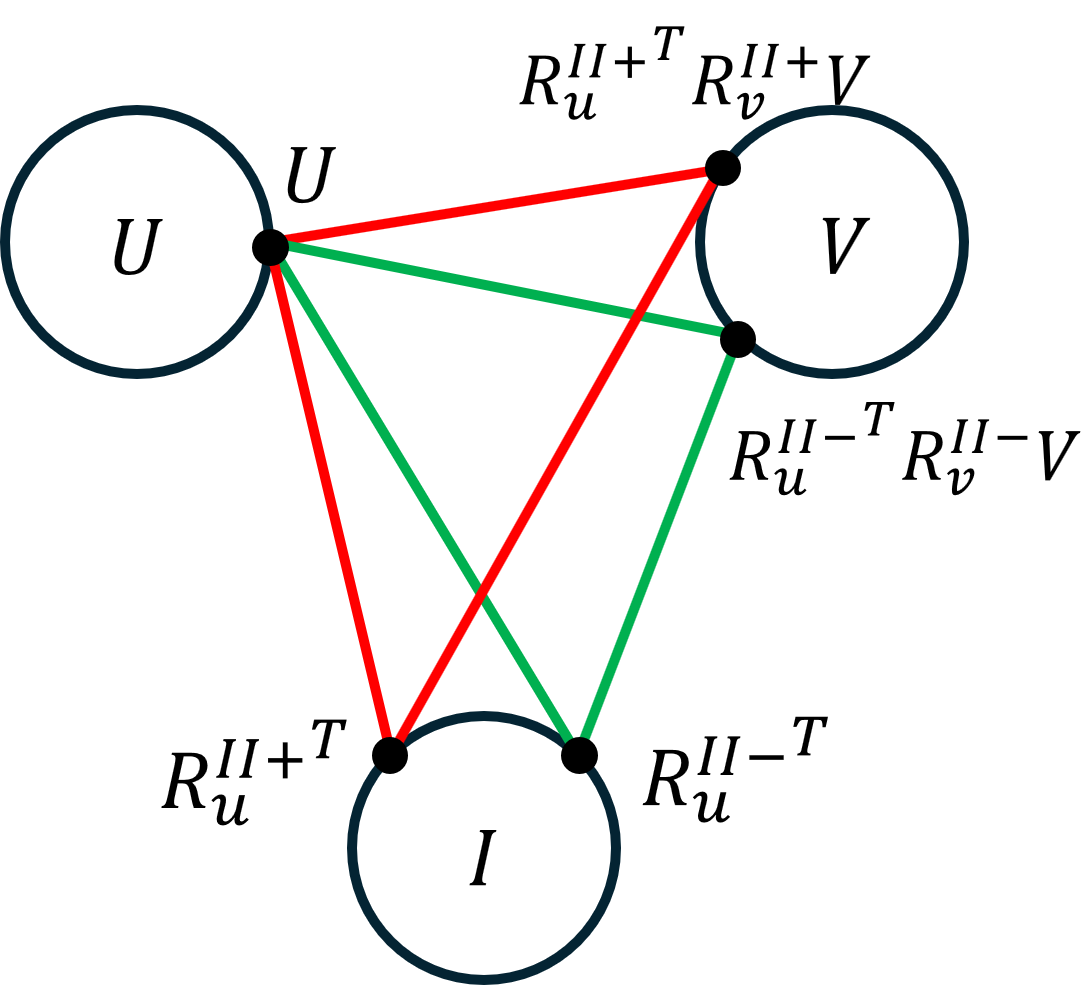}
    \subcaption{}
    \label{fig:bothsolutions_b}
    \end{minipage}
    \caption{The two distinct triple clusters in \autoref{thm:firstkind} and \autoref{thm:secondkind} imply that both the solutions of the twinning equation form triple clusters with austenite, (a) is obtained from \autoref{fig:firstkind_a} by multiplying the compatibility equations with appropriate rotations (b) is obtained from \autoref{fig:secondkind_a} in the same manner.}
    \label{fig:bothsolutions}
\end{figure}

\autoref{fig:bothsolutions_a} shows the compatibility equations \eqref{eq:46} and \eqref{eq:47}. Similarly, we can obtain the equations for compatibility shown in \autoref{fig:bothsolutions_b} by premultiplying \eqref{eq:43},\eqref{eq:44} by $(R_u^{II+})^T$ and \eqref{eq:45} by $(R_u^{II-})^T$. In both the cases, we see that both the solutions of twinning equation between $U$ and $V$ can form distinct triple clusters with austenite. This is not possible if $U$ and $V$ form Type I/II twins, as \autoref{thm:dellaportaR} does not allow Type I and Type II solution to form triple clusters \eqref{eq:CC2} simultaneously.
\end{remark}
\autoref{thm:firstkind} and \autoref{thm:secondkind} enable a pair of commuting compound domains to form two distinct triple clusters with austenite. An important question that arises is whether compound domains always give rise to exactly two distinct triple clusters, or do other possibilities exist. Moreover, what possibilities exist for compound domains that do not commute? These questions are addressed in detail in \autoref{sec:triple_compound}.

\subsection{Triple clusters with compound domains}
\label{sec:triple_compound}
Using the commutation property, we can construct triple clusters involving two compound domains and austenite. These compound domains may be the conventional compound twins or non-conventional compound domains. These triple clusters helps in eliminating the transition layer between austenite and martensite. We discuss this in detail in \autoref{sec:extreme_comp}. However, it must be noted that the commutation property doesn't encompass all cases of compound domains. For example, the conventional compound twins in cubic to monoclinic transformations (both I \& II),  do not commute, and hence are not covered under the commutation property. In fact, in all the transformations discussed above, which have monoclinic martensite, the conventional compound twins fail to commute. To address these more general cases, we derive the general conditions under which any pair of compound domains, whether or not they commute, can form triple clusters of the first or the second kind with austenite, see  \autoref{thm:compound_triple}. The conditions established in \autoref{thm:firstkind} and  \autoref{thm:secondkind} can be seen as special instances of these generalized criteria. We first prove a proposition, which we use later in the derivation .\\[5 pt]
\begin{proposition}
\label{prop:compound}
Let $U\in \mathbb{R}^{3 \times 3}_{+sym}$ be a variant of martensite such that $\lambda_1,\;\lambda_2=1,\; \lambda_3$ are the ordered eigenvalues of $U$. If $(R_i\:,b_i,\:\hat{m}_i)$, $i=1,2$ be the solutions of the equation $QU-I=b\:\otimes\:\hat{m}$, then, $R_1\hat{u}_2 = R_2\hat{u}_2 = \hat{u}_2$. Also, if $R_1$ is a rotation of $\theta$ about $\hat{u}_2$, then $R_2$ is a rotation of $-\theta$ about $\hat{u}_2$ such that $R_1R_2=R_2R_1=I$, where $\theta$ is given by

\begin{equation}
    \theta = \cos^{-1}\frac{1+\lambda_1\lambda_3}{\lambda_1 +\lambda_3} = \sin^{-1}\frac{\sqrt{(1-\lambda_1^2)(\lambda_3^2-1)}}{\lambda_1 +\lambda_3}.\nonumber
\end{equation}
\end{proposition} 
\begin{proof}
For $i=1,2$, we have $R_iU-I\:=\:b_i \otimes m_i$. Then,
\begin{gather}
\nonumber
(R_iU-I)\:\hat{u}_2=\:(b_i \otimes m_i)\:\hat{u}_2 \implies (R_iU-I)\:\hat{u}_2\:= (m_i.\hat{u}_2)\:b_i =0 \qquad \textit{using \eqref{eq:11}}\\[5 pt]
   \nonumber \implies R_iU_2\hat{u}_2= \hat{u}_2 \implies R_i\hat{u}_2 = \hat{u}_2 \implies\hat{u}_2 = R_i^T\hat{u}_2.
\end{gather}
This proves the first assertion. Now, we multiply the equation $R_iU-I\:=\:b_i \otimes m_i$, first by $\hat{u}_j$ and then take the dot product with $\hat{u}_k$ to get,
\begin{gather}
\lambda_jR_i\hat{u}_j.\hat{u}_k - \delta_{jk}= (\hat{b}_i.\hat{u}_k)(\hat{m}_i.\hat{u}_j)
\label{eq:48}
\end{gather}

From the expressions in \eqref{eq:11}, we find that \\[8 pt]
$\bullet \;(b_1.\hat{u}_1)=(b_2.\hat{u}_1)$
$\qquad \bullet \; (b_1.\hat{u}_3)=-(b_2.\hat{u}_3)$
$\qquad \bullet \;(m_1.\hat{u}_1)=(m_2.\hat{u}_1)$
$\qquad \bullet \;(m_1.\hat{u}_3)=-(m_2.\hat{u}_3)$\\[5pt]
where subscripts `1,2' are used for solutions corresponding to $\kappa=1,-1$ respectively. Using these relations, we can write,
\begin{equation}
    \begin{gathered}
      ( b_1.\hat{u}_1)(m_1.\hat{u}_1)=(b_2.\hat{u}_1)(m_2.\hat{u}_1) \implies R_1 \hat{u}_1.\hat{u}_1= R_2\hat{u}_1.\hat{u}_1, \\[5pt]
(b_1.\hat{u}_3)(m_1.\hat{u}_1) = -(b_2.\hat{u}_3)(m_2.\hat{u}_1) \implies R_1 \hat{u}_1.\hat{u}_3= - 
 \:R_2\hat{u}_1.\hat{u}_3, \\[5pt]
( b_1.\hat{u}_1)(m_1.\hat{u}_3)=-(b_2.\hat{u}_1)(m_2.\hat{u}_3) \implies R_1 \hat{u}_3.\hat{u}_1 =- R_2\hat{u}_3.\hat{u}_1, \\[5pt]
( b_1.\hat{u}_3)(m_1.\hat{u}_3)=(b_2.\hat{u}_3)(m_2.\hat{u}_3) \implies R_1 \hat{u}_3.\hat{u}_3 = R_2\hat{u}_3.\hat{u}_3.  \label{eq:49}
    \end{gathered}
\end{equation}
Let $\theta_1,\theta_2$ be the angles of rotation for the rotations $R_1,R_2$ respectively. Then, from equations $\eqref{eq:49}_1$ and $\eqref{eq:49}_4$, we have, $\cos{\theta_1} = \cos{\theta_2}$. Similarly, from equations $\eqref{eq:49}_2$ and $\eqref{eq:49}_3$, we have, $\sin{\theta_1} = -\sin{\theta_2}$
   $\implies\theta_1=-\theta_2=\theta \;(say)$. This means that the rotations $R_1$ and $R_2$ are equal in magnitude but opposite in sense with $\hat{u}_2$ as their common axis, \textit{see \autoref{fig:eigenbasisU}}. Then their composition must leave the space invariant (identity matrix). $\theta$ can be computed by observing that $\cos\theta = \sqrt{1-\sin^{2}\theta}= R_1\hat{u}_1.\hat{u}_1$ and using values from \eqref{eq:11}.
\end{proof} 

\begin{figure}
    \centering
    \begin{minipage}[b]{0.33\textwidth}
    \centering
    \includegraphics[scale=0.85]{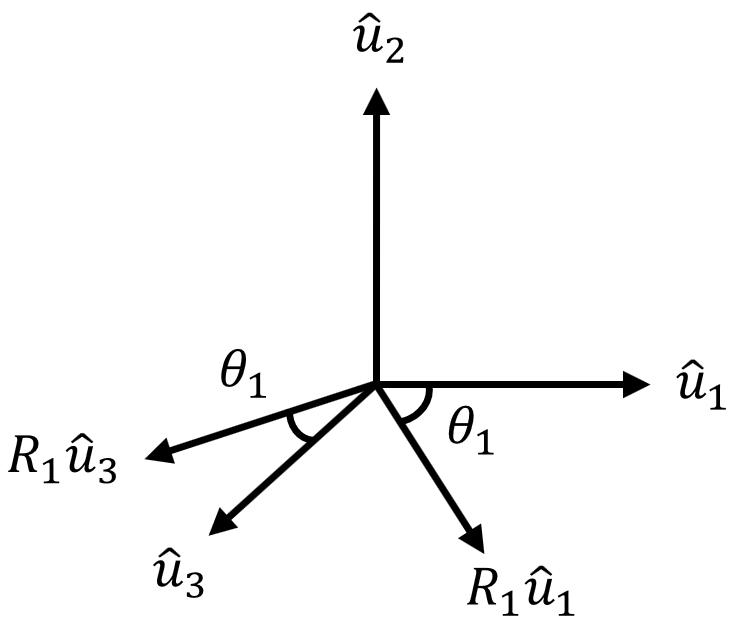}
    \subcaption{}
    \label{fig:eigenbasisU_a}
    \end{minipage}%
    \begin{minipage}[b]{0.33\textwidth}
    \centering
    \includegraphics[scale=0.85]{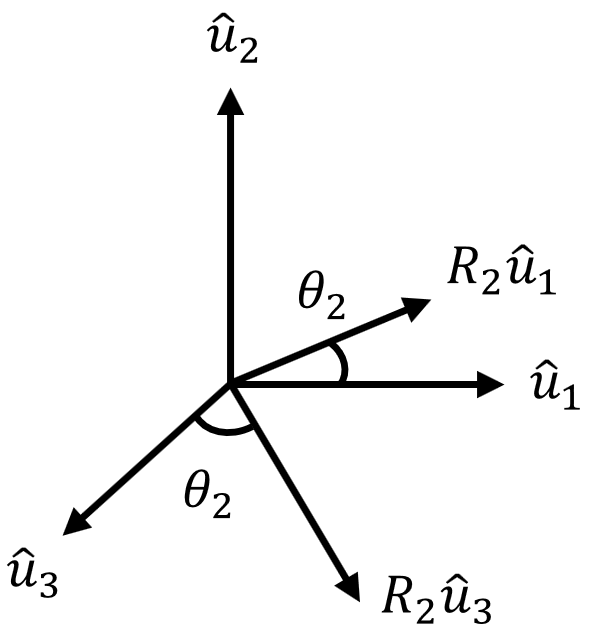}
    \subcaption{}
    \label{fig:eigenbasisU_b}
    \end{minipage}%
    \begin{minipage}[b]{0.33\textwidth}
    \centering
    \includegraphics[scale=0.85]{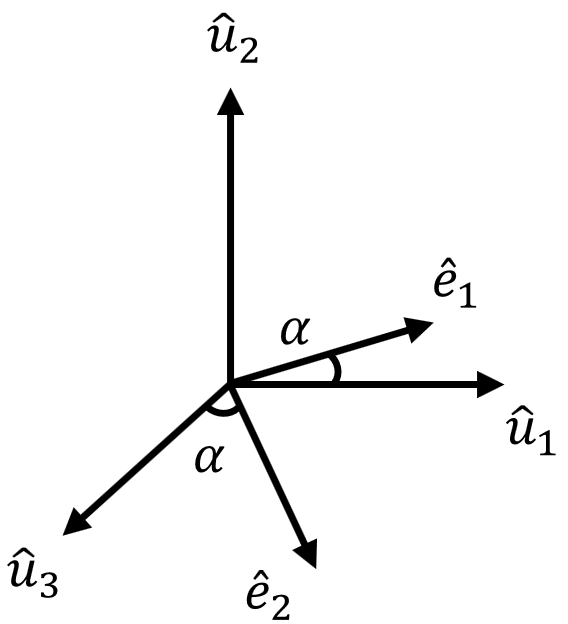}
    \subcaption{}
    \label{fig:eigenbasisU_c}
    \end{minipage}%
    \caption{(a,b) The effect of rotations $(R_1,R_2)$ respectively on the eigenbasis of $U$. $\hat{u}_2$ is the common axis of rotation for $R_1$ and $R_2$. The angles of rotation are equal but opposite in sense. i.e. $\theta_2=-\theta_1$ (c) Structure of the symmetry axes in the eigenbasis of $U$, when it forms a compound domain with another variant $V$.}
    \label{fig:eigenbasisU}
\end{figure}

\begin{lemma}
\label{lem:ruddock}
    (G.Ruddock, \cite{Ruddock1994-yq}) Suppose that $x_i$ and $y_i$ are non-zero for $i=1,2,3$ and that $\sum_{i=1}^{3}x_i \otimes\:y_i = 0$. Then both $\{x_1,x_2,x_3\}$ and $\{y_1,y_2,y_3\}$ are sets of coplanar vectors and at least one of them is a set of parallel vectors.
\end{lemma}
\begin{lemma}
\label{lem:commutationQ}
     Let $U$ and $V$ be two matrices that commute. Then for any $Q\in SO(3)$, $QUQ^T$ and $QVQ^T$ also commute.
\end{lemma}
\begin{proof}
    $(QUQ^T)(QVQ^T) = QUVQ^T = QVUQ^T = (QVQ^T)(QUQ^T).$ 
\end{proof}

\begin{theorem}
\label{thm:compound_triple}
    Let $U,V \in$ $\mathbb{R}^{3\times 3}_{+sym}$ be two variants of martensite \textit{(see \autoref{fig:compoundUV})} such that $U$ and $V$ form compound domains with symmetry axes $\hat{e}_1$ and $\hat{e}_2$. Let $\lambda_1<1,\; \lambda_2=1,\;\lambda_3>1,$ be the eigenvalues of $U$ and $\hat{u}_1,\;\hat{u}_2,\;\hat{u_3}$ be the corresponding eigenvectors such that $(\hat{u}_2,\;\hat{e}_1,\;\hat{e}_2)$ forms an orthonormal basis. If there exists a triple cluster between $U,V$ and $I$, then at least one of the following must hold.\\[5 pt]
 \begin{equation}
        \lambda_1\sqrt{\lambda_3^2-1}(\hat{u}_3.\hat{e}_1) = -\lambda_3\sqrt{1- \lambda_1^2}(\hat{u}_1.\hat{e}_1).\tag{EC1}
     \end{equation}

\begin{equation} 
\lambda_1\sqrt{\lambda_3^2-1}(\hat{u}_1.\hat{e}_1) = -\lambda_3\sqrt{1-\lambda_1^2}(\hat{u}_3.\hat{e}_1).\tag{EC2} \end{equation}
        
\begin{equation} 
\sqrt{1-\lambda_1^2}\:(\hat{u}_3.\hat{e}_1)= -\:\sqrt{\lambda_3^2-1}\:(\hat{u}_1.\hat{e}_1).\tag{EC3} \end{equation}
      
\begin{equation}
\sqrt{1-\lambda_1^2}\:(\hat{u}_1.\hat{e}_1)= -\:\sqrt{\lambda_3^2-1}\:(\hat{u}_3.\hat{e}_1).\tag{EC4} \end{equation}\\[5 pt]
\end{theorem} 
\begin{proof}
$U$ and $V$ form compound domains, so there exist unit vectors $\hat{e}_1,\hat{e}_2$ such that $\hat{e}_1\;\perp \;\hat{e}_2$ and $V=Q_1UQ_1^T=Q_2UQ_2^T$, where $Q_1 = -I + 2\:\hat{e}_1\otimes \hat{e}_1$ and $Q_2 = -I + 2\:\hat{e}_2\otimes \hat{e}_2$. Since the middle eigenvalue of $U$ is unity, rank-one connections between the variants and austenite exist. \\[10 pt]
\begin{minipage}[]{0.6\textwidth}
    Then the solutions to the equation of compatibility \eqref{eq:9} between $U$ and $V$ can be written using \autoref{thm:compound_domains}.
    \begin{gather}
    Q_1V-U=\zeta\: U\hat{e}_2 \otimes \: \hat{e_1} \label{eq:50},\\[4 pt]
    Q_2V-U=\eta\: U\hat{e}_1 \otimes  \:\hat{e_2} \label{eq:51}.
\end{gather}

where $\zeta$ and $\eta$ are given by \eqref{eq:14}. The rank-1 connections of $U$ with identity are given as
\begin{gather}
    U-R_1 = b^1_u \:\otimes \:\hat{m}^1_u \label{eq:52},\\[4pt]
    U-R_2 = b^2_u \:\otimes \:\hat{m}^2_u \label{eq:53}.
\end{gather}
Let the rank-1 connections between the wells $V$ and $I$ be as follows:
    \begin{gather}
    Q_1V-R_1 = b^1_v \:\otimes \:\hat{m}^1_v \label{eq:54},\\[4pt]
    Q_2V-R_2 = b^2_v \:\otimes \:\hat{m}^2_v \label{eq:55}.
    \end{gather}
\end{minipage}%
\hspace{0.09\textwidth}
    \begin{minipage}[]{0.3\textwidth}
    \includegraphics[scale=0.65]{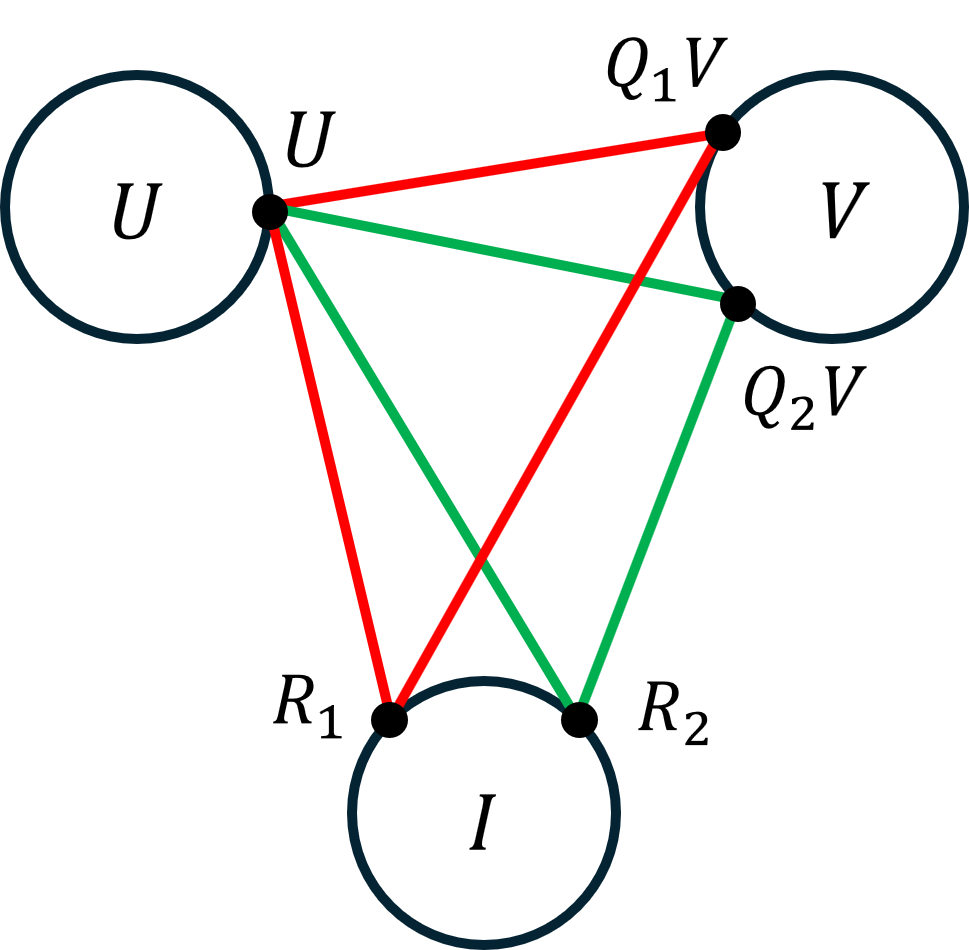}
    \captionof{figure} {$U$ and $V$ are compound domains and both the solutions between $U$ and $V$ form triple clusters with austenite. Red triangle represents equations \eqref{eq:50}, \eqref{eq:52} and \eqref{eq:54}. Green triangle represents equations \eqref{eq:51}, \eqref{eq:53} and \eqref{eq:55}.}
    \label{fig:compoundUV}
\end{minipage}\\[10 pt]
Assume that there exists a planar triple cluster, either the red one or the green one, as shown in \autoref{fig:compoundUV}. Then one of the following cases must hold. $(i)$ $(U,Q_1V,R_1)$ is a planar triple cluster of the first kind, $(ii)$ $(U,Q_2V,R_2)$ is a planar triple cluster of first kind, $(iii)$ $(U,Q_1V,R_1)$ is a triple cluster of second kind or $(iv)$  $(U,Q_2V,R_2)$ is a triple cluster of second kind. The equivalent mathematical expressions for these cases are
\begin{enumerate}[label=(\it{\roman*})]
    \item $\zeta\: U\hat{e}_2$ $||$ $b^1_u$ $||$ $b^1_v$ \: and \: $(\hat{e}_1,\;\hat{m}_u^1, \;\hat{m}_v^1)$ are coplanar.
    \item $\eta\: U\hat{e}_1$ $||$ $b^2_u$ $||$ $b^2_v$ \: and \: $(\hat{e}_2,\;\hat{m}_u^2, \;\hat{m}_v^2)$ are coplanar.
    \item $\hat{e}_1=\hat{m}_u^1=\hat{m}_v^1\;$ and $\;(\zeta U\hat{e}_2,\;b^1_u,\;b^1_v)$  are coplanar.
    \item $\hat{e}_2=\hat{m}_u^2=\hat{m}_v^2\;$ and $\;(\eta U\hat{e}_1,\;b^2_u,\;b^2_v)$  are coplanar.
\end{enumerate}

We now examine each of these cases individually, in the same order.
 
\subsection*{Case (i) Triple cluster of the first kind (red triangle indicated in \autoref{fig:compoundUV})}
For this case to be true, condition \textit{(i)} must hold. It suffices to establish $\zeta\: U\hat{e}_2$ $||$ $b^1_u$. Then the other requirements of condition \textit{(i)} are automatically satisfied by the structure of equations \eqref{eq:50},\eqref{eq:52},\eqref{eq:54} and \autoref{lem:ruddock}.\\[8pt]
$\zeta\: U\hat{e}_2$ $||$ $b^1_u$ $\implies$ $U\hat{e}_2 \times b^1_u =0$. The expression for $b^1_u$ can be obtained by converting \eqref{eq:52} into the standard form (\textit{by premultiplying with $R_1^T$}) and applying \autoref{corollary:identity_solutions}. Here, we take $R_1^T,R_2^T$ to be the rotations corresponding to the solutions $\kappa=+1$ and $\kappa=-1$ respectively. Then, using the spectral decomposition of $U$, we rewrite condition 1 as 
\begin{gather}\nonumber
    (\lambda_1\hat{u}_1 \otimes\hat{u}_1 + \hat{u}_2\otimes \hat{u}_2 + \lambda_3\hat{u}_3\otimes \hat{u}_3)\:\hat{e}_2\; \times \; \frac{\rho}{\sqrt{\lambda_3^2 - \lambda_1^2}}\left(\lambda_3\sqrt{1-\lambda_1^2}\;R_1\hat{u}_1\; +\lambda_1\sqrt{\lambda_3^2 -1}\;\:R_1\hat{u}_3 \right) = 0.
 \end{gather}

 \begin{align}
    \nonumber\implies & \frac{\rho}{\sqrt{\lambda_3^2 - \lambda_1^2}} 
    \{ \lambda_1\lambda_3\sqrt{1-\lambda_1^2}\:(\hat{u}_1.\hat{e}_2)\; \hat{u}_1 \times R_1\hat{u}_1 
    \;\; +\: \lambda_1^2 \sqrt{\lambda_3^2-1}\:(\hat{u}_1.\hat{e}_2)\; \hat{u}_1 \times R_1\hat{u}_3 \\[5pt]
     & +\;  \lambda_3^2\sqrt{1-\lambda_1^2}\:(\hat{u}_3.\hat{e}_2)\; \hat{u}_3 \times R_1\hat{u}_1 +
    \; \lambda_1 \lambda_3\sqrt{\lambda_3^2-1}\:(\hat{u}_3.\hat{e}_2)\; \hat{u}_3 \times R_1\hat{u}_3 \} = 0
    \label{eq:56}.
\end{align}%
The expression involves cross products of $\hat{u}_i$ and $R_1\hat{u}_j$, $i,j=\{1,3\}$. To evaluate these terms, we use \autoref{prop:compound} and \textit{\autoref{fig:eigenbasisU}} to get,\\[7pt]
$\bullet \;\;\hat{u}_1 \times R_1\hat{u}_1 =- \sin\theta\:\hat{u}_2$ $\qquad\quad$
$\bullet \;\;\hat{u}_1 \times R_1\hat{u}_3 =- \cos\theta\:\hat{u}_2$ $\qquad\quad$
$\bullet \;\;\hat{u}_3 \times R_1\hat{u}_1 = \cos\theta\:\hat{u}_2$ \\
$\bullet \;\;\hat{u}_3 \times R_1\hat{u}_3 = -\sin\theta\:\hat{u}_2$\\[5pt]
Substituting these expressions in \eqref{eq:56} and using the value of $\theta$ from \autoref{prop:compound}, we get
\begin{equation}
    \lambda_1\sqrt{\lambda_3^2-1}(\hat{u}_1.\hat{e}_2) = \lambda_3\sqrt{1-\lambda_1^2}(\hat{u}_3.\hat{e}_2).\label{eq:57}
\end{equation}
From \autoref{fig:eigenbasisU_c}, we observe that\\[5pt] 
$\bullet \;\;\hat{u}_1.\hat{e}_1 = \hat{u}_3.\hat{e}_2$ $\qquad\quad$
$\bullet \;\;-\hat{u}_1. \hat{e}_2 =\;\hat{u}_3. \hat{e}_1$\\[5pt]
Rewriting \eqref{eq:57}\\
\begin{equation}
    \lambda_1\sqrt{\lambda_3^2-1}(\hat{u}_3.\hat{e}_1) = -\lambda_3\sqrt{1-\lambda_1^2}(\hat{u}_1.\hat{e}_1).\label{eq:58}
    \tag{EC1}
\end{equation}

\subsection*{Case (ii) Triple cluster of the first kind (green triangle indicated in \autoref{fig:compoundUV})}

To obtain an equivalent expression for this case , we replace $R_1$ by $R_2$ in \eqref{eq:56} and multiply the second and fourth terms by $-1$. After simplification, we can show that condition \textit{(ii)} is satisfied if

\begin{equation}
\label{eq:59}
    \lambda_1\sqrt{\lambda_3^2-1}(\hat{u}_1.\hat{e}_1) = -\lambda_3\sqrt{1-\lambda_1^2}(\hat{u}_3.\hat{e}_1).
    \tag{EC2}
\end{equation}

\subsection*{Case (iii) Triple cluster of the second kind (red triangle indicated in \autoref{fig:compoundUV})}
 For this case, condition \textit{(iii)} must hold. Again, we impose $\hat{e}_1 \times \hat{m}^1_u = 0$,  then the other requirements of condition 3 are automatically satisfied by the structure of equations \eqref{eq:50}, \eqref{eq:52}, \eqref{eq:54} and \autoref{lem:ruddock}. 

\begin{equation}
    \hat{e}_1 \times \hat{m}^1_u = 0 \implies 
      \frac{\lambda_3 - \lambda_1}{\rho\sqrt{\lambda_3^2 - \lambda_1^2}}\left(-\sqrt{1-\lambda_1^2}\;\;\hat{e}_1 \times \hat{u}_1\;  +\;\sqrt{\lambda_3^2 -1}\;\;\hat{e}_1 \times\hat{u}_3 \right) =0.\label{eq:60}
\end{equation}
From \autoref{fig:eigenbasisU_c}, we have,\\[5 pt]
$\bullet \;\;\hat{e}_1 \times \hat{u}_1 = \sin\alpha \:\hat{u}_2 = (\hat{e_1}.\hat{u}_3)\;\hat{u}_2$ \qquad\quad
$\bullet \;\;\hat{e}_1 \times \hat{u}_3 = -\cos\alpha \:\hat{u}_2 = -(\hat{e_1}.\hat{u}_1)\;\hat{u}_2$ \\[5pt]
Substituting in \eqref{eq:60} and simplifying, we get
\begin{equation}
    \sqrt{1-\lambda_1^2}\:(\hat{u}_3.\hat{e}_1)= -\:\sqrt{\lambda_3^2-1}\:(\hat{u}_1.\hat{e}_1).\label{eq:61}
    \tag{EC3}
\end{equation}
\subsection*{Case (iv) Triple cluster of the second kind (green triangle indicated in \autoref{fig:compoundUV})}
For condition \textit{(iv)}, we must have
$\hat{e}_2 \times \hat{m}^2_u = 0$. 
\begin{equation}
    \hat{e}_2 \times \hat{m}^2_u = 0 \implies 
      \frac{\lambda_3 - \lambda_1}{\rho\sqrt{\lambda_3^2 - \lambda_1^2}}\left(-\sqrt{1-\lambda_1^2}\;\;\hat{e}_2 \times \hat{u}_1\;  -\;\sqrt{\lambda_3^2 -1}\;\;\hat{e}_2 \times\hat{u}_3 \right) =0.\label{eq:62}
\end{equation}
From \autoref{fig:eigenbasisU_c}, we have,\\[5 pt]
$\bullet \;\;\hat{e}_2 \times \hat{u}_1 = \cos\alpha \:\hat{u}_2 = (\hat{e_2}.\hat{u}_3)\;\hat{u}_2 =\:(\hat{e_1}.\hat{u}_1)\;\hat{u}_2$ $\qquad\quad$
$\bullet \;\;\hat{e}_2 \times \hat{u}_3 = \sin\alpha \:\hat{u}_2 = -(\hat{e_2}.\hat{u}_1)\;\hat{u}_2 = \: (\hat{e_1}.\hat{u}_3)\;\hat{u}_2$ \\[5pt]
Substituting in \eqref{eq:62} and simplifying, we get 
\begin{equation}
    \sqrt{1-\lambda_1^2}\:(\hat{u}_1.\hat{e}_1)= -\:\sqrt{\lambda_3^2-1}\:(\hat{u}_3.\hat{e}_1).\label{eq:63}
    \tag{EC4}
\end{equation}
\end{proof}
\begin{corollary}
    \label{cor:45degrees}
If the red and the green triple clusters coexist simultaneously and are of the same kind, then $(\hat{u}_1.\hat{e}_1)=-(\hat{u}_3.\hat{e}_1)$ must hold. Furthermore, when $(\hat{u}_1.\hat{e}_1)=-(\hat{u}_3.\hat{e}_1)$ holds, the variants $U$ and $V$ commute.
\end{corollary} 
\begin{proof}  
Two situations arise. First is that both red and green triple clusters are of the first kind, for which equations \eqref{eq:58} and \eqref{eq:59} must hold simultaneously. Then, $\eqref{eq:58}+\eqref{eq:59}$ and $\eqref{eq:58}-\eqref{eq:59}$ must hold as well. $\eqref{eq:58}+\eqref{eq:59}$ implies,

\begin{equation}
    (\lambda_1\sqrt{\lambda_3^2-1}  -\lambda_3\sqrt{1-\lambda_1^2}\;)\;(\hat{u}_3.\hat{e}_1)=(\lambda_1\sqrt{\lambda_3^2-1}  -\lambda_3\sqrt{1-\lambda_1^2}\;)\;(\hat{u}_1.\hat{e}_1).\label{eq:64}
\end{equation}

Similarly, $\eqref{eq:58}-\eqref{eq:59}$ implies,

\begin{equation}
     (\lambda_1\sqrt{\lambda_3^2-1}  +\lambda_3\sqrt{1-\lambda_1^2}\;)\;(\hat{u}_3.\hat{e}_1)=-(\lambda_1\sqrt{\lambda_3^2-1}  +\lambda_3\sqrt{1-\lambda_1^2}\;)\;(\hat{u}_1.\hat{e}_1).\label{eq:65}
\end{equation}

If we assume the expressions multiplying the dot product terms in \eqref{eq:64} and \eqref{eq:65} to be non-zero, then  $(\hat{u}_1.\hat{e}_1)=(\hat{u}_3.\hat{e}_1)$ and $(\hat{u}_1.\hat{e}_1)=-(\hat{u}_3.\hat{e}_1)$ must both be true. This is possible only if both the dot products are individually equal to zero and $\hat{e}_1$ is perpendicular to both $\hat{u}_1$ and $\hat{u}_3$. Then, $\hat{e}_1$ must be parallel to $\hat{u}_2$, which further implies that $(-I+2\;\hat{u}_2\otimes\hat{u}_2)\;U(-I+2\;\hat{u}_2\otimes\hat{u}_2) = U$. This contradicts the fact that $U$ and $V$ are distinct. This shows that our assumption is wrong and the coefficients of \eqref{eq:64} and \eqref{eq:65} must not be all non-zero. Since, the eigenvalues of $U$ are all positive and distinct, the coefficients of \eqref{eq:65} are all positive quantities. Then, the coefficients of \eqref{eq:64} must be zero. Therefore, for equations \eqref{eq:58} and \eqref{eq:59} to hold simultaneously, we must have, 

\begin{equation}
   \lambda_1\sqrt{\lambda_3^2-1}  =\lambda_3\sqrt{1-\lambda_1^2}\;\quad \text{and} \quad(\hat{u}_3.\hat{e}_1)= -\: (\hat{u}_3.\hat{e}_1).\label{eq:66}
\end{equation}

Rewriting $\eqref{eq:66}_1$ to express $\lambda_3$ as an explicit function of $\lambda_1$  yields \eqref{eq:27}.\\[5pt]
The second situation arises if both red and green triple clusters are of the second kind, for which equations \eqref{eq:61} and \eqref{eq:63} must hold simultaneously. By a similar analysis as performed for the first situation, it can be shown that equations \eqref{eq:61} and \eqref{eq:63} hold simultaneously only if 

\begin{equation}
   \sqrt{\lambda_3^2-1}  =\sqrt{1-\lambda_1^2}\;\quad \text{and} \quad(\hat{u}_3.\hat{e}_1)= -\: (\hat{u}_1.\hat{e}_1).\label{eq:67}
\end{equation}

Rewriting $\eqref{eq:67}_1$ to express $\lambda_3$ as an explicit function of $\lambda_1$  yields \eqref{eq:38}.\\[5 pt]
\eqref{eq:66} and \eqref{eq:67} show that for both the twinning solutions to form triple clusters of the same kind $(\hat{u}_3.\hat{e}_1)= -\: (\hat{u}_1.\hat{e}_1)$ is a necessary condition. Then depending on whether \eqref{eq:27} or \eqref{eq:38} is additionally satisfied, triple clusters of either the first or the second kind are possible. This proves the first assertion of \autoref{cor:45degrees}. We now proceed to prove the second assertion, which claims that under the hypothesis of \autoref{thm:compound_triple}, if $(\hat{u}_3.\hat{e}_1)= -\: (\hat{u}_1.\hat{e}_1)$, then $U$ and $V$ commute. \\[5 pt]
First, we note that the condition $(\hat{u}_3.\hat{e}_1)= -\: (\hat{u}_1.\hat{e}_1)$ implies $\alpha=\frac{\pi}{4}$ in \autoref{fig:eigenbasisU_c}. Then, we can write, $(\hat{u}_3.\hat{e}_1)= -\: (\hat{u}_1.\hat{e}_1)=-\frac{1}{\sqrt{2}}.$ \\[5 pt]
Now, the middle eigenvalue of $U$ is 1, therefore, we have $U\hat{u}_2=\hat{u}_2\implies\;Q_1VQ_1^T\hat{u}_2=\hat{u}_2$.\\
\begin{equation}
    \implies (-I + 2\;\hat{e}_1\otimes\hat{e}_1)\;V (-I + 2\;\hat{e}_1\otimes\hat{e}_1)=\hat{u_2}\;\;\implies V\hat{u}_2-2(V\hat{e}_1.\hat{u}_2)\:\hat{e}_1=\hat{u}_2.\label{eq:68}
\end{equation}

Similarly, we can write, $Q_2VQ_2^T\hat{u}_2=\hat{u}_2$,
\begin{equation}
   \implies (-I + 2\;\hat{e}_2\otimes\hat{e}_2)\;V (-I + 2\;\hat{e}_2\otimes\hat{e}_2)=\hat{u_2}\;\;\implies V\hat{u}_2-2(V\hat{e}_2.\hat{u}_2)\:\hat{e}_2=\hat{u}_2.\label{eq:69}
\end{equation}
From \eqref{eq:68} and \eqref{eq:69}, we find $(V\hat{e}_1.\hat{u}_2) = (V\hat{e}_2.\hat{u}_2) = 0$. Substituting back in \eqref{eq:68}, we see that $V\hat{u}_2 = \hat{u}_2,$ implying that $\hat{u}_2$ is an eigenvector of $V$ corresponding to the eigenvalue $\lambda_2=1$.\\[5pt]
Again, $Q_1UQ_1^T=V$. Writing $U$ and $V$ in their spectral decomposition, we have, 
\begin{equation}  
(I+2\hat{e}_1\otimes\hat{e}_1)\:(\lambda_1\hat{u}_1\otimes\hat{u}_1+\hat{u}_2\otimes\hat{u}_2 +\lambda_3\hat{u}_3\otimes\hat{u}_3)\;(I+2\hat{e}_1\otimes\hat{e}_1)\;= \lambda_1\hat{v}_1\otimes\hat{v}_1+\hat{u}_2\otimes\hat{u}_2 +\lambda_3\hat{v}_3\otimes\hat{v}_3. \label{eq:70}
\end{equation}
Expanding \eqref{eq:70} and substituting $(\hat{u}_3.\hat{e}_1)= -\: (\hat{u}_1.\hat{e}_1)=-\frac{1}{\sqrt{2}}, $ we get,
\begin{equation}
    \implies\lambda_1\{\hat{u}_1\otimes\hat{u}_1-\sqrt{2}\:(\hat{u}_1\otimes\hat{e}_1 \;+\;\hat{e}_1\otimes\hat{u}_1) \:+2\:\hat{e}_1\otimes\hat{e}_1\} + \lambda_3\{\:\hat{u}_3\otimes\hat{u}_3+\sqrt{2}\:(\hat{u}_3\otimes\hat{e}_1 \;+\;\hat{e}_1\otimes\hat{u}_3) \:+2\:\hat{e}_1\otimes\hat{e}_1\} 
    \notag
\end{equation}
\begin{equation}
    =\;\lambda_1\;\{\hat{v}_1 \otimes\hat{v}_1\} + \lambda_3\;\{\hat{v}_3 \otimes\hat{v}_3\}.\label{eq:71}
\end{equation}

\begin{equation}
    \implies\lambda_1\{(\hat{u}_1-\sqrt{2}\:\hat{e}_1)\otimes (\hat{u}_1-\sqrt{2}\:\hat{e}_1)\} + \lambda_3\{(\hat{u}_3+\sqrt{2}\:\hat{e}_1)\otimes (\hat{u}_3+\sqrt{2}\:\hat{e}_1)\} 
    \notag
\end{equation}
\begin{equation}
    =\;\lambda_1\;\{\hat{v}_1 \otimes\hat{v}_1\} + \lambda_3\;\{\hat{v}_3 \otimes\hat{v}_3\}.\label{eq:72}
\end{equation}

Comparing the coefficients of $\lambda_1$ and $\lambda_3$ on both sides, we have,

\begin{equation}
    (\hat{u}_1-\sqrt{2}\:\hat{e}_1)=\pm \hat{v}_1 \quad\text{and}\quad(\hat{u}_3+\sqrt{2}\:\hat{e}_1)=\pm \hat{v}_3.\label{eq:73}
\end{equation}

Using $\hat{e_1}= \frac{1}{\sqrt{2}}(\hat{u}_1-\hat{u}_3)$, we have,

\begin{equation}
    \hat{u}_3=\pm\hat{v_1}\quad \text{and} \quad\hat{u}_1=\pm \hat{v}_3.\label{eq:74}
\end{equation}

The same set of eigenvectors (unique up to sign) diagonalizes both $U$ and $V$. Then by \autoref{lem:commutation}, $U$ and $V$ commute.
\end{proof}
\begin{remark}
\label{rem:45_O_NCM}
    If $d=1$ holds for compound twins in cubic to orthorhombic  transformations in \eqref{eq:19}, then the condition $(\hat{u}_3.\hat{e}_1)= -\: (\hat{u}_1.\hat{e}_1)$ is automatically satisfied. In addition, if \eqref{eq:27} or \eqref{eq:38} holds, then \autoref{thm:firstkind} or \autoref{thm:secondkind} becomes applicable, respectively. For non-conventional compound domains in cubic to monoclinic-II transformations in \eqref{eq:22}, the condition $(\hat{u}_3.\hat{e}_1)= -\: (\hat{u}_1.\hat{e}_1)$ is automatically satisfied when $d=1$ and either \eqref{eq:27} or \eqref{eq:38} holds. Furthermore, in either of these transformations, if $d=1$ and \eqref{eq:27} is satisfied, then \eqref{eq:58} and \eqref{eq:59} hold simultaneously and the two triple clusters are of the first kind. Similarly, if $d=1$ and \eqref{eq:38} is satisfied, then \eqref{eq:61} and \eqref{eq:63} hold simultaneously and the two triple clusters are of the second kind. 
\end{remark}
\begin{corollary}
\label{cor:first_second_impossible}
    It is not possible to have one triple cluster of the first kind and another of the second kind simultaneously between $U$, $V$, and $I$. In other words, if both the red and green triple clusters shown in \autoref{fig:compoundUV} exist, they must be of the same kind.
\end{corollary}

\begin{proof}
Two situations arise. First, red triple cluster is of first kind and green triple cluster is of the second kind. In this case, \eqref{eq:58} and \eqref{eq:63} must hold simultaneously. The second possible situation is when green triple cluster is of the first kind and red triple cluster is of the second kind. In this case, \eqref{eq:59} and \eqref{eq:61} must hold simultaneously.\\[5pt]
In both the situations, it can be shown that the simultaneous satisfaction of the pair of equations requires either $\lambda_1=\lambda_3$ or $\lambda_1\lambda_3=-\frac{1}{2}$. The former is a case of degeneracy and can be ignored. The latter is impossible, since $U$ is positive definite. Therefore, having two different kinds of triple clusters for the same pair is not possible.  
\end{proof}

\begin{corollary}
\label{cor:first_second_simul}
    If a single triple cluster (either red or green) is of the first kind as well as the second kind simultaneously, then it is necessary that $det\:(U)=1$.
\end{corollary}
\begin{proof}
There are two scenarios under which such a situation can arise. The first is when \textit{Case (i)} and \textit{Case (iii)} are both true. In this situation, the pair of equations \eqref{eq:58} and \eqref{eq:61} must hold simultaneously. In the second scenario, \textit{Case (ii)} and \textit{Case (iv)} are both true. In this case, the pair of equations \eqref{eq:59} and \eqref{eq:63} must be satisfied at the same time. \\[5 pt]
In either scenario, for the pairs of equations to hold simultaneously, it is necessary that 

\begin{equation}
      \lambda_1\frac{\sqrt{\lambda_3^2-1}}{\sqrt{1-\lambda_1^2}}= \lambda_3\frac{\sqrt{1-\lambda_1^2}}{\sqrt{\lambda_3^2-1}}.
      \notag
    \end{equation}
    \begin{equation}
      \implies \lambda_1\lambda_3=1.\label{eq:75}
    \end{equation}

Since $\lambda_2=1,$ \eqref{eq:75} implies $det\;(U)=1.$
\end{proof}
\begin{remark}
    \label{rem:commutation_det_one_impossible}
When a pair of compound domains forms two distinct triple clusters of the same kind with austenite, then the 
 two triple clusters are either of first kind or second kind, but cannot be both at the same time. This is because for such a situation to hold, either equations \eqref{eq:27} and \eqref{eq:75} must hold simultaneously or equations \eqref{eq:38} and \eqref{eq:75} must hold simultaneously. In either case, this leads to a degeneracy, where $\lambda_1=\lambda_2=\lambda_3=1$, implying no transformation at all. Therefore, we conclude that if a triple cluster is of the first as well as the second kind simultaneously, then only one of the twinning solutions between compound domains can form triple clusters with austenite. For example, in cubic to orthorhombic transformations, every compound twin pair forms two distinct triple clusters with austenite, so they cannot be of the first as well as the second kind simultaneously.
\end{remark}
\begin{remark}
    \label{rem:different_compound_tendancy}
All compound domains do not have the same tendency to form triple clusters with austenite. There are compound domains in which \autoref{cor:45degrees} holds and are therefore capable of forming two distinct planar triple clusters with austenite for a given pair of compound domains. Examples of this kind of domains are the compound twins in cubic to orthorhombic transformation and non-conventional compound domains in cubic to monoclinic-II transformation, under $d=1$. There are other compound domains, e.g. the compound twins in cubic to monoclinic-II transformations, which do not obey \autoref{cor:45degrees} and in which, at most, one solution of the twinning equation can form a triple cluster with austenite. This is probably the reason that no generalized statement about the triple clusters involving compound domains is made in \cite{CHEN20132566}.
\end{remark}

\begin{figure}[hbt!]
    \centering
    \includegraphics[scale=0.55]{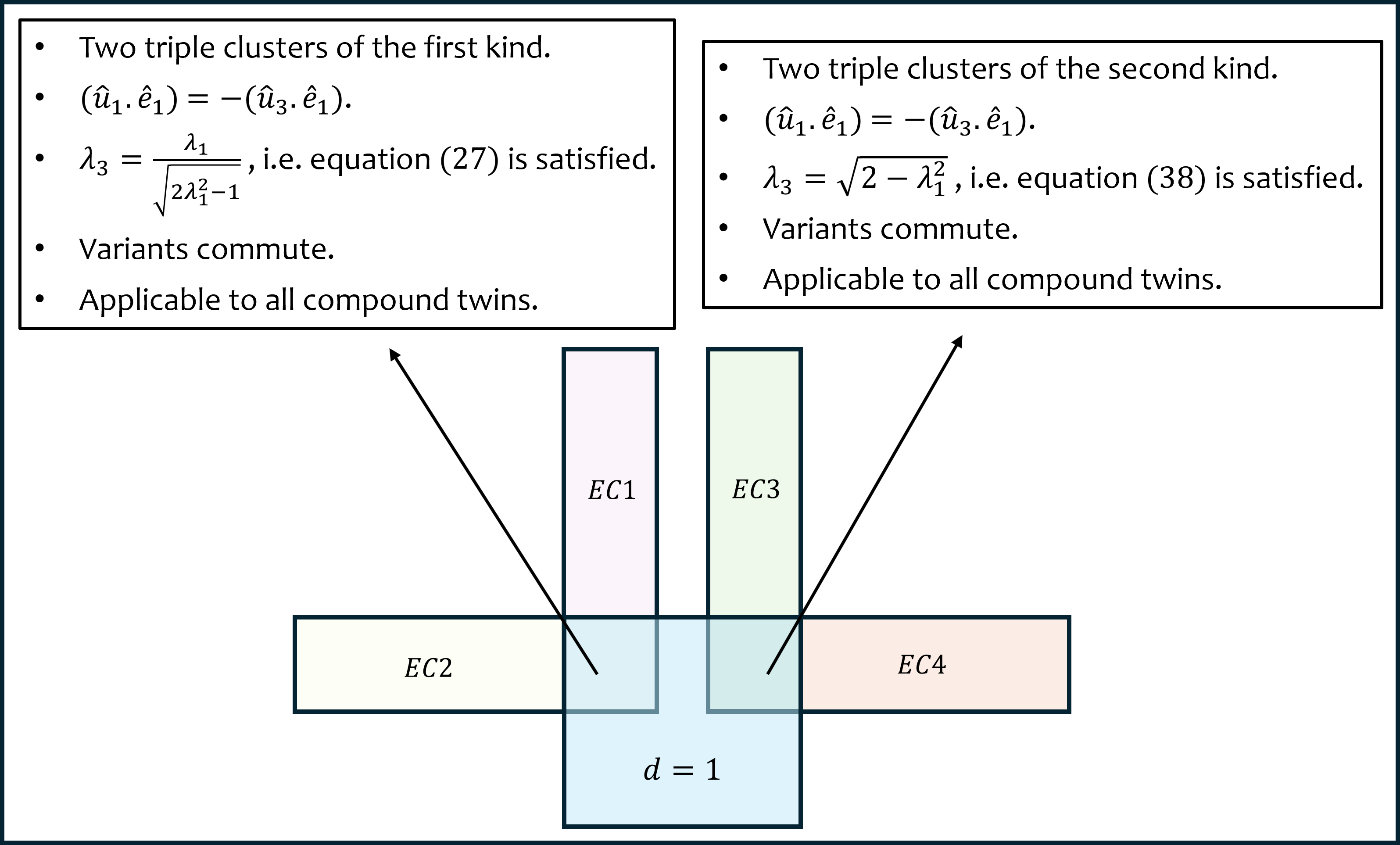}
   \caption{Summary of results obtained in \autoref{sec:triple_beyond} for cubic to orthorhombic transformations. When $d=\lambda_2=1$ is satisfied then for all compound twin pairs, $(\hat{u}_2,\hat{e}_1,\hat{e}_2)$ forms an orthonormal basis. The Venn diagram illustrates regions where the stated conditions hold. Overlapping regions represent simultaneous satisfaction of the associated conditions. See main text for details.}
       \label{fig:3ortho}
\end{figure}

\begin{figure}[hbt!]
    \centering
    \includegraphics[scale=0.45]{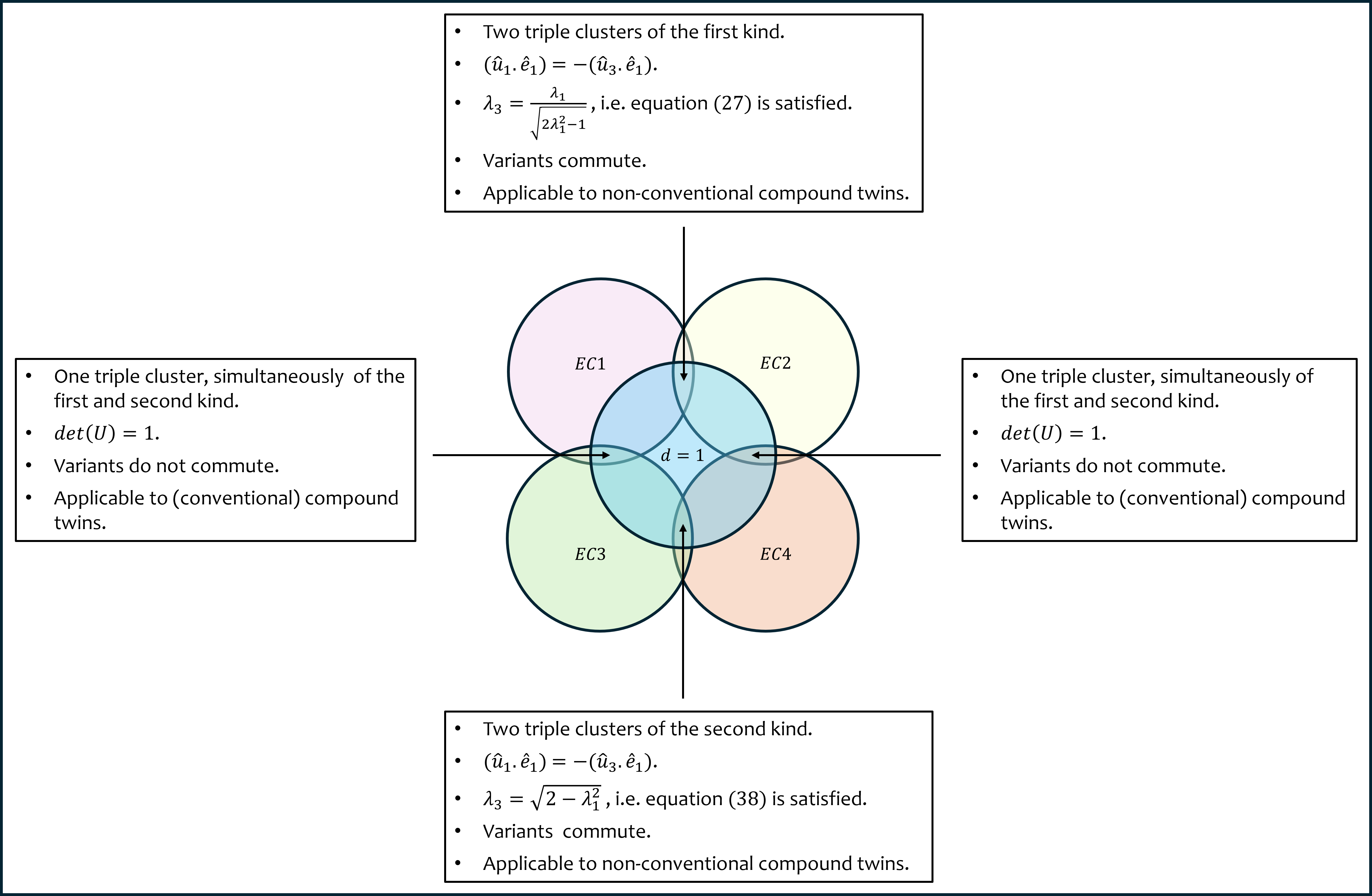}
   \caption{Summary of results obtained in \autoref{sec:triple_beyond} for cubic to monoclinic-II transformations. When $d=\lambda_2=1$ is satisfied then for all compound domains, $(\hat{u}_2,\hat{e}_1,\hat{e}_2)$ forms an orthonormal basis. The Venn diagram illustrates regions where the stated conditions hold. Overlapping regions represent simultaneous satisfaction of the associated conditions. See main text for details.}
       \label{fig:3mono}
\end{figure}

\Cref{fig:3ortho,fig:3mono} summarize the main results of \autoref{sec:triple_beyond}, as applied to cubic to orthorhombic and cubic to monoclinic-II transformations, respectively, using Venn diagrams. For both these transformations, when \( d = 1 \) holds, the vectors \( (\hat{u}_2, \hat{e}_1, \hat{e}_2) \) form an orthonormal basis of \( \mathbb{R}^3 \) for all compound domains, therefore satisfying the hypothesis of \autoref{thm:compound_triple}.\\[5pt]
For cubic to orthorhombic transformations with $d=1$, \eqref{eq:58} and \eqref{eq:59} are simultaneously satisfied and together imply the satisfaction of \eqref{eq:27}. Similarly, \eqref{eq:61} and \eqref{eq:63} are simultaneously satisfied and together imply the satisfaction of \eqref{eq:38}. This situation is represented by the regions common to \eqref{eq:58}, \eqref{eq:59} and \eqref{eq:61}, \eqref{eq:63} in \autoref{fig:3ortho} and correspond to the results of \autoref{cor:45degrees}, where the two compound domains form two distinct triple clusters of the same kind with austenite.\\[5pt]
For cubic to monoclinic-II transformations with \( d = 1 \), simultaneous satisfaction of the pairs of equations~\eqref{eq:58}, \eqref{eq:59} or~\eqref{eq:61}, \eqref{eq:63} happens only for non-conventional twins. In this case, either \eqref{eq:27} or \eqref{eq:38} is satisfied, leading to the formation of two distinct triple clusters of the same kind, similar to the case of cubic to orthorhombic transformations. This case is represented by the upper and lower boxes of \autoref{fig:3mono}, in accordance with \autoref{cor:45degrees}. The regions corresponding to \eqref{eq:58} and \eqref{eq:63} are disjoint, as are the regions corresponding to \eqref{eq:59} and \eqref{eq:61}. This implies that two distinct triple clusters of different kinds cannot coexist$-$an observation established in \autoref{cor:first_second_impossible}. Finally, the left and right black boxes represent the intersections of regions \eqref{eq:58} with \eqref{eq:61}, and \eqref{eq:59} with \eqref{eq:63}, respectively. These intersections correspond to conditions under which only one triple cluster is possible between two compound domains and austenite, which is simultaneously of the first and second kind, as outlined in \autoref{cor:first_second_simul}. Such a configuration is not possible for cubic to orthorhombic transformations but may arise in cubic to monoclinic-II transformations.

In this section, we have analyzed various possibilities for triple clusters between compound domains and austenite. We now proceed to examine how these relate to the formation of triple clusters between Type I/II domains and austenite. This connection leads to the concept of \emph{extreme compatibility}, discussed in \autoref{sec:extreme_comp}.

\section {Conditions of Extreme Compatibility}
\label{sec:extreme_comp}
The strongest known conditions of geometric compatibility between the phases are the cofactor conditions, which have been comprehensively studied by Chen et al \cite{CHEN20132566}. Cofactor conditions consist of three sub-conditions - $\eqref{eq:CC1}$, $\eqref{eq:CC2}$ and $\eqref{eq:CC3}$. While $\eqref{eq:CC1}$ and $\eqref{eq:CC2}$ provide the necessary conditions, $\eqref{eq:CC3}$ serves as the sufficiency condition for the existence of stress-free interfaces between Type I/Type II laminates and austenite phase for all volume fractions. It is important to note that the satisfaction of the cofactor conditions does not guarantee the formation of stress-free interfaces between compound laminates and austenite.\\[5pt]
To address this limitation of the cofactor conditions, we have derived the conditions for the existence of triple clusters between a pair of compound domains and austenite in \autoref{sec:triple_beyond}. It is important to emphasize that the existence of such triple clusters constitutes only a necessary condition for the realization of stress-free interfaces between austenite and the laminates formed by compound domains, for all volume fractions $\mu\:\in\:[0,1]$. The corresponding sufficiency condition for achieving stress-free interfaces, in this context, is once again provided by $\eqref{eq:CC3}$, as outlined by Chen et al \cite{CHEN20132566}. The overarching goal is to enable the participation of all twinning systems in generating stress-free interfaces between the phases. Due to the absence of established terminology, we shall refer to these conditions as \textit{extreme compatibility}. The extreme compatibility conditions define the criteria under which the laminates of Type I/II twins as well as compound twins can simultaneously maintain compatibility with austenite, for all volume fractions, $\mu\in[0,1]$, without the need of any transition layer at the phase interface. As highlighted in \autoref{rem:different_compound_tendancy}, this situation varies from one transformation to another. Consequently, we have different extreme compatibility conditions for different transformations.\\[5 pt]
We investigate extreme compatibility for two main types of transformations, \textit{(i)} Cubic to orthorhombic and \textit{(ii)} Cubic to monoclinic-II. These transformations have been studied in the context of cofactor conditions by Chen et al \cite{CHEN20132566}, and Della Porta \cite{DELLAPORTA201927}. The limited set of alloys known to satisfy the cofactor conditions undergo either cubic-to-orthorhombic or cubic-to-monoclinic-II phase transformations. Examples include Ti\textsubscript{54}Ni\textsubscript{34}Cu\textsubscript{12} \cite{chluba} and (Ti\textsubscript{54}Ni\textsubscript{34}Cu\textsubscript{12})\textsubscript{90}Nb\textsubscript{10} \cite{Tong2019} which exhibit a cubic to orthorhombic transformation, and Zn\textsubscript{45}Au\textsubscript{30}Cu\textsubscript{25} \cite{Song2013} and V\textsubscript{0.976}W\textsubscript{0.024}O\textsubscript{2} \cite{VO2} which undergo a cubic to monoclinic-II transformation.
\subsection{Cubic to orthorhombic transformation}
\label{sec:extreme_ortho}
In cubic to orthorhombic transformation, we have six orthorhombic variants as given by \eqref{eq:19}. Each variant forms a compound twin with exactly one other variant and Type I/II twins with the other variants, see \autoref{fig:ortho}.\\[5pt]
The cofactor conditions for cubic to orthorhombic transformations have been studied by Della Porta, with the results summarized in Table 2 of \cite{DELLAPORTA201927}. For reference, his findings are imported here as \autoref{tab:DP_CO}. Based on the satisfaction of cofactor conditions, Della Porta has categorized the pairs of variants into two columns. Column A contains all Type I/II twin pairs and column B contains all compound twin pairs for this transformation. If the cofactor conditions are satisfied for any pair in column A (or column B), then they are satisfied for all the pairs in column A (or column B). The difference between column A and column B is that when the cofactor conditions are satisfied for the pairs in column A, all the pairs in column A can form triple clusters with austenite and eliminate the transition layer at the phase interface, while the same is not true for column B. The satisfaction of cofactor conditions for column B does not enable the elimination of transition layer at the phase interface. Ideally, it would be desirable to have all the pairs in column A and column B, form triple cluster with austenite, leading to increased ways of forming stress-free interfaces with austenite. This idea of elimination of transition layers for all possible twin pairs (Type I/II and compound) leads to the conditions of extreme compatibility for cubic to orthorhombic transformation.
\begin{figure}
    \centering
    \begin{minipage}[b]{0.5\textwidth}
    \centering
    \includegraphics[scale=0.65]{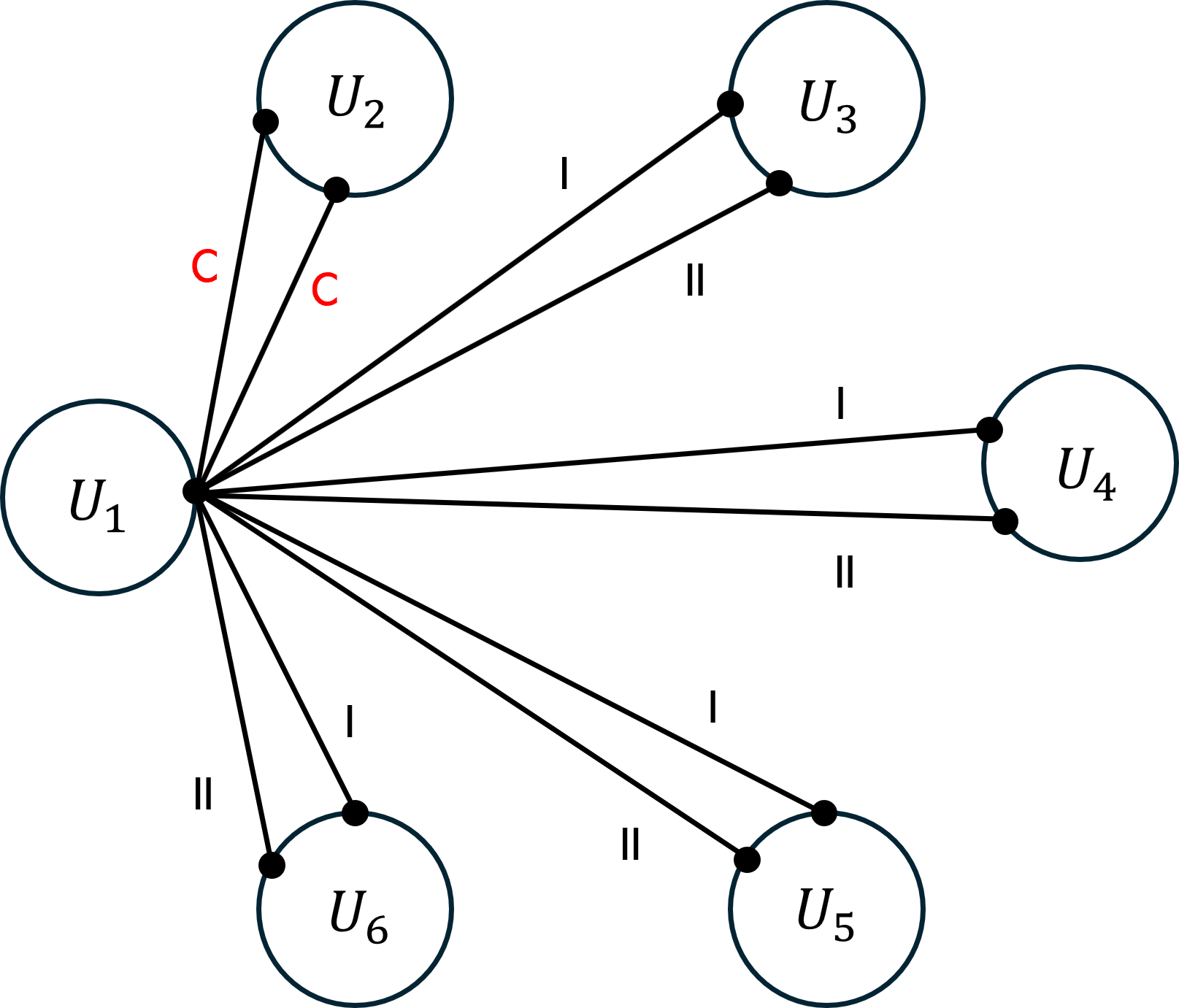}
    \subcaption{}
    \label{fig:ortho_a}
    \end{minipage}%
    \begin{minipage}[b]{0.5\textwidth}
    \centering
    \includegraphics[scale=0.65]{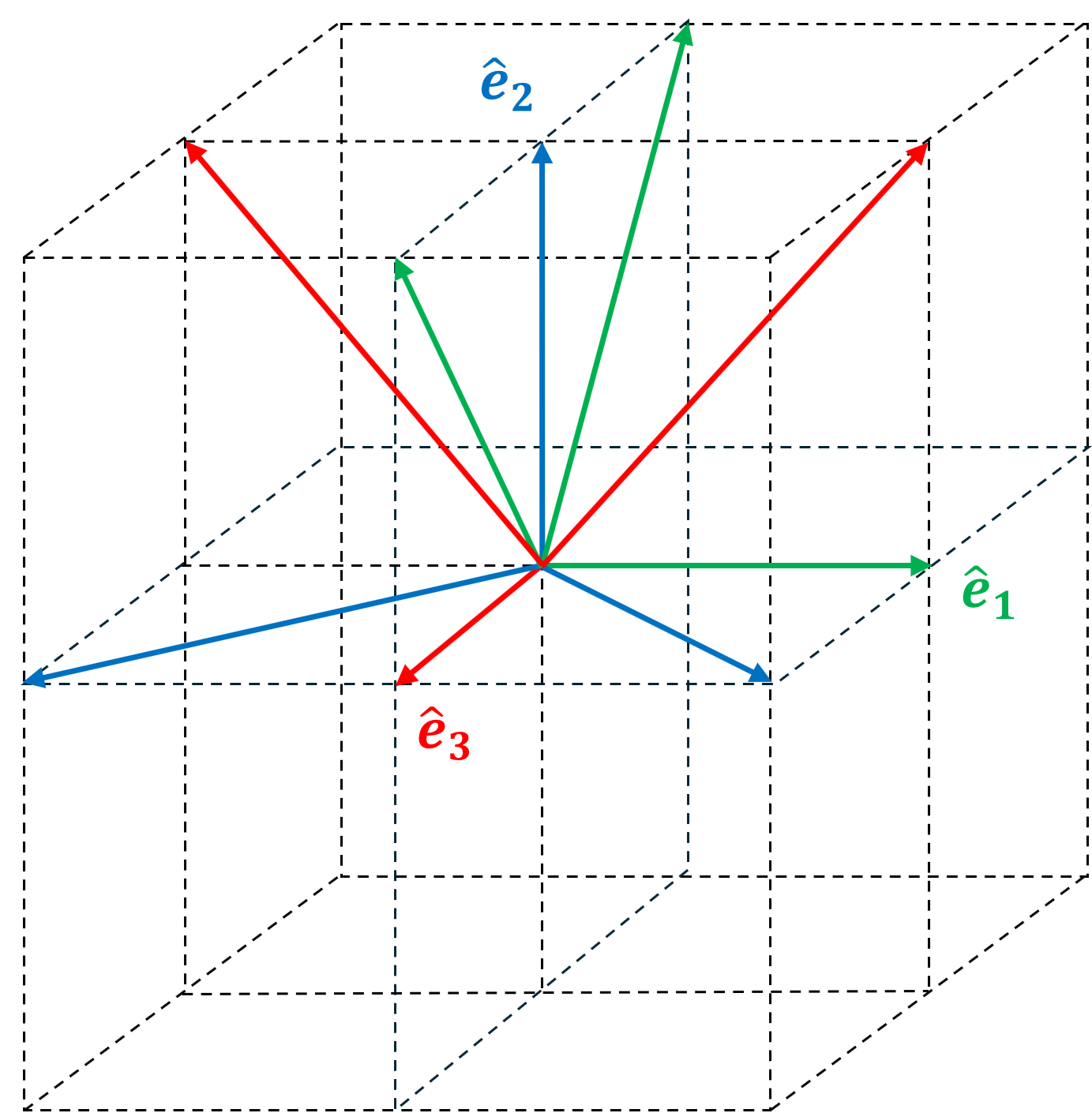}
    \subcaption{}
    \label{fig:ortho_b}
    \end{minipage}%
    \caption{(a) The strain space for cubic to orthorhombic transformations. I,II and C represents  Type-I, Type-II, and compound connections, respectively. (b) green vectors form the eigenbasis for  variants $U_1$ and $U_2$, blue vectors form the eigenbasis for variants $U_3$ and $U_4$, red vectors form the eigenbasis for variants $U_5$ and $U_6$}
    \label{fig:ortho}
\end{figure}
 \begin{table}
    \centering
 \begin{tabular}{ |c|c|c| } 
 \hline
 $R(\theta, axis)$ & Type I/II (Col A) & Compound (Col B)\\
 \hline
 &&\\[-5pt]
 $\pi,\{1,0,0\}$ & &{\color{red} ($U_3,U_4$),($U_5,U_6$)} \\[5pt] 
 $\pi,\{0,1,0\}$ & &{\color{red} ($U_1,U_2$),($U_5,U_6$)} \\[5pt] 
 $\pi,\{0,0,1\}$ & & {\color{red} ($U_1,U_2$),($U_3,U_4$)} \\[5pt] 
 $\pi,\{1,1,0\}$ & {\color{blue} ($U_1,U_3$),($U_2,U_4$)} &  \\[5pt] 
 $\pi,\{1,-1,0\}$ & {\color{blue}($U_1,U_4$),($U_2,U_3$)} & \\[5pt] 
 $\pi,\{0,1,1\}$ & {\color{blue}($U_3,U_6$),($U_4,U_5$)} & \\[5pt] 
 $\pi,\{0,1,-1\}$ & {\color{blue}($U_3,U_5$),($U_4,U_6$)} & \\[5pt] 
 $\pi,\{1,0,1\}$ & {\color{blue}($U_1,U_5$),($U_2,U_6$)} & \\[5pt]
 $\pi,\{1,0,-1\}$ & {\color{blue}($U_1,U_6$),($U_2,U_5$)} & \\[5pt] 
 \hline
\end{tabular}
\caption{Della Porta's classification of twin pairs in cubic to orthorhombic transformations. If a pair in any of the two columns satisfies cofactor conditions, then so do all the pairs of that column. }
\label{tab:DP_CO}
\end{table}
\\[5 pt]
 We now proceed to show that the conditions of extreme compatibility for cubic to orthorhombic transformation are achieved when $d=1$ and either \eqref{eq:27} or \eqref{eq:38} is satisfied. Here, $d$ is the ratio of the lattice parameter of the orthorhombic phase to that of the cubic phase along the $<100>_{cubic}||<100>_{orthorhombic}$ axes. $d$ is also an eigenvalue of the stretch tensor for this transformation. The aim here is to show that when $d=1$ and either \eqref{eq:27} or \eqref{eq:38} holds, then the laminates of all Type I/II as well as compound pairs listed in \autoref{tab:DP_CO} can simultaneously eliminate the transition layer at the phase interface for all volume fractions of the two variants. We employ the following strategy to arrive at this result. \\[5 pt]
      First, we show that when $d=1$ and \eqref{eq:27} is satisfied, then the cofactor conditions are satisfied for all Type I twins in column A,\autoref{tab:DP_CO}. This allows all Type I laminates to be compatible with austenite for all volume fractions, without requiring any transition layer. Moreover, under these conditions, \autoref{thm:firstkind} becomes applicable, which enables the formation of triple clusters of first kind between austenite and the pairs of compound twins in column B, \autoref{tab:DP_CO}. This provides the necessary condition for the compound laminates to be compatible with austenite for all volume fractions, without requiring any transition layer. Finally, we show that the condition of sufficiency for the case of compound twins, given by $\eqref{eq:CC3}$, is automatically satisfied. Therefore, $d=1$ and \eqref{eq:27} provide the necessary as well as the sufficient conditions for the existence of stress-free interfaces for the laminates of all pairs of \autoref{tab:DP_CO} with austenite for all volume fractions. We now present the analysis.\\[7pt]
Let $U_1$ be as in \eqref{eq:19}, such that $d=1$ and the other two eigenvalues satisfy \eqref{eq:27}. $\eqref{eq:CC1}$ is trivially satisfied because $d$ is an eigenvalue of $U_1$. Consider $U_1$ and $U_3$ which form Type I/II twins. The symmetry axis relating $U_1$ and $U_3$ is $\hat{e}=\frac{1}{\sqrt{2}}(\hat{\underline{e}}_1 + \hat{\underline{e}}_2)$. For Type-I twins, $\eqref{eq:CC2}$ can be written as (equation (34), \cite{CHEN20132566})
\begin{equation}
    \lambda_3 \sqrt{1-\lambda_1^2}\:(\hat{u}_1.\hat{e}) = \pm \lambda_1 \sqrt{\lambda_3^2 -1}\:(\hat{u}_3.\hat{e}).\label{eq:76}
\end{equation}\\[5 pt]
Without loss of generality, we can assume that $\lambda_1 = a-b$ and $\lambda_3 = a+b$. Then, we have $\hat{u}_1 = \frac{1}{\sqrt{2}}(-\hat{\underline{e}}_2 +\hat{\underline{e}}_3)$ and $\hat{u}_3 = \frac{1}{\sqrt{2}}(\hat{\underline{e}}_2 +\hat{\underline{e}}_3)$. Substituting these values in \eqref{eq:76}, gives us \eqref{eq:27}. This implies that by assuming $d=1$ and \eqref{eq:27} holds, \eqref{eq:CC1} and \eqref{eq:CC2} are satisfied for Type-I twins. The condition of sufficiency for Type-I laminates to be compatible with austenite for all volume fractions is given by the third cofactor condition, \eqref{eq:CC3}. We adopt a more general form of \eqref{eq:CC3}, following Chen et al. (equation 23, \cite{CHEN20132566}).
\begin{equation}
\label{eq:CC3prime}
    tr(U_1^2)-det(U_1^2)+(\mu^2 -\mu){|a|^2} - 2 = (1-\lambda_1^2(\mu))(\lambda_3^2(\mu)-1)> 0,
    \tag{CC3$^\prime$}
\end{equation}

where $|a|$ is the magnitude of the shear vector obtained by solving the twinning equation for the two variants, $\mu$ is the volume fraction of $U_1$ in the laminate, $\lambda_1(\mu)$ and $\lambda_3(\mu)$ are the eigenvalues of the stretch tensor of average deformation gradient formed by laminate of volume fraction $\mu$. Since $U_1 \in \mathbb{R}^{3 \times 3}_{sym+}$, the eigenvalues of $U_1^2$ are $\lambda_1^2,\;1,\;\lambda_3^2$.\\[5 pt]
The expression for the shear vector for Type-I solution is given by $\eqref{eq:13}_1$. In terms of eigenvalues, the expression for \eqref{eq:CC3prime} for the Type-I solution between $U_1$ and $U_3$ is

\begin{equation}
    -1+\lambda _1^2+\lambda _3^2 -\lambda _3^2 \lambda _1^2+\frac{\mu (\mu-1) \left(\left(2 \lambda _3^2+1\right) \lambda _1^4+2 \left(\lambda _3^4-5 \lambda _3^2+1\right) \lambda _1^2+\lambda _3^4+2 \lambda _3^2\right)}{\left(2 \lambda _3^2+1\right) \lambda _1^2+\lambda _3^2} > 0.\label{eq:77}
\end{equation}

Using \eqref{eq:27}, \eqref{eq:77} can be rewritten as
\begin{equation}
    \frac{(2 \mu(\mu-1) +1) \left(\lambda _1^2-1\right){}^2}{2 \lambda _1^2-1} > 0 .\label{eq:78}
\end{equation}

For $\mu \in[0,1]$, \eqref{eq:78} is satisfied for all admissible values of $\lambda_1$, i.e. $\frac{1}{\sqrt{2}}<\lambda_1<1$, see \autoref{fig:type1cc3}. This restriction on $\lambda_1$ is necessary for $U_1$ to have real and distinct eigenvalues and is inherent in the assumption of \eqref{eq:27}. Therefore, when $d=1$ and \eqref{eq:27} holds, the cofactor conditions are satisfied for the Type-I twin solution of the pair $(U_1,U_3)$ and therefore for all the pairs in column A of \autoref{tab:DP_CO}. This provides the necessary and sufficient conditions for Type-I laminates to eliminate the transition layer at the phase interface for all volume fractions. We now proceed to analyze the corresponding necessary and sufficient conditions for compound twins listed in column B of \autoref{tab:DP_CO}.\\[5 pt]
Consider the pairs $(U_1,U_2)$ which forms compound twins such that the eigenvector corresponding to $d$ is $\hat{\underline{e}}_1$ and symmetry axes relating $U_1$ and $U_2$ are $\hat{\underline{e}_2}$ and $\hat{\underline{e}}_3$. Clearly, $\hat{u}_2.\hat{e_1}$ = $\hat{u}_2.\hat{e_2}$ = 0, implying that $CC2$ is satisfied for the twinning solutions of $U_1$ and $U_2$ \cite{CHEN20132566}. Furthermore, $d=1$ and \eqref{eq:27} holds, then $(U_1,U_2)$ fulfill the premises of \autoref{thm:firstkind}. Therefore, two distinct triple clusters of first kind are possible between the pair $(U_1,U_2)$ and austenite. The only thing left to do is to check if the condition of sufficiency, \eqref{eq:CC3prime}, is satisfied. \\[5pt] 
For the compound twins, $(U_1$,$U_2)$, the shear vector is given by the expression in \eqref{eq:14} and the axes of symmetry are $\hat{\underline{e}}_1= \{0,1,0\}$ and $\hat{\underline{e}}_2= \{0,0,1\}$. Using these values in \eqref{eq:14}, we get,

\begin{equation}
    |a_{I}^c|^2 \;= |a_{II}^c|^2 \;= \;2\frac{(\lambda_3^2 - \lambda_1^2)^2}{\lambda_3^2 + \lambda_1^2}.\label{eq:79}
\end{equation}

Then, using \eqref{eq:79} and \eqref{eq:27} in \eqref{eq:CC3prime}, we have,

\begin{equation}
\frac{(1-2 \mu)^2 \left(\lambda _1^2-1\right){}^2}{2 \lambda _1^2-1}> 0.
 \label{eq:80}
\end{equation}

For $\mu\in[0,1]\setminus\{\frac{1}{2}\}$, \eqref{eq:80} is satisfied for all admissible values of $\lambda_1$, i.e. $\frac{1}{\sqrt{2}}<\lambda_1<1$, see \autoref{fig:compoundcc3}. However, at $\mu=\frac{1}{2}$, the L.H.S. of \eqref{eq:80} reduces to zero, irrespective of the value of $\lambda_1$. This implies that at a volume fraction of exactly one-half, the stretch tensor corresponding to the average deformation gradient of the compound laminate has at least two eigenvalues equal to one. By a brute force computation, one can show that when $d=1$ and \eqref{eq:27} holds, the eigenvalues of the laminate formed by $U_1$ and $U_2$ with volume fraction $\mu=\frac{1}{2}$ are $\lambda_1(\frac{1}{2})=1,\;\lambda_2(\frac{1}{2})=1,$ and $\; \lambda_3(\frac{1}{2})=\frac{\lambda_1^2}{\sqrt{2\lambda_1^2-1}}$. For this set of eigenvalues, the two solutions of \eqref{eq:17} coincide and collapse into a single solution. For all other volume fractions, \eqref{eq:17} has two distinct solutions. In both the cases, there exists at least one stress-free interface between the two phases.

\begin{figure}[hbt!]
    \centering
    \begin{minipage}[b]{0.5\textwidth}
    \centering
    \includegraphics[scale=0.0495]{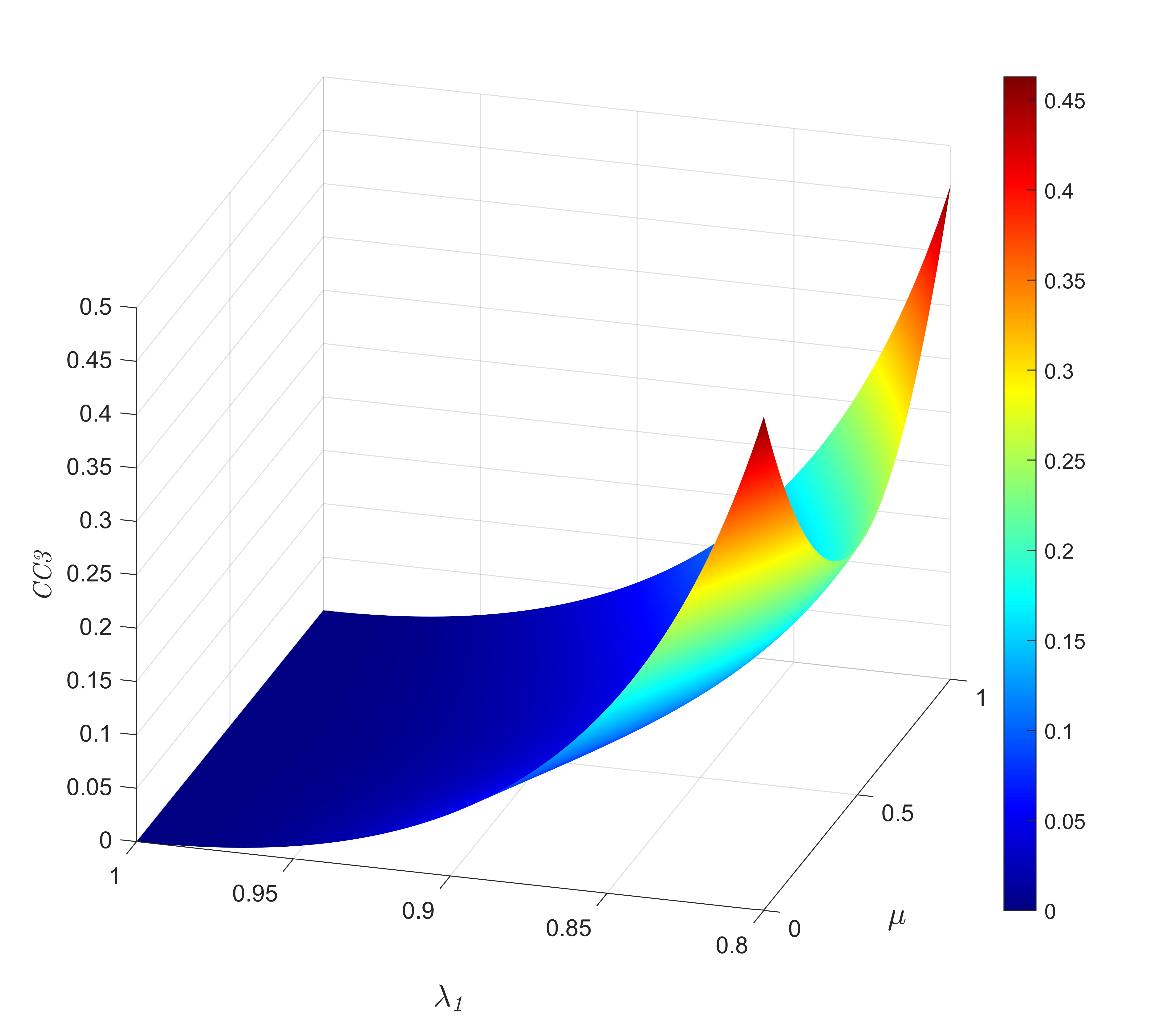}
    \subcaption{}
    \label{fig:type1cc3}
    \end{minipage}%
    \begin{minipage}[b]{0.5\textwidth}
    \centering
    \includegraphics[scale=0.0495]{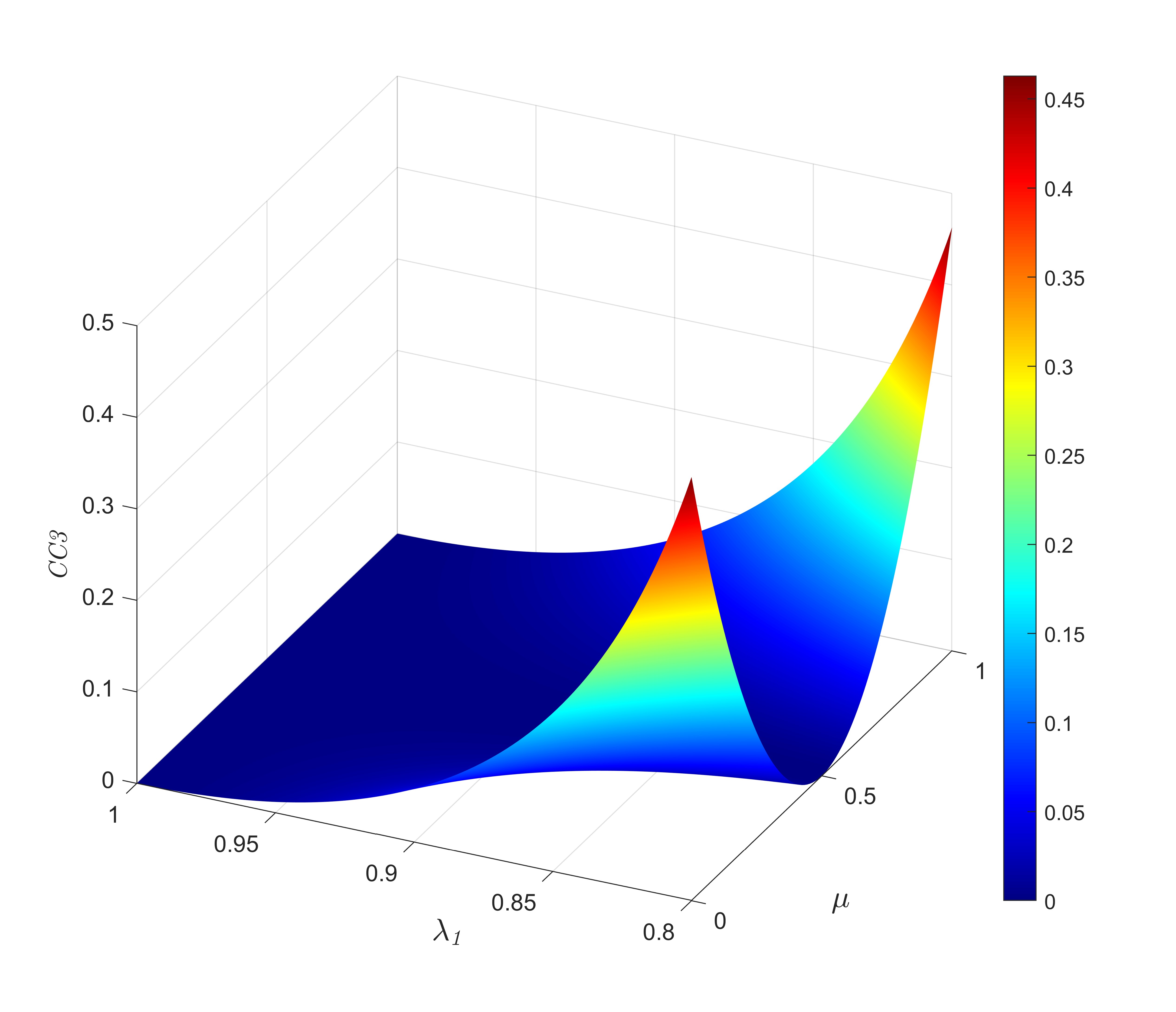}
    \subcaption{}
    \label{fig:compoundcc3}
    \end{minipage}%
    \caption{Plot of \eqref{eq:CC3prime} when $d=1$ and \eqref{eq:27} holds (a) for Type-I laminates, given by \eqref{eq:78} and (b) for compound laminates, given by \eqref{eq:80}}
    \label{fig:ortho_plot}
\end{figure}

Therefore, the conditions $d=1$ and \eqref{eq:27} are necessary as well as sufficient for the laminate of $(U_1,U_2)$ to be compatible with austenite for all volume fractions, without the need of any transition layer at the phase interface. We can repeat this analysis to get the same results for the compound pairs $(U_3,U_4)$ and $(U_5,U_6)$ as well. Therefore, we have established the first case of extreme compatibility for cubic to orthorhombic transformations, in which column A pairs of \autoref{tab:DP_CO} form Type-I triple clusters and column B pairs of \autoref{tab:DP_CO} form the triple clusters of first kind with austenite and their laminates form stress-free phase interfaces for all volume fractions. \\[5 pt]
The second case of extreme compatibility occurs when column A pairs of \autoref{tab:DP_CO} form Type-II triple clusters and column B pairs of \autoref{tab:DP_CO} form the triple clusters of second kind.  This happens when when $d=1$ and \eqref{eq:38} holds. In this case, the pairs of column A, satify the cofactor conditions for Type-II solution and the pairs of column B fulfill the premises of \autoref{thm:secondkind} and form triple clusters of second kind at the phase interface. Through an analysis similar to the first case, it can be shown that the conditions $d=1$ and \eqref{eq:38}, are necessary as well as sufficient for the laminates of Type-II twins and compound twins to maintain compatibility with austenite for all volume fractions without the need of a transition layers at the phase interface. As in the first case, for a volume fraction of $\mu=\frac{1}{2}$, the two solutions between the compound laminate and austenite collapse into a single solution. In this case, $\lambda_1$ can take values from a broader set, i.e. from the interval (0,1). \\[5 pt]
The overall structure of the analysis for the second case remains largely unchanged, with two key differences. First, instead of \eqref{eq:76}, we use \eqref{eq:CC2} for Type-II solution, given by (equation (34), \cite{CHEN20132566})

\begin{equation}
    \sqrt{1-\lambda_1^2}\:(\hat{u}_1.\hat{e}) = \pm \sqrt{\lambda_3^2 -1}\:(\hat{u}_3.\hat{e}).\label{eq:81}
\end{equation}

Second, in applying \eqref{eq:CC3prime} to check the sufficiency for Type-II laminates, we use the expression for shear vector corresponding to Type-II solution, given by $\eqref{eq:13}_2$. For brevity, we present the result directly without repeating the full computation.\\[5pt]
Under the conditions, when $d=1$ and \eqref{eq:38} holds, the Type-II laminates formed by the pairs in column A of \autoref{tab:DP_CO} and the compound laminates formed by the pairs in column B of \autoref{tab:DP_CO} are all compatible with austenite for all volume fractions, without the need of any transition layer at the phase interface.\\[5pt] 
\begin{table}
    \centering
\begin{tabular}{ |c|c| } 
 \hline
 $R(\theta, axis)$ & Twin system \\ 
 \hline
 &\\[- 5pt]
 $\pi,\{1,0,0\}$ & {\color{red} ($U_3,U_4$),($U_5,U_6$)} \\[5pt] 
 $\pi,\{0,1,0\}$ & {\color{red} ($U_1,U_2$),($U_5,U_6$)} \\[5pt] 
 $\pi,\{0,0,1\}$ & {\color{red} ($U_1,U_2$),($U_3,U_4$)} \\[5pt] 
 $\pi,\{1,1,0\}$ & {\color{blue} ($U_1,U_3$),($U_2,U_4$)}  \\[5pt] 
 $\pi,\{1,-1,0\}$ & {\color{blue}($U_1,U_4$),($U_2,U_3$)} \\[5pt] 
 $\pi,\{0,1,1\}$ & {\color{blue}($U_3,U_6$),($U_4,U_5$)} \\[5pt] 
 $\pi,\{0,1,-1\}$ & {\color{blue}($U_3,U_5$),($U_4,U_6$)} \\[5pt] 
 $\pi,\{1,0,1\}$ & {\color{blue}($U_1,U_5$),($U_2,U_6$)} \\[5pt]
 $\pi,\{1,0,-1\}$ & {\color{blue}($U_1,U_6$),($U_2,U_5$)} \\[5pt] 
 \hline
\end{tabular}
\caption{Red pairs from compound twins and blue pairs form Type I/II twins. If $d=1$  and \eqref{eq:27} holds, then the red pairs form two triple clusters of first kind with austenite and blue pairs form a single Type-I triple clusters with austenite. Similarly if $d=1$ and \eqref{eq:38} holds, then the red pairs form two triple clusters of second kind with austenite and blue pairs form a single Type-II triple clusters with austenite. In both the cases, all laminates form stress-free interfaces with austenite $\forall\:\mu\in[0,1]$.}
\label{tab:CO}
\end{table}
\begin{figure}[hbt!]
    \centering
    \begin{minipage}[b]{0.5\textwidth}
    \centering
    \includegraphics[scale=0.09]{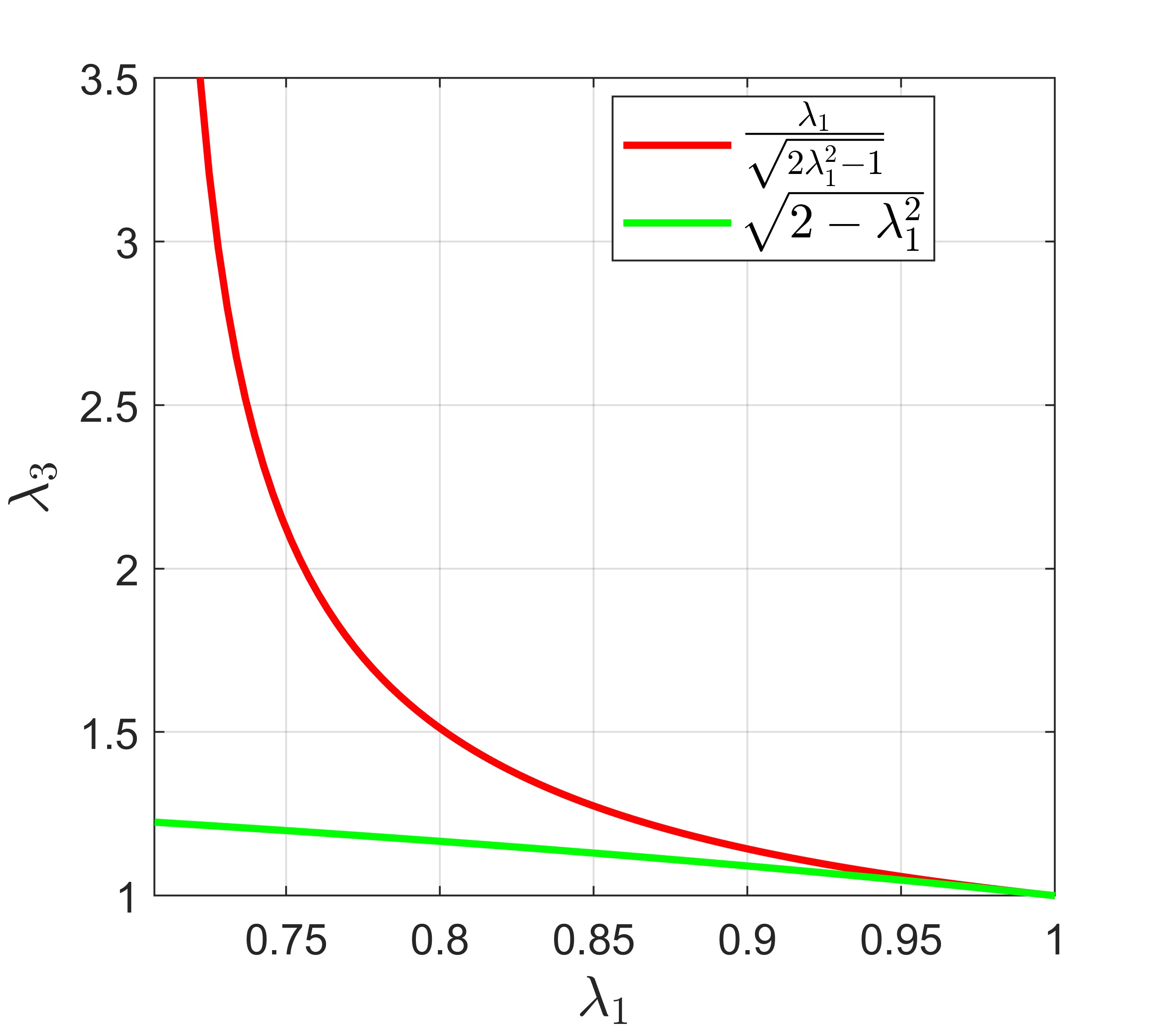}
    \subcaption{}
    \end{minipage}%
    \begin{minipage}[b]{0.5\textwidth}
    \centering
    \includegraphics[scale=0.09]{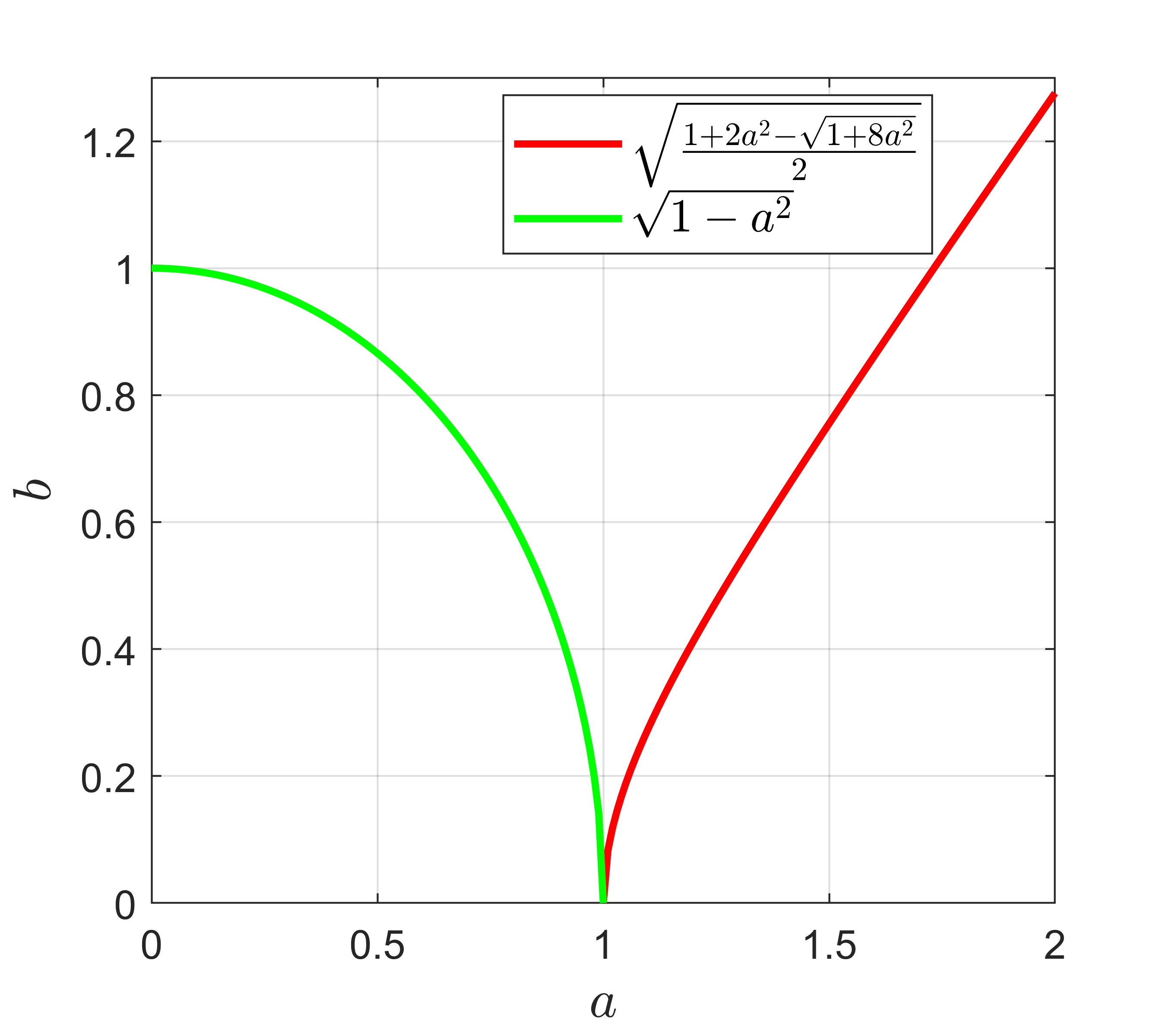}
    \subcaption{}
    \end{minipage}%
    \caption{Plot of equations \eqref{eq:27} (in red) and \eqref{eq:38} (in green) in (a) eigenvalue space space and (b) lattice parameter space. When $d=1$ is satisfied, then these curves represent the two cases of extreme compatibility in cubic to orthorhombic transformations, see \autoref{tab:CO}.}
    \label{fig:ortho_plot}
\end{figure}

\autoref{fig:ortho_plot} shows the conditions of extreme compatibility for cubic to orthorhombic transformation in the eigenvalue space and the lattice parameter space. We also tabulate the results of extreme compatibility for cubic to orthorhombic transformation, see \autoref{tab:CO} . Under the extreme compatibility conditions, all the pairs in Table \autoref{tab:CO}  can eliminate the transition layer at the phase interface simultaneously. Foe example, if $d=1$  and \eqref{eq:27} holds, then the pair $(U_1,U_2)$ forms two triple clusters of first kind with austenite and the pair $(U_1,U_3)$ forms a Type-I triple cluster with austenite. Similarly, if $d=1$ and \eqref{eq:38} holds, then the pair $(U_1,U_2)$ forms two triple clusters of second kind with austenite and the pair $(U_1,U_3)$ forms a Type-II triple cluster with austenite. In both the cases, the laminates of $(U_1,U_2)$ and $(U_1,U_3)$ are compatible with austenite for all volume fractions, without requiring any transition layer. The difference between \autoref{tab:DP_CO} and \autoref{tab:CO}, is that we have reduced the two columns in \autoref{tab:DP_CO} to a single column in \autoref{tab:CO}. These additional modes for eliminating the transition layer at the phase interface highlight the advantage of the extreme compatibility conditions over the cofactor conditions for cubic to orthorhombic transformation. Furthermore, the enhanced compatibility due to triple clusters in compound twins gives rise to interesting microstructures, which is discussed in \autoref{sec:new_micro}. \\[5pt]

\subsection{Cubic to monoclinic-II transformation}
\label{sec:extreme_mono}
For cubic to monoclinic-II transformation, we have twelve possible variants in the martensitic phase as given by \eqref{eq:22}. Each variant forms Type I/II twins with four variants, compound twins with two variants, non-conventional compound domain with one variant, and is not compatible with the four remaining variants, see \autoref{fig:mono_variants}.\\[5pt]
The implications of the cofactor conditions for cubic to monoclinic-II transformation have been examined by Della Porta \cite{DELLAPORTA201927}. Similar to the case of cubic to orthorhombic transformation, if the cofactor conditions are satisfied for Type I/II twins, it enables the elimination of transition layer for Type I/II laminates at the phase interface. But unlike cubic to orthorhombic transformations, satisfaction of CC's for any given Type I/II pair does not imply satisfaction of CC's for all Type I/II pairs, but only for a subset of Type I/II pairs, specifically, for half of the total possible Type I/II pairs. Consequently, Della Porta has classified all possible pairs of variants into three columns (Table 1, \cite{DELLAPORTA201927}). For ease of reference, the aforementioned table is reproduced herein as \autoref{tab:DP_CM}. \\[5pt]

\begin{figure}[hbt!]
    \centering
    \includegraphics[scale=0.75]{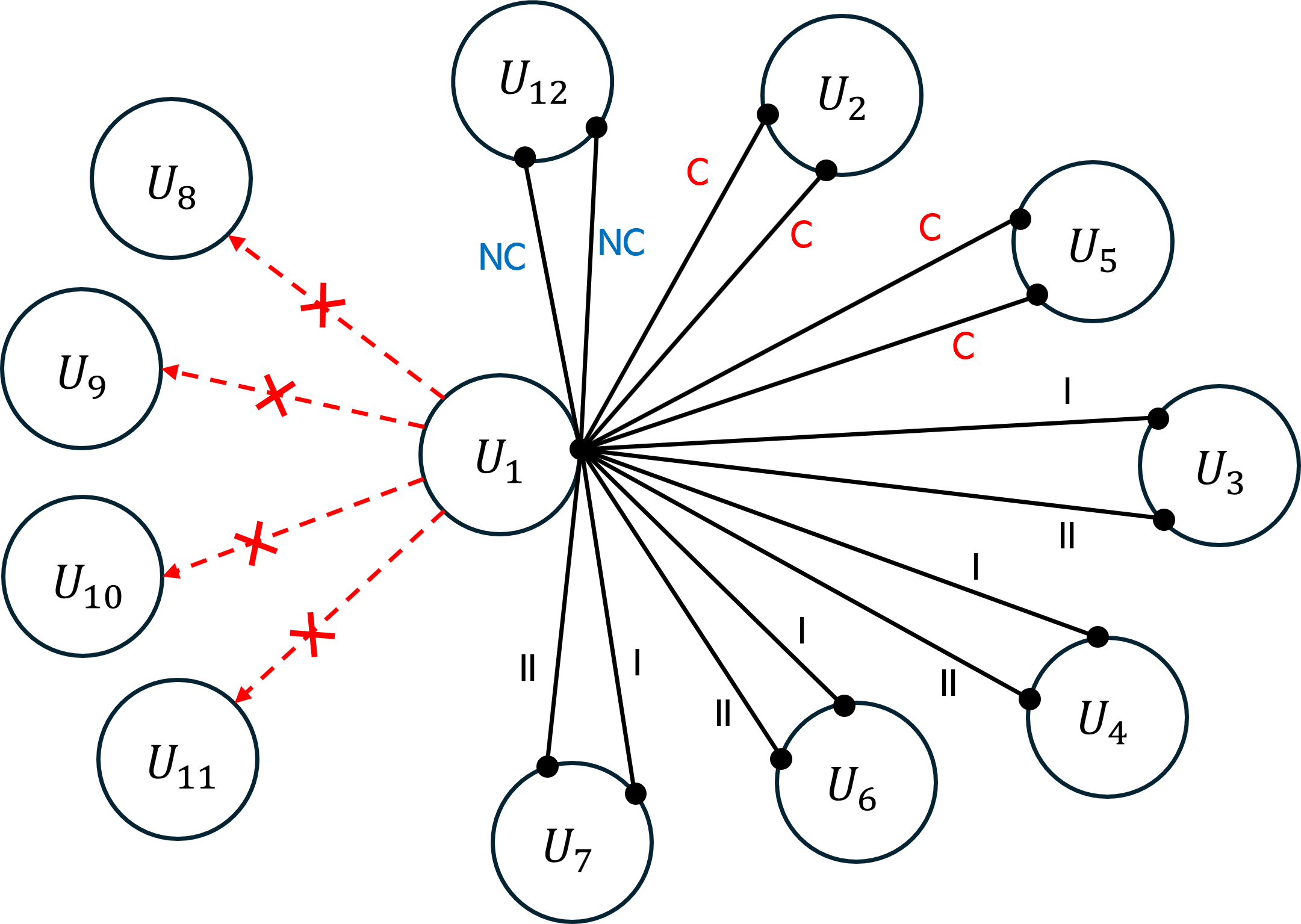}
    \caption{ The strain space for cubic to monoclinic-II transformations. I,II C and NC represent Type-I, Type-II, compound and non-conventional connections respectively. Red crosses denote incompatibility.}
    \label{fig:mono_variants}
\end{figure}
The first two columns consist of Type I/II twin pairs, such that if any pair from the two columns satisfies CC's, then all the pairs in that column also satisfy CC's. The third column has two different twin groups. These are the compound twin pairs (in red) and non-conventional compound domain pairs (in black). If any red pair in the third column satisfies CC's, then they are satisfied by all red pairs. Likewise, if any black pair in the third column satisfies CC's, then so do all the black pairs. However, the satisfaction of the cofactor conditions for any pair, whether red or black, does not imply the elimination of transition layers at the phase interface.\\[5pt]
In \autoref{sec:triple_beyond}, we have established the conditions which allow triple clusters in compound domains. We will use these conditions here to enable the elimination of transition layer for the pairs in column 3, in addition to the pairs in column 1 and column 2 of \autoref{tab:DP_CM}. The goal here is to regroup the variant pairs in column 3 of \autoref{tab:DP_CM}, based on their ability to eliminate the transition layers between the two phases, and not the cofactor conditions. This is because cofactor conditions do not allow for the elimination of transition layer for compound domains, and therefore do not constitute the optimal compatibility conditions for this transformation.

\begin{table}
    \centering
\begin{tabular}{ |c|c|c|c| } 
 \hline
 $R(\theta, axis)$ & Type I/II (Col 1) & Type I/II (Col 2) & Compound, NC (Col 3)\\
 \hline
 & & &\\[-5pt]
 $\pi,\{1,0,0\}$ & & & {\color{red} ($U_3,U_4$),($U_6,U_7$)}  \\[5pt]
 $\pi,\{1,0,0\}$ & & & {\color{red} ($U_8,U_{11}$),($U_9,U_{10}$)}  \\[5pt]
 $\pi,\{0,1,0\}$ & & & {\color{red} ($U_1,U_2$),($U_5,U_{12}$)}  \\[5pt]
  $\pi,\{0,1,0\}$ & & & {\color{red} ($U_6,U_7$),($U_8,U_{11}$)}  \\[5pt]
 $\pi,\{0,0,1\}$ & & & {\color{red} ($U_1,U_2$),($U_3,U_4$)} \\[5pt]
 $\pi,\{0,0,1\}$ & & & {\color{red} ($U_5,U_{12}$),($U_9,U_{10}$)} \\[5pt]
 $\pi,\{1,1,0\}$ & {\color{blue} ($U_1,U_3$)} &{\color{blue} ($U_5,U_{10}$)} & {\color{red} ($U_6,U_{11}$)}\\[5pt]
  $\pi,\{1,1,0\}$ & {\color{blue} ($U_2,U_4$)} &{\color{blue} ($U_9,U_{12}$)} & {\color{red} ($U_7,U_{8}$)}\\[5pt]
 $\pi,\{1,-1,0\}$ & {\color{blue} ($U_1,U_4$)} &{\color{blue} ($U_5,U_{9}$)} & {\color{red} ($U_6,U_{11}$)}\\[5pt]
 $\pi,\{1,-1,0\}$ & {\color{blue} ($U_2,U_3$)} &{\color{blue} ($U_{10},U_{12}$)} & {\color{red} ($U_7,U_{8}$)}\\[5pt]
 $\pi,\{0,1,1\}$ & {\color{blue} ($U_6,U_9$)} &{\color{blue} ($U_3,U_{8}$)} & {\color{red} ($U_1,U_{5}$)}\\[5pt]
 $\pi,\{0,1,1\}$ & {\color{blue} ($U_7,U_{10} $)} &{\color{blue} ($U_{4},U_{11}$)} & {\color{red} ($U_2,U_{12}$)}\\[5pt]
 $\pi,\{0,1,-1\}$ & {\color{blue} ($U_6,U_{10}$)} &{\color{blue} ($U_3,U_{11}$)} & {\color{red} ($U_1,U_{5}$)}\\[5pt]
  $\pi,\{0,1,-1\}$ & {\color{blue} ($U_7,U_{9}$)} &{\color{blue} ($U_4,U_{8}$)} & {\color{red} ($U_2,U_{12}$)}\\[5pt]
  $\pi,\{1,0,1\}$ & {\color{blue} ($U_5,U_{11}$)} &{\color{blue} ($U_1,U_{6}$)} & {\color{red} ($U_3,U_{10}$)}\\[5pt]
 $\pi,\{1,0,1\}$ & {\color{blue} ($U_8,U_{12} $)} &{\color{blue} ($U_{2},U_{7}$)} & {\color{red} ($U_4,U_{9}$)}\\[5pt]
   $\pi,\{1,0,-1\}$ & {\color{blue} ($U_{11},U_{12}$)} &{\color{blue} ($U_2,U_{6}$)} & {\color{red} ($U_3,U_{10}$)}\\[5pt]
  $\pi,\{1,0,-1\}$ & {\color{blue} ($U_5,U_{8} $)} &{\color{blue} ($U_{1},U_{7}$)} & {\color{red} ($U_4,U_{9}$)}\\[5pt]
 $\frac{\pi}{2},\{1,0,0\}$ &  & & {\color{black} ($U_1,U_{12}$),($U_2,U_5$)}\\[5pt]
 $\frac{-\pi}{2},\{1,0,0\}$ & & & {\color{black} ($U_1,U_{12}$),($U_2,U_5$)} \\[5pt]
 $\frac{\pi}{2},\{0,1,0\}$ & & & {\color{black} ($U_3,U_9$),($U_4,U_{10}$)} \\[5pt]
 $\frac{-\pi}{2},\{0,1,0\}$ & & & {\color{black} ($U_3,U_9$),($U_4,U_{10}$)} \\[5pt]
 $\frac{\pi}{2},\{0,0,1\}$ & & & {\color{black} ($U_6,U_8$),($U_7,U_{11}$)} \\[5pt]
 $\frac{-\pi}{2},\{0,0,1\}$ & & & {\color{black} ($U_6,U_8$),($U_7,U_{11}$)} \\[5pt]
 \hline
\end{tabular}
\caption{Della Porta's classification of variant pairs to understand the cofactor conditions for cubic to monoclinic-II transformation. In blue are the Type I/II pairs, in red are the compound pairs and in black are the non-conventional pairs.}
\label{tab:DP_CM}
\end{table}

In order to regroup \autoref{tab:DP_CM} on the basis of triple clusters, we make use of \autoref{thm:dellaportaQ} in the following way. Recall equations \eqref{eq:58},\eqref{eq:59}, \eqref{eq:61} and \eqref{eq:63}. These are the conditions for the existence of triple cluster(s) between a pair of compound domains and austenite. In all four equations we have $(\hat{u}_1.\hat{e}_1)$ term and $(\hat{u}_3.\hat{e}_1)$ term. Now, for any $Q \in SO(3)$, we have $(\hat{u}_1.\hat{e}_1)\:=\:$$(Q\hat{u}_1.Q\hat{e}_1)$ and $(\hat{u}_3.\hat{e}_1)$\:=\:$(Q\hat{u}_3.Q\hat{e}_1)$. From \autoref{thm:dellaportaQ}, we know that if $U$ and $V$ form compound domains, with symmetry axes $\hat{e}_1$ and $\hat{e}_2$, then $QUQ^T$ and $QVQ^T$ also form compound domains, with symmetry axes $Q\hat{e}_1$ and $Q\hat{e}_2$. Further, if $\hat{u}_i, i=1,2,3,$ are the eigenvectors of $U$, then $Q\hat{u}_i, i=1,2,3,$ are the eigenvectors of $QUQ^T$. Then, we conclude that if $U$ and $V$ satisfy any of the equations \eqref{eq:58},\eqref{eq:59}, \eqref{eq:61} or \eqref{eq:63}, then so do $QUQ^T$ and $QVQ^T$. To regroup \autoref{tab:DP_CM}, we restrict the set $Q$ to $\mathscr{P}(a)$.  \\[5 pt]
For different choices of $Q \in \mathcal{P}(a)$, the twins pairs listed in column 3 of \autoref{tab:DP_CM} can be sorted into three separate columns, as shown in \autoref{tab:CM}. Columns 3 and 4 of \autoref{tab:CM} comprise of compound twins such that if any pair of variants in column 3 (or column 4) forms a triple cluster of the first or the second kind, then so do all the pairs of column 3 (or column 4). Column 5 comprises of non-conventional compound domains which commute. Again, for different choices of $Q \in \mathscr{P}(a)$, we can generate all the pairs of non-conventional compound domains, starting from a single pair. Also, \autoref{lem:commutationQ} tells us that if one pair in column 5 commutes, then so do all the pairs. Under the conditions, when $d=1$ and \eqref{eq:27} holds, \autoref{thm:firstkind} is applicable to the pairs in column 5, resulting in two distinct triple clusters of the first kind for every pair of column 5. Similarly, when $d=1$ and \eqref{eq:38} holds, \autoref{thm:secondkind} is applicable to the pairs in column 5, resulting in two distinct triple clusters of the second kind for every pair of column 5. This finalizes the regrouping of twin pairs from \autoref{tab:DP_CM} into the new structure of \autoref{tab:CM}. It is important to note that this new classification of variants is based on the ability to form triple clusters with austenite and not the cofactor conditions.\\[5 pt]
\begin{table}
    \centering
\begin{tabular}{ |c|s|s|s|s|c|} 
 \hline
 \multicolumn{1}{|c|}{Symmetry} & \multicolumn{2}{|c|}{Type I/II} & \multicolumn{2}{|c|}{Compound}  &  \multicolumn{1}{|c|}{Non-Conventional}\\
 \hline
 $R(\theta, axis)$ & Col 1 & Col 2 & Col 3 & Col 4 & Col 5\\
 \hline
 & & & & & \\[-5pt]
 $\pi,\{1,0,0\}$ & & & {\color{red} ($U_3,U_4$),($U_6,U_7$)} & &  \\[5pt]
 $\pi,\{1,0,0\}$ & & & {\color{red} ($U_8,U_{11}$),($U_9,U_{10}$)} & &  \\[5pt]
 $\pi,\{0,1,0\}$ & & & {\color{red} ($U_1,U_2$),($U_5,U_{12}$)}  & & \\[5pt]
  $\pi,\{0,1,0\}$ & & & {\color{red} ($U_6,U_7$),($U_8,U_{11}$)}  & & \\[5pt]
 $\pi,\{0,0,1\}$ & & & {\color{red} ($U_1,U_2$),($U_3,U_4$)} & & \\[5pt]
 $\pi,\{0,0,1\}$ & & & {\color{red} ($U_5,U_{12}$),($U_9,U_{10}$)} & & \\[5pt]
 $\pi,\{1,1,0\}$ & {\color{blue} ($U_1,U_3$)} & {\color{blue} ($U_5,U_{10}$)} & & {\color{red} ($U_6,U_{11}$)} &\\[5pt]
  $\pi,\{1,1,0\}$ & {\color{blue} ($U_2,U_4$)} &{\color{blue} ($U_9,U_{12}$)} & & {\color{red} ($U_7,U_{8}$)} &\\[5pt]
 $\pi,\{1,-1,0\}$ & {\color{blue} ($U_1,U_4$)} &{\color{blue} ($U_5,U_{9}$)} & & {\color{red} ($U_6,U_{11}$)} & \\[5pt]
 $\pi,\{1,-1,0\}$ & {\color{blue} ($U_2,U_3$)} &{\color{blue} ($U_{10},U_{12}$)} & & {\color{red} ($U_7,U_{8}$)} & \\[5pt]
 $\pi,\{0,1,1\}$ & {\color{blue} ($U_6,U_9$)} &{\color{blue} ($U_3,U_{8}$)} & & {\color{red} ($U_1,U_{5}$)} & \\[5pt]
 $\pi,\{0,1,1\}$ & {\color{blue} ($U_7,U_{10} $)} &{\color{blue} ($U_{4},U_{11}$)} & & {\color{red} ($U_2,U_{12}$)} &\\[5pt]
 $\pi,\{0,1,-1\}$ & {\color{blue} ($U_6,U_{10}$)} &{\color{blue} ($U_3,U_{11}$)} & & {\color{red} ($U_1,U_{5}$)} &\\[5pt]
  $\pi,\{0,1,-1\}$ & {\color{blue} ($U_7,U_{9}$)} &{\color{blue} ($U_4,U_{8}$)} & & {\color{red} ($U_2,U_{12}$)} &\\[5pt]
  $\pi,\{1,0,1\}$ & {\color{blue} ($U_5,U_{11}$)} &{\color{blue} ($U_1,U_{6}$)} & & {\color{red} ($U_3,U_{10}$)} &\\[5pt]
 $\pi,\{1,0,1\}$ & {\color{blue} ($U_8,U_{12} $)} &{\color{blue} ($U_{2},U_{7}$)} & & {\color{red} ($U_4,U_{9}$)} &\\[5pt]
   $\pi,\{1,0,-1\}$ & {\color{blue} ($U_{11},U_{12}$)} &{\color{blue} ($U_2,U_{6}$)} & & {\color{red} ($U_3,U_{10}$)}&\\[5pt]
  $\pi,\{1,0,-1\}$ & {\color{blue} ($U_5,U_{8} $)} &{\color{blue} ($U_{1},U_{7}$)} & & {\color{red} ($U_4,U_{9}$)} & \\[5pt]
 $\frac{\pi}{2},\{1,0,0\}$ &  & & & & {\color{black} ($U_1,U_{12}$),($U_2,U_5$)}\\[5pt]
 $\frac{-\pi}{2},\{1,0,0\}$ & & & & & {\color{black} ($U_1,U_{12}$),($U_2,U_5$)} \\[5pt]
 $\frac{\pi}{2},\{0,1,0\}$ & & & & & {\color{black} ($U_3,U_9$),($U_4,U_{10}$)} \\[5pt]
 $\frac{-\pi}{2},\{0,1,0\}$ & & & & & {\color{black} ($U_3,U_9$),($U_4,U_{10}$)} \\[5pt]
 $\frac{\pi}{2},\{0,0,1\}$ & & & & & {\color{black} ($U_6,U_8$),($U_7,U_{11}$)} \\[5pt]
 $\frac{-\pi}{2},\{0,0,1\}$ & & & & & {\color{black} ($U_6,U_8$),($U_7,U_{11}$)} \\[5pt]
 \hline
\end{tabular}
\caption{New classification of variant pairs for cubic to monoclinic-II transformation based on the ability to eliminate the transition layer at the phase interface. In blue are the Type I/II pairs, in red are the compound pairs and in black are the non-conventional pairs. If a pair in any column eliminates the transition layer at the phase interface, then so do all the members of that column. Shaded columns (Col 1, Col 2, Col 3, Col 4) represent the region covered by optimal compatibility conditions for this transformation, where all the conventional twins pairs can simultaneously eliminate the transition layer at the phase interface. This happens when \eqref{eq:87} or \eqref{eq:88} is satisfied. }
\label{tab:CM}
\end{table}
The necessary conditions for the formation of triple clusters for every column of \autoref{tab:CM} are now known. The sufficient condition for the laminates of the pairs in these columns to be compatible with austenite for all volume fractions, $\mu\in[0,1]$, is given by \eqref{eq:CC3} in each case. The only term in the expression for \eqref{eq:CC3}, which differs from column to column, is the $\frac{|a|^2}{4}$ term. Noting that $|a|^2=|Qa|^2$ for any $Q\in SO(3)$, we can use \autoref{thm:dellaportaQ} to establish that if \eqref{eq:CC3} is satisfied for any pair, in any column, for a particular kind of triple cluster, then it is satisfied for all the pairs of that column for that kind of triple cluster. To arrive at the conditions of extreme compatibility in cubic to monoclinic-II transformations, what remains to be checked, is that for how many of these columns are triple clusters possible at the same time. We present the argument in the following six points. We will assume $d=1$, till the end of this section.\\[5 pt]
\begin{enumerate}[label=\Roman*)]
    \item It is not possible to form two distinct triple clusters of the same kind for a given pair of compound twins listed in column 3 and column 4 of \autoref{tab:CM} . In other words, equation $(\hat{u}_1.\hat{e}_1)=-(\hat{u}_3.\hat{e}_1)$ cannot be satisfied here, see \autoref{cor:45degrees}. It suffices to show this for $U_1$, which forms compound twins with $U_2$ and $U_5$. For $(U_1,U_2)$, the symmetry axis, $\hat{e}_A$ is $\{0,1,0\}$ and the eigenvectors of $U_1$ are\\[5pt]
    \begin{equation}
        \hat{u}_1= \frac{1}{\rho_1}\begin{pmatrix} 0 \\ \frac{a-c-\sqrt{D}}{2b} \\ 1 \end{pmatrix},\qquad \qquad \hat{u}_3= \frac{1}{\rho_2}\begin{pmatrix} 0 \\ \frac{a-c+\sqrt{D}}{2b} \\ 1 \end{pmatrix}.\;\; 
        \notag
    \end{equation}\\[5pt]
    where $D= (a-c)^2+4b^2$ and $\rho_1,\rho_2$ are chosen so that $|\hat{u}_1|=|\hat{u}_3|=1$. Now, $(\hat{u}_1.\hat{e}_1) = -\: (\hat{u}_3. \hat{e}_1)$ implies\\[5pt]
    \begin{equation}
        (\hat{u}_1.\hat{e}_A) = -\: (\hat{u}_3. \hat{e}_A)
        \notag
    \end{equation}
    \begin{equation}
       \implies \rho_2^2(a-c-\sqrt{D})^2 = \rho_1^2(a-c+\sqrt{D})^2.
        \notag
    \end{equation}
    \begin{equation}
        \implies (\frac{(a-c+\sqrt{D})^2}{4b^2} +1)(a-c-\sqrt{D})^2 = (\frac{(a-c-\sqrt{D})^2}{4b^2} +1)(a-c+\sqrt{D})^2.
        \notag
    \end{equation}
     \begin{equation}
        \implies a-c-\sqrt{D} = \:\pm\; (a-c+\sqrt{D}).
        \notag
    \end{equation}
    \begin{equation}
        \implies \text{either }\: D=0\;\; \text{or }\:a=c.
        \notag
    \end{equation}\\[5pt]
    $D=0$ further implies $b=0$ \text{and} $a=c$. This reduces the transformation to cubic to orthorhombic and therefore, for this compound pair,  $(\hat{u}_1.\hat{e}_1) = -\: (\hat{u}_3. \hat{e}_1)$ does not hold.\\[10pt]
    For the pair $(U_1,U_5)$, the axis of symmetry, $\hat{e}_B=$ $\frac{1}{\sqrt{2}}{\{0,1,1\}}$ and $(\hat{u}_1.\hat{e}_1) = -\: (\hat{u}_3. \hat{e}_1)$ implies\\[5pt]
    \begin{equation}
        (\hat{u}_1.\hat{e}_B) = -\: (\hat{u}_3. \hat{e}_B).
        \notag
    \end{equation}
    \begin{equation}
        \implies \rho_2^2(a-c-\sqrt{D}+2b)^2 = \rho_1^2(a-c+\sqrt{D}+2b)^2.
        \notag
    \end{equation}
     \begin{equation}
        \implies(\frac{(a-c+\sqrt{D})^2}{4b^2} +1)(a-c-\sqrt{D} + 2b)^2 = (\frac{(a-c-\sqrt{D})^2}{4b^2} +1)(a-c+\sqrt{D}+2b)^2.
        \notag
    \end{equation}
     \begin{equation}
        \implies-4b^2 (a-c-\sqrt{D})+4b^2 (a-c+\sqrt{D}) =  -4b^2 (a-c+\sqrt{D})+4b^2 (a-c-\sqrt{D}).
        \notag
    \end{equation}\\[5pt]
    which again requires either $b=0$ or $D=0$, which is forbidden. \\[5pt]
Alternatively, we can arrive at the same result by noting that the pairs listed in columns 3 and 4 do not commute and therefore by \autoref{cor:45degrees}, cannot form two distinct triple clusters of the same kind.
    \item It is possible to form two distinct triple clusters of the same kind, simultaneously, from a single pair of non-conventional compound domains in column 5 of \autoref{tab:CM}. This is achieved when $d=1$ and either \eqref{eq:27} or \eqref{eq:38} is satisfied, see \autoref{thm:firstkind} and \autoref{thm:secondkind}.\\[5pt]
    For example, the pair $(U_1,U_{12})$ is a pair of commuting compatible variants, and can form two triple clusters of first or the second kind, depending on whether \eqref{eq:27} or \eqref{eq:38} holds along with $d=1$, see \autoref{rem:45_O_NCM}.

    \item If the pairs in column 1 form a Type-I triple cluster, then the pairs in column 3 form a triple cluster of the first kind and vice-versa. This can be shown as follows.\\[5pt]
    Consider the representative pair $(U_1,U_3)$ from column 1 which is related by $\hat{e}_T = \frac{1}{\sqrt{2}}\{1,1,0\}$. Suppose $(U_1,U_3)$ satisfy the cofactor conditions for Type-I, implying a triple cluster of the first kind between this pair and austenite. Then, \eqref{eq:CC2} implies\\[5pt]
    \begin{equation}
        \lambda_3\sqrt{1-\lambda_1^2}\:(\hat{u}_1.\hat{e}_T)=- \lambda_1\sqrt{\lambda_3^2-1}\:(\hat{u}_3.\hat{e}_T).\label{eq:82}
    \end{equation}
    \\[5 pt]
   We have chosen the minus sign here because $(\hat{u}_1.\hat{e}_T)$ and $(\hat{u}_1.\hat{e}_T)$ are always of opposite sign. For the triple cluster of the first kind to hold for column 3, either \eqref{eq:58} or \eqref{eq:59} should hold. Again, consider the representative pair $(U_1,U_2)$ from column 3, related by $\hat{e}_A = \{{0,1,0}\}$. For \eqref{eq:58} to hold, we must have \\[5 pt]
    \begin{equation}
        \lambda_3\sqrt{1-\lambda_1^2}\:(\hat{u}_1.\hat{e}_A)=- \lambda_1\sqrt{\lambda_3^2-1}\:(\hat{u}_3.\hat{e}_A).
    \end{equation}\label{eq:83}
    \\[5 pt]
    Then for the triple clusters of the first kind to hold simultaneously for column 1 and column 3, we must have \\[5 pt]
    \begin{equation}
        \frac{(\hat{u}_1.\hat{e}_T)}{(\hat{u}_1.\hat{e}_A)} = \frac{(\hat{u}_3.\hat{e}_T)}{(\hat{u}_3.\hat{e}_A)},
        \notag
    \end{equation}
    \\[5 pt]
    which is always true. We show the calculation here for the sake of completeness.\\[5pt]
    \begin{equation}
        \hat{u}_1= \frac{1}{\rho_1}\begin{pmatrix} 0 \\ \frac{a-c-\sqrt{D}}{2b} \\ 1 \end{pmatrix}, \qquad \hat{u}_3= \frac{1}{\rho_2}\begin{pmatrix} 0 \\ \frac{a-c+\sqrt{D}}{2b} \\ 1 \end{pmatrix},\qquad \hat{e}_T=\frac{1}{\sqrt{2}}\begin{pmatrix}
            1\\1\\0
        \end{pmatrix},\qquad \hat{e}_A=\begin{pmatrix} 0\\1\\0\end{pmatrix} .
        \notag
    \end{equation}
    \\[5pt]
    Substituting the values in the expression above, we have\\[5pt]
    \begin{equation}
        \frac{(\hat{u}_1.\hat{e}_T)}{(\hat{u}_1.\hat{e}_A)} = \frac{(\hat{u}_3.\hat{e}_T)}{(\hat{u}_3.\hat{e}_A)}=\frac{1}{\sqrt{2}}.
        \notag
    \end{equation}
    \\[5pt]
    Similarly, we can show if the pairs in column 1 form a Type-II triple cluster, then the pairs in column 3 form a triple cluster of the second kind and vice-vera. For example, if the pair $(U_1,U_3)$ of column 1 forms a Type-I triple cluster with austenite, then the pair $(U_1,U_2)$ forms a triple cluster of first kind with austenite. Similarly, if the pair $(U_1,U_2)$ of column 3 forms a triple cluster of second kind with austenite, then the pair $(U_1,U_3)$ from column 1 forms a triple cluster of second kind with austenite.

    \item It is possible to form the same kind of triple cluster for the pairs in column 3 and column 4. To show this, we again analyze the representative pairs from each column. From column 3, we choose the pair $(U_1,U_2)$ as before, related by $\hat{e}_A=\{0,1,0\}$. Let $(U_1,U_5)$ be the representative pair from column 4, where the variants are related by $\hat{e}_B=\frac{1}{\sqrt{2}}\{0,1,1\}$. We present the analysis for the triple cluster of first kind. \\[5 pt]
    Let $(U_1,U_2)$ form a triple cluster of the first kind and satisfy \eqref{eq:58}. Then we have \\[5 pt]
    \begin{equation}
        \lambda_3\sqrt{1-\lambda_1^2}\:(\hat{u}_1.\hat{e}_A)=- \lambda_1\sqrt{\lambda_3^2-1}\:(\hat{u}_3.\hat{e}_A).
        \notag
    \end{equation}
    \\[5 pt]
    In order to have the first kind of triple cluster for $(U_1,U_5)$ as well, either \eqref{eq:58} or \eqref{eq:59} must hold. Lets assume \eqref{eq:58} holds, then we must have \\[5 pt]
    \begin{equation}
        \lambda_3\sqrt{1-\lambda_1^2}\:(\hat{u}_1.\hat{e}_B)=- \lambda_1\sqrt{\lambda_3^2-1}\:(\hat{u}_3.\hat{e}_B).\label{eq:84}
    \end{equation}
    \\[5 pt]
     Then, for both cases to hold at once, we must have, 
    \begin{equation}
        \frac{(\hat{u}_1.\hat{e}_A)}{(\hat{u}_1.\hat{e}_B)} = \frac{(\hat{u}_3.\hat{e}_A)}{(\hat{u}_3.\hat{e}_B)}
          \notag
    \end{equation}
    \begin{equation}
        \implies \frac{\frac{a-c-\sqrt{D}}{2b\rho_1}}{\frac{1}{\rho_1}(\frac{a-c-\sqrt{D}}{2b}+1)} = \frac{\frac{a-c+\sqrt{D}}{2b\rho_2}}{\frac{1}{\rho_2}(\frac{a-c+\sqrt{D}}{2b}+1)}
        \notag
    \end{equation}
    \begin{equation}
        \implies \frac{2b}{a-c-\sqrt{D}} = \frac{2b}{a-c+\sqrt{D}} \implies D=0,
        \notag
    \end{equation}
    which is forbidden. Therefore, this possibility does not exist. If, otherwise, $(U_1,U_5)$ satisfies \eqref{eq:59}, then, we must have \\[5 pt]
    \begin{equation}
        \lambda_3\sqrt{1-\lambda_1^2}\:(\hat{u}_3.\hat{e}_B)=- \lambda_1\sqrt{\lambda_3^2-1}\:(\hat{u}_1.\hat{e}_B).
        \notag
    \end{equation}
    \\[5 pt]
    Again, for both cases to hold simultaneously, we must have,
    \begin{equation}
        \frac{(\hat{u}_1.\hat{e}_A)}{(\hat{u}_3.\hat{e}_B)} = \frac{(\hat{u}_3.\hat{e}_A)}{(\hat{u}_1.\hat{e}_B)}
          \notag
    \end{equation}
    \begin{equation}
       \implies \frac{\frac{a-c-\sqrt{D}}{2b\rho_1}}{\frac{1}{\rho_2}(\frac{a-c+\sqrt{D}}{2b}+1)} = \frac{\frac{a-c+\sqrt{D}}{2b\rho_2}}{\frac{1}{\rho_1}(\frac{a-c-\sqrt{D}}{2b}+1)}
        \notag
    \end{equation}
    \begin{equation}
        \implies (a-2b-c)^2=12b^2 + (a-c)^2
        \notag
    \end{equation}
    \begin{equation}
        \implies a-c=-2b\label{eq:85}
    \end{equation}\\[5pt]
    Similarly, we can show that for the triple clusters of second kind to coexist for columns 3 and 4, \eqref{eq:85} must hold. \\[5pt]
    For example, if the pair $(U_1,U_2)$ of column 3 forms a triple cluster of first kind with austenite and \eqref{eq:85} holds, then the pair $(U_1,U_5)$ of column 4 also forms a triple cluster of first kind with austenite. Conversely, if the pair $(U_1,U_5)$ of column 4 forms a triple cluster of first kind with austenite and \eqref{eq:85} holds, then the pair $(U_1,U_2)$ also forms a triple cluster of first kind with austenite. The same can be said about the triple clusters of second kind in these pairs as well. \\[5pt]
    Here, it is important to emphasize the difference between columns 1,2 and the columns 3,4. All the pairs in the columns 1,2 are Type I/II twins and cannot form the same kind of triple cluster, as shown by Della Porta \cite{DELLAPORTA201927}. Columns 3,4 consist of pairs of compound twins only, and it is possible for both the columns to form the same kind of triple cluster. This suggests that compound twins offer a greater flexibility to form compatible triple clusters than Type I/II twins.

    \item Pairs of column 1 and column 2 cannot form the same type of triple cluster simultaneously. This was observed by Della Porta (Remark 3.2, \cite{DELLAPORTA201927}). However, it is possible to have different types of triple clusters in column 1 and column 2, at the same time. To see this, let us assume that the pairs in column 1 form a Type-I triple cluster with austenite. Then, \eqref{eq:CC2} for the representative pair $(U_1,U_3)$ implies \\[5 pt]
    \begin{equation}
        \lambda_3\sqrt{1-\lambda_1^2}\:(\hat{u}_1.\hat{e}_T)=- \lambda_1\sqrt{\lambda_3^2-1}\:(\hat{u}_3.\hat{e}_T).
        \notag
    \end{equation}
    In order to have a Type-II triple cluster in column 2, \eqref{eq:CC2} must be satisfied for the Type-II solution. Let $(U_1,U_6)$, related by $\hat{e}_S=\frac{1}{\sqrt{2}}\{1,0,1\}$, be a representative pair of column 2. Then we have\\[5 pt]
    \begin{equation}
        \sqrt{1-\lambda_1^2}\:(\hat{u}_1.\hat{e}_S)= \sqrt{\lambda_3^2-1}\:(\hat{u}_3.\hat{e}_S).\label{eq:86}
    \end{equation}
    We choose the plus sign here, because $(\hat{u}_1.\hat{e}_S)$ and $(\hat{u}_3.\hat{e}_S)$ are positive quantities. Then, for both the conditions to hold simultaneously, we must have \\[5pt]
    \begin{equation}
       \lambda_3 \frac{(\hat{u}_1.\hat{e}_T)}{(\hat{u}_1.\hat{e}_S)} = - \lambda_1\frac{(\hat{u}_3.\hat{e}_T)}{(\hat{u}_3.\hat{e}_S)}.
          \notag
    \end{equation}\\[5pt]
    Substituting the expressions for $\hat{u}_1,\hat{u}_3,\hat{e}_T,\hat{e}_S, \lambda_1$ and $\lambda_3$, we get,\\[5pt]
    \begin{equation}
       (a+c+\sqrt{D})(a-c-\sqrt{D}) = - (a+c-\sqrt{D})(a-c+\sqrt{D})
          \notag
    \end{equation}
    \begin{equation}
      \implies c(a-c)=2b^2 \label{eq:87}
    \end{equation}

    If instead, column 1 satisfies \eqref{eq:CC2} for Type-II solutions, then a similar calculation shows that column 2 satisfies \eqref{eq:CC2} for Type-I solutions, if  \\[5pt]
    \begin{equation}
       \frac{1}{\lambda_3} \frac{(\hat{u}_1.\hat{e}_T)}{(\hat{u}_1.\hat{e}_S)} = - \frac{1}{\lambda_1}\frac{(\hat{u}_3.\hat{e}_T)}{(\hat{u}_3.\hat{e}_S)}.
          \notag
    \end{equation}
    \begin{equation}
      \implies a(a-c)=-2b^2 \label{eq:88}
    \end{equation}

     For example, if the pair $(U_1,U_3)$ of column 1 forms a Type-I triple cluster with austenite and \eqref{eq:87} holds, then the pair $(U_1,U_6)$ of column 2 forms a Type-II triple cluster with austenite. Similarly, if the pair $(U_1,U_6)$ of column 2 forms a Type-I triple cluster with austenite and \eqref{eq:88} holds, then the pair $(U_1,U_3)$ of column 1 forms a Type-II triple cluster with austenite. \\[5pt]

    \item It is possible in that a given triple cluster is of the first kind as well as the second kind at the same time. This possibility exists only for pairs of compound twins, i.e. for columns 3 and 4, and not for pairs of Type I/II twins, i.e for columns 1 and 2.\\[5pt]
    For this situation to occur, $det(U)=1$ is a necessary condition, see \autoref{cor:first_second_simul}. Therefore, we conclude that for a compound triple cluster to be of the first kind (parallel shear vectors) and the second kind (same normals) at the same time, the transformation must be volume-preserving. This case bears a close connection to the concept of self-accommodation in martensite \cite{Bhattacharya1992}, and will be explored in detail in a forthcoming work \cite{forthcoming}.

    For example, if the pair $(U_1,U_2)$ of column 3 forms a triple cluster of first kind with austenite and \eqref{eq:75} holds, then the same triple cluster is also a triple cluster of second kind. Similarly, if the same pair $(U_1,U_2)$ of column 3 forms a triple cluster of second kind with austenite and \eqref{eq:75} holds, then the same triple cluster is also a triple cluster of first kind. The same holds for the pair $(U_1,U_5)$ as well.\\[5pt]
    However, the pairs of column 5, i.e. the non-conventional compound domains in cubic to monoclinic-II transformations do not exhibit this property. This is due to the fact that the formation of triple clusters in column 5 requires either \eqref{eq:27} or \eqref{eq:38} to hold, which cannot simultaneously hold with \eqref{eq:75}, see \autoref{rem:commutation_det_one_impossible}.
 \end{enumerate}

 Based on the analysis presented in points (I-VI), we conclude that there are two cases of extreme compatibility conditions in cubic to monoclinic-II transformations.\\[5pt]
 \begin{enumerate}[label=\arabic*)]
     \item  Pairs of column 1 form Type-I triple clusters, pairs of column 2 form Type-II triple clusters and the pairs of columns 3 and 4 form triple clusters of the first kind. \\[5pt]
     For this case to be true, equations \eqref{eq:82},\eqref{eq:85} and \eqref{eq:87} must all hold at the same time. The stretch tensor for the cubic to monoclinic-II transformation has four parameters, namely $a,b,c$ and $d$, see \eqref{eq:22}. After we assume $d=1$, we are left with three parameters, which are uniquely fixed by the three equations, \eqref{eq:82},\eqref{eq:85} and \eqref{eq:87}. After simplification, we obtain the stretch tensor in this case to be\\[5 pt]
    \begin{equation}
        U_1 = 
    \begin{bmatrix}
        1 & 0 & 0 \\ 0 & \frac{3}{\sqrt{2}} & -\frac{1}{\sqrt{2}} \\ 0 & -\frac{1}{\sqrt{2}} & \frac{1}{\sqrt{2}}
    \end{bmatrix}.\label{eq:89}
    \end{equation} 

    \item  Pairs of column 1 form Type-II triple clusters, pairs of column 2 form Type-I triple clusters and the pairs of columns 3 and 4 form triple clusters of the second kind. \\[5pt]
     For this case to be true, equations \eqref{eq:85},\eqref{eq:86} and \eqref{eq:88} must all hold at the same time. Again, the stretch tensor is uniquely fixed for this transformation as\\[5 pt]
    \begin{equation}
        U_1 = 
    \begin{bmatrix}
        1 & 0 & 0 \\ 0 & \frac{1}{\sqrt{2}} & \frac{1}{\sqrt{2}} \\ 0 & \frac{1}{\sqrt{2}} & \frac{3}{\sqrt{2}}
    \end{bmatrix}.\label{eq:90}
    \end{equation} 
 \end{enumerate}
     
 It is easy to check that the sufficiency condition, \eqref{eq:CC3}, is satisfied for all columns 1,2,3 and 4, for both the matrices in \eqref{eq:89} and \eqref{eq:90}. Furthermore, these matrices satisfy \eqref{eq:75}, implying that the transformation is volume preserving. Therefore, all triple clusters formed by the corresponding compound twins of columns 3 and 4 will be of the first kind as well as the second kind at the same time.
 
\begin{remark}
 \label{rem:mono_extreme}
     It is important to note that the conditions of extreme compatibility obtained above do not include the triple clusters for non-conventional twins, i.e column 5 of \autoref{tab:CM}. This is because the matrices obtained in \eqref{eq:89} and \eqref{eq:90} do not satisfy either \eqref{eq:27} or \eqref{eq:38}. From the microstructural perspective, it is more feasible to optimize the compatibility for conventional twins rather than the non-conventional twins. This preference is supported by the experimental evidence that the conventional twins are much more commonly observed than the non-conventional twins.
 \end{remark} 

\begin{remark}
\label{rem:restrictive}
    The matrices obtained in \eqref{eq:89} and \eqref{eq:90} are extremely restrictive (no free parameters) and uniquely determine the lattice parameters of the transformation. However, by relaxing one or more of the constraints given in \cref{eq:85,eq:87,eq:88}, it is possible to introduce free parameters into the transformation matrices. This allows for more flexibility by imposing a relatively weaker compatibility but for a wider range of lattice parameters. Even so, the resulting compatibility conditions remain more stringent than the cofactor conditions.
\end{remark}
\section{New Microstructures}
\label{sec:new_micro}

In this section, we employ the extreme compatibility conditions derived in \autoref{sec:extreme_ortho} to explore additional microstructures that may exhibit zero elastic energy in cubic to orthorhombic transformations. The inclusion of compound twins within the zero-energy framework suggests the possibility of new nucleation mechanisms beyond those permitted by the cofactor conditions. These configurations may contribute to transformation pathways with reduced hysteresis and enhanced reversibility, while also exhibiting a broader range of geometric complexity.\\[5pt]
In addition to facilitating the transformation between austenite and martensite, the extreme compatibility conditions also significantly improve compatibility among martensitic variants themselves. Specifically, under these conditions, the triplet conditions \cite{DellaPorta2022Triplet} are automatically satisfied. Moreover, they enable the formation of new types of crossing twins in which all interfaces belong to the same twin system, i.e. either all Type-I or all Type-II. This microstructural flexibility enhances the mechanical properties of the martensitic phase by enabling smoother variant interactions and reducing internal stresses. We note that the extreme compatibility conditions associated with cubic-to-monoclinic-II transformations give rise to distinct microstructural features, the detailed investigation of which will be pursued in subsequent studies.
The stretch tensors for cubic to orthorhombic transformation are given in \eqref{eq:19}. We import them here and assume $d=1$. 

\begin{gather}
    U_1 = 
    \begin{bmatrix}
        1 & 0 & 0 \\ 0 & a & b \\ 0 & b & a
    \end{bmatrix},\qquad
    U_2 = 
    \begin{bmatrix}
        1 & 0 & 0 \\ 0 & a & -b \\ 0 & -b & a
    \end{bmatrix},\qquad
    U_3 = 
    \begin{bmatrix}
        a & 0& -b \\ 0 & 1 & 0 \\ -b & 0 & a
    \end{bmatrix},\nonumber
    \\\nonumber\\
    U_4 = 
    \begin{bmatrix}
        a & 0& b \\ 0 & 1 & 0 \\ b & 0 & a
    \end{bmatrix},\qquad
    U_5 = 
    \begin{bmatrix}
        a & -b& 0 \\ -b & a & 0 \\ 0 & 0 & 1
    \end{bmatrix},\qquad
    U_6 = 
    \begin{bmatrix}
        a & b& 0 \\ b & a & 0 \\ 0 & 0 & 1
    \end{bmatrix}.\label{eq:91}
\end{gather}

The ordered eigenvalues for $U_1-U_6$ are $a-b,\:1,\:a+b$, where we assume $b>0$, without loss of generality. The corresponding normalized eigenvectors are listed in \autoref{tab:Eigen_O}.
\begin{table}
    \centering
\begin{tabular}{ |c|c|c|c|} 
 \hline
 \multirow{2}{*}{Variant} &  \multicolumn{3}{|c|}{Eigenbasis}\\ 
 \cline{2-4}
  &  $\hat{u}_1$ & $\hat{u}_2$ & $\hat{u}_3$\\
 \hline
 &&&\\
 $U_{1}$ & $\frac{1}{\sqrt{2}}\{0,-1,1\}$ & $\{1,0,0\}$ & $\frac{1}{\sqrt{2}}\{0,1,1\}$\\
 \hline
 &&&\\
 $U_{2}$ & $\frac{1}{\sqrt{2}}\{0,1,1\}$  & $\{1,0,0\}$& $\frac{1}{\sqrt{2}}\{0,-1,1\}$\\
 \hline
 &&&\\
 $U_{3}$ & $\frac{1}{\sqrt{2}}\{-1,0,1\}$ & $\{0,1,0\}$ & $\frac{1}{\sqrt{2}}\{1,0,1\}$\\
 \hline
 &&&\\
$U_{4}$ & $\frac{1}{\sqrt{2}}\{1,0,1\}$ & $\{0,1,0\}$ & $\frac{1}{\sqrt{2}}\{-1,0,1\}$\\
 \hline
 &&&\\
 $U_{5}$ & $\frac{1}{\sqrt{2}}\{-1,1,0\}$ & $\{0,0,1\}$ & $\frac{1}{\sqrt{2}}\{1,1,0\}$\\
 \hline
 &&&\\
 $U_{6}$ & $\frac{1}{\sqrt{2}}\{1,1,0\}$ & $\{0,0,1\}$ & $\frac{1}{\sqrt{2}}\{-1,1,0\}$\\
 \hline

\end{tabular}
\caption{The orthonormal eigenbases for the orthorhombic variants given by \eqref{eq:91}, see also \autoref{fig:ortho_b}.}
\label{tab:Eigen_O}
\end{table}

Since the middle eigenvalue of the stretch tensors is unity, rank-1 connections between the martensitic variants and austenite exist. Let,

\begin{equation}
    R_i^{\pm}U_i-I=b_i^{\pm} \otimes \hat{m}_i^{\pm},\label{eq:92}
\end{equation}

where $i=1,2,...,6$. The expressions for $b_i^{\pm}$ and $\hat{m}_i^{\pm}$ are given by \eqref{eq:11}. 
\begin{equation}
    b_i^{\pm} = \frac{\rho}{\sqrt{\lambda_3^2 - \lambda_1^2}}\left(\lambda_3\sqrt{1-\lambda_1^2}\;\hat{u}_1\; \pm\; \; \lambda_1\sqrt{\lambda_3^2 -1}\;\hat{u}_3 \right)
    \nonumber,
\end{equation}

\begin{equation}
    \hat{m}_{i}^{\pm} = \frac{\lambda_3 - \lambda_1}{\rho\sqrt{\lambda_3^2 - \lambda_1^2}}\left(-\sqrt{1-\lambda_1^2}\;\hat{u}_1\; \pm\; \sqrt{\lambda_3^2 -1}\;\hat{u}_3 \right).
    \nonumber
\end{equation}

where $\hat{u}_1,\hat{u}_3$ for each choice of $i$ are taken from \autoref{tab:Eigen_O}.\\[5pt]
If \eqref{eq:27} holds, then we have,
\begin{equation}
    \lambda_3=\frac{\lambda_1}{\sqrt{2\lambda_1^2 -1}} \implies \lambda_3\sqrt{1-\lambda_1^2}\;=\;\lambda_1\sqrt{\lambda_3^2-1}.
    \notag
\end{equation}

Substituting this result into the expression of $b_i^\pm$ and $\hat{m}_i^\pm$, we get,
\begin{gather}\nonumber
    b_i^{\pm} = C(\hat{u}_1 \pm \hat{u}_3),\\[3 pt]\nonumber
    \nonumber
    \hat{m}_i^{\pm} = D(-\sqrt{ 2\lambda_1^2-1}\; \hat{u}_1 \pm \hat{u}_3).
\end{gather}
where $C$ and $D$ are one-parameter constants ($\lambda_1$), given by \eqref{eq:34}. A direct computation shows that 
\begin{gather}
\nonumber
    b_1^+=b_2^+=b_3^+=b_4^+ = \sqrt{2}C\begin{pmatrix}
        0\\0\\1
    \end{pmatrix}, \qquad \alpha_1\hat{m}_1^++\alpha_2\hat{m}_2^++\alpha_3\hat{m}_3^++\alpha_4\hat{m}_4^+ \neq 0,\\[3pt]
   - b_1^-=b_2^-=b_5^+=b_6^+ = \sqrt{2}C\begin{pmatrix}
        0\\1\\0
    \end{pmatrix}, \qquad \alpha_1\hat{m}_1^-+\alpha_2\hat{m}_2^-+\alpha_3\hat{m}_5^++\alpha_4\hat{m}_6^+ \neq 0 \label{eq:93},\\[3pt]
    \nonumber
    b_3^-=-b_4^-=b_5^-=-b_6^- = \sqrt{2}C\begin{pmatrix}
        1\\0\\0
    \end{pmatrix}, \qquad \alpha_1\hat{m}_3^-+\alpha_2\hat{m}_4^-+\alpha_3\hat{m}_5^-+\alpha_4\hat{m}_6^- \neq 0. 
\end{gather}\\[5 pt]
for any real values of $\alpha_1,\alpha_2,\alpha_3,\alpha_4$, not all zero at once. Based on \eqref{eq:93}, we can categorize the twelve deformation gradients ($R_i^{\pm}U_i\:,\;\; i=1,2,...,6$) into three groups, each consisting of four deformation gradients. We call these groups 'quartets', see \autoref{tab:Quartets}. The deformations in each quartet are rank-1 compatible with austenite with the same shear vector (up to a change of sign) and the normals are not coplanar. This gives rise to a spearhead-shaped nucleus of martensite, surrounded by austenite, see \autoref{fig:spearhead}.\\[5 pt]
\begin{table}[hbt!]
    \centering
\begin{tabular}{ |c|c|c|} 
 \hline
 Quartet 1 & Quartet 2 & Quartet 3\\
 \hline
 &&\\
 $R_1^+U_{1}$ & $R_1^-U_{1}$ & $R_3^-U_{3}$\\
 \hline
 &&\\
 $R_2^+U_{2}$ & $R_2^-U_{2}$ & $R_4^-U_{4}$\\
 \hline

 &&\\
 $R_3^+U_{3}$ & $R_5^+U_{5}$ & $R_5^-U_{5}$\\
 \hline
 &&\\
 $R_4^+U_{4}$ & $R_6^+U_{6}$ & $R_6^-U_{6}$\\
 \hline
\end{tabular}

\caption{Classification of all possible deformation gradients that satisfy (82) into three quartets.}
\label{tab:Quartets}
\end{table}\\[5pt]
\begin{figure}[hbt!]
    \centering
    \begin{minipage}[b]{0.32\textwidth}
    \centering
    \includegraphics[scale=0.29]{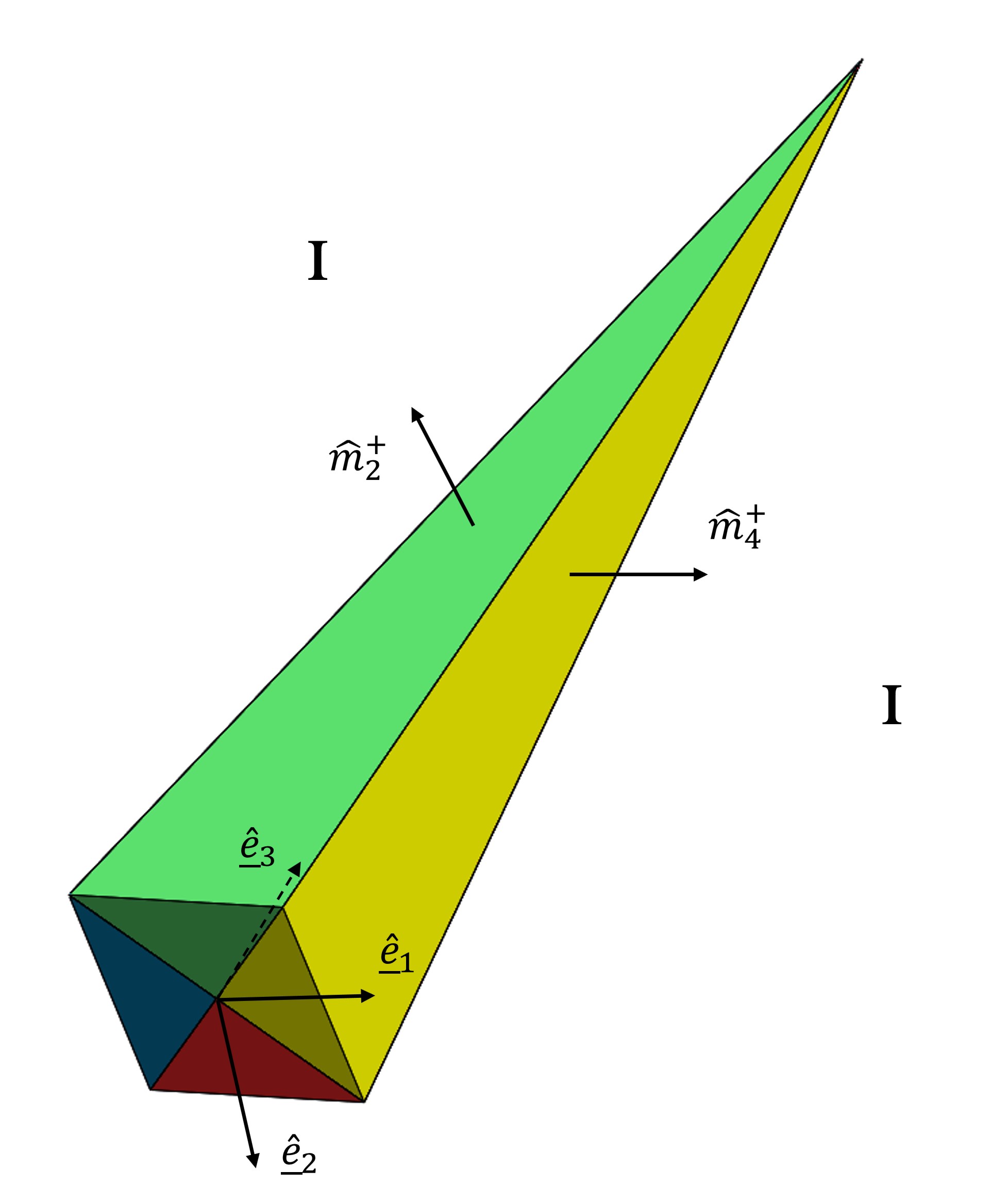}
    \subcaption{}
        \label{fig:spearhead_a}
    \end{minipage}%
    \begin{minipage}[b]{0.32\textwidth}
    \centering
    \includegraphics[scale=0.29]{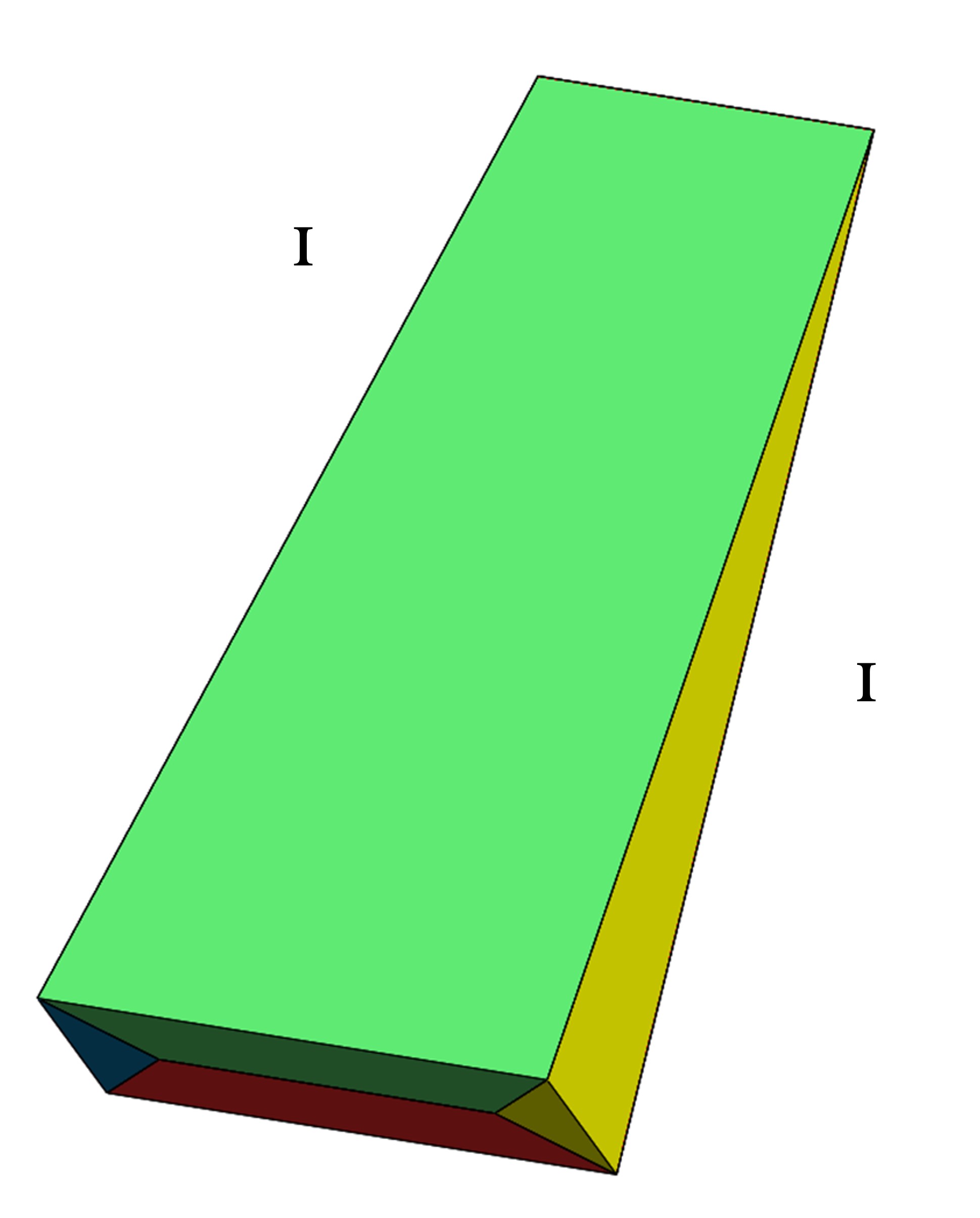}
    \subcaption{}
        \label{fig:spearhead_b}
    \end{minipage}%
     \begin{minipage}[b]{0.32\textwidth}
    \centering
    \includegraphics[scale=0.29]{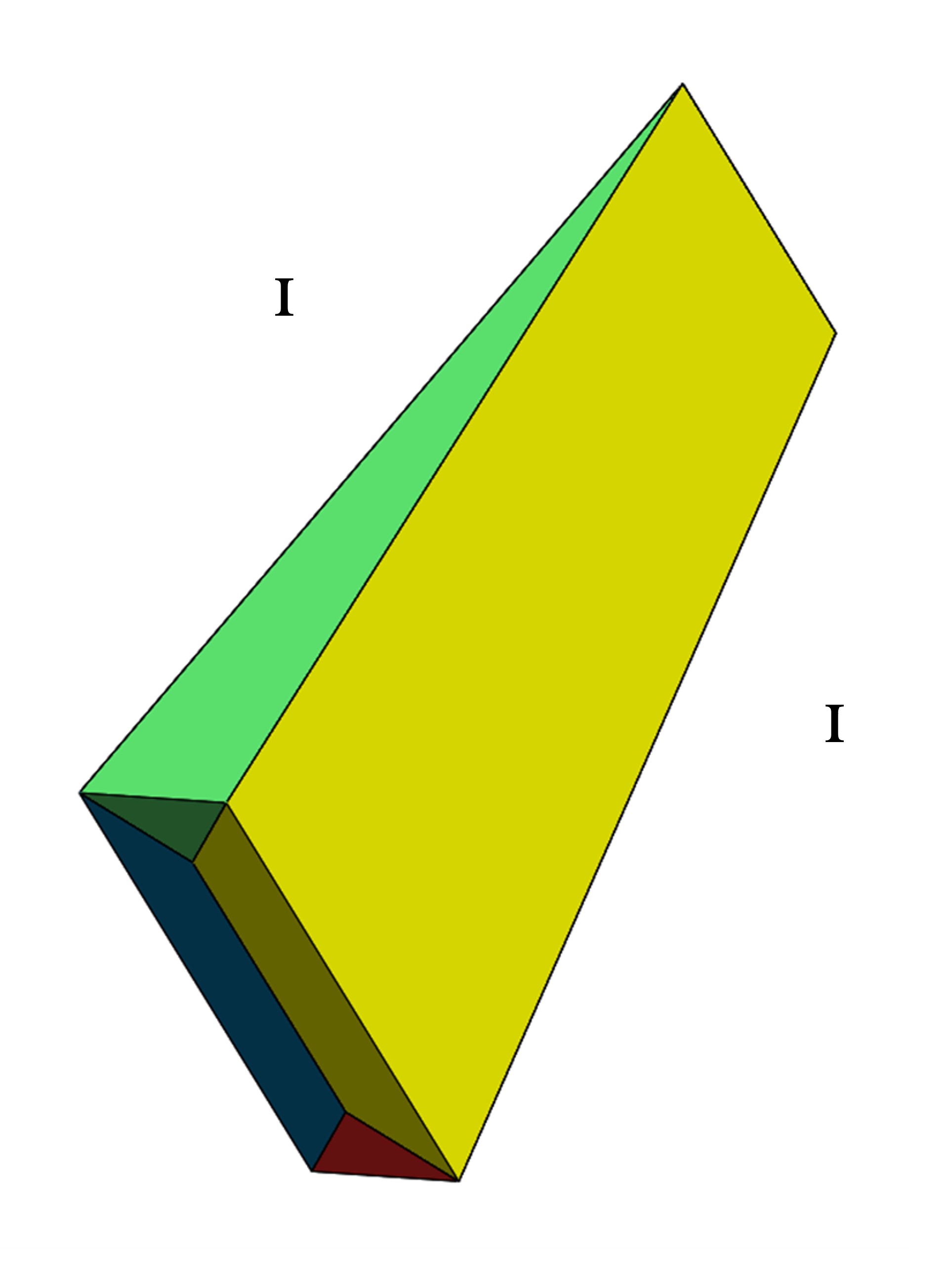}
    \subcaption{}
        \label{fig:spearhead_c}
    \end{minipage}%
    \caption{(a) A spearhead nucleus of four variants, surrounded by austenite, formed by the deformations that constitute quartet 1. Red, green, blue and yellow regions have gradients $R_1^+U_1,\; R_2^+U_2,\; R_3^+U_3$ and $R_4^+U_4$ respectively (b) Elongation of the nucleus facilitated by the formation of a compound interface between red and green regions (c) Elongation of the nucleus facilitated by the formation of a compound interface between blue and yellow regions.}
    \label{fig:spearhead}
\end{figure}

\begin{figure}[hbt!]
    \centering
    \includegraphics[scale=0.75]{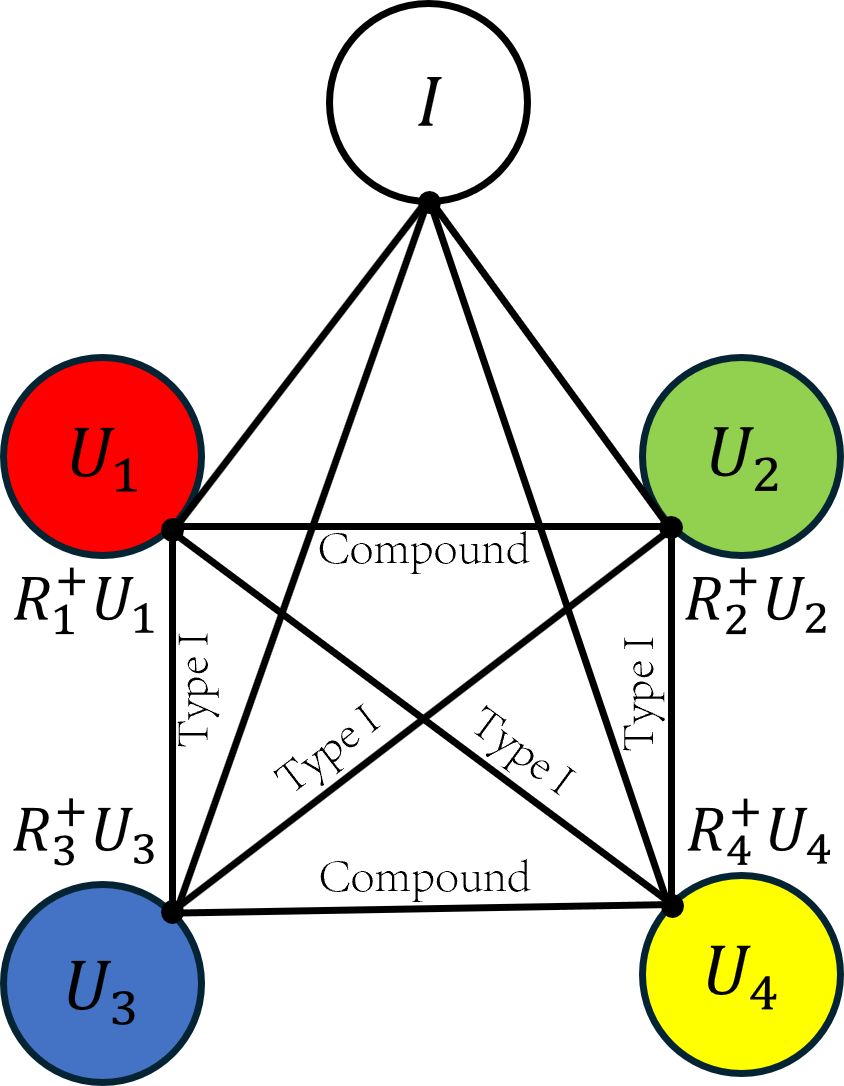}
    \caption{Compatibility diagram in strain space for quatret 1 under optimal compatibility conditions}.
    \label{fig:quartet_1_comp}
\end{figure}

The spearhead-shaped nucleus comprises of four martensitic variants that are mutually compatible and also compatible with the surrounding austenite, see \autoref{fig:quartet_1_comp}. The axis of the spearhead points along the direction of the common shear vector, i.e. along $\hat{\underline{e}}_3$ direction (for quartet 1). As shown in \autoref{fig:spearhead_a}, the nucleus contains four austenite–martensite interfaces and four martensite–martensite interfaces, all of which are Type-I interfaces.\\[5 pt]
An important feature of this nucleus is its ability to elongate by forming an additional, fifth martensite-martensite interface—specifically, a compound interface, without losing compatibility. This interface can emerge in one of two configurations: either between the red and green variants, \autoref{fig:spearhead_b}, or between the blue and yellow variants, \autoref{fig:spearhead_c}. In both scenarios, a planar triple cluster of the first kind is established between the compound twin pair and the austenite.\\[5 pt]
The spearhead nucleus presents the first explicit example of a microstructure featuring triple clusters involving compound twins. It also brings out the significance of such triple clusters in introducing an additional degree of freedom, thereby enhancing the mobility of the nucleus. This increased mobility facilitates a more flexible, bidirectional growth behavior, distinguishing this configuration from triple clusters formed solely with Type-I twins.\\[5 pt]
From \autoref{tab:Quartets}, two more similar but distinct spearhead configurations can be identified, with their axes pointing along the $\hat{\underline{e}}_2$ and $\hat{\underline{e}}_1$ directions. These nuclei can be constructed from the deformation gradients in second and third columns (quartets) respectively. Thus, three such nuclei are possible, which grow in mutually orthogonal directions inside a material with cubic austenite.\\[5 pt]
In addition to spearhead martensitic nucleus, the extreme compatibility conditions for cubic to orthorhombic transformations also enable the formation of fully compatible stress-free inclusions of austenite embedded within the martensitic matrix. We present two such configurations, shown in \autoref{fig:inclusion1} and \autoref{fig:inclusion2}. These configurations can be constructed by combining different spearhead configurations in \autoref{fig:spearhead}, as illustrated in \autoref{fig:inclusion1_a} and \autoref{fig:inclusion2_a}. 

\begin{figure}[hbt!]
    \centering
    \begin{minipage}[b]{0.5\textwidth}
    \centering
    \includegraphics[scale=0.45]{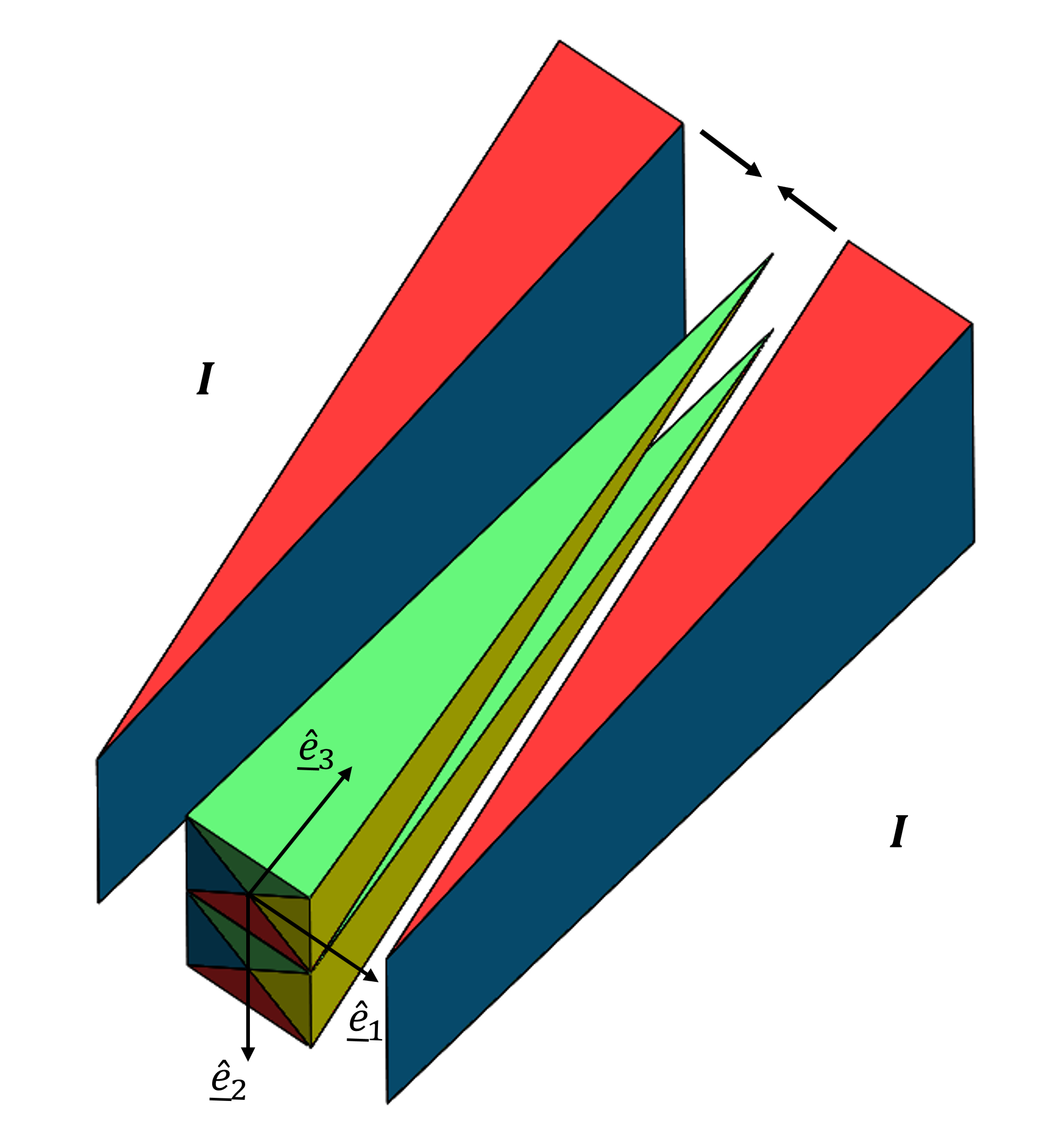}
    \subcaption{}
        \label{fig:inclusion1_a}
    \end{minipage}%
    \begin{minipage}[b]{0.5\textwidth}
    \centering
    \includegraphics[scale=0.45]{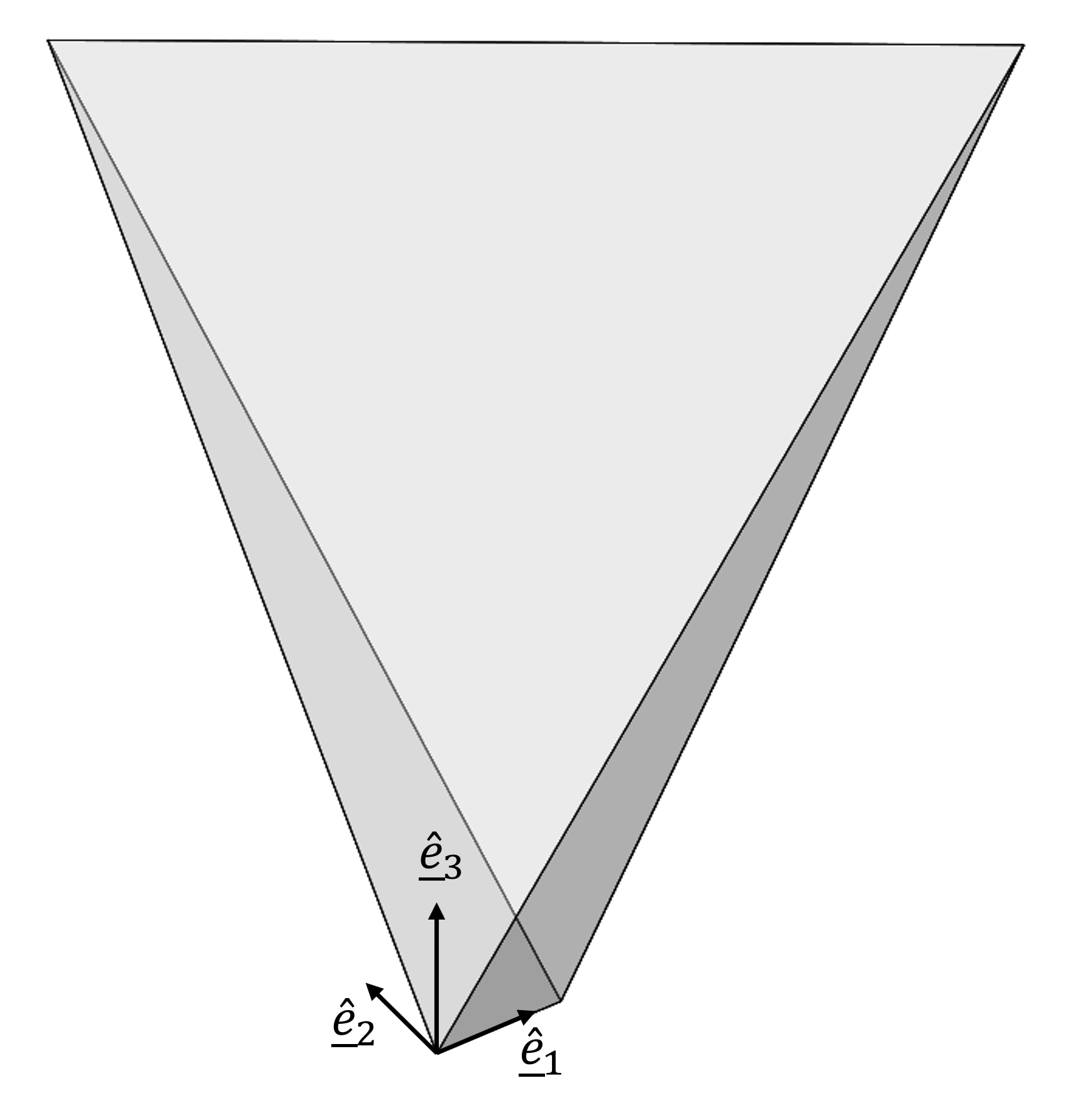}
    \subcaption{}
        \label{fig:inclusion1_b}
    \end{minipage}%
    \caption{(a) Construction of an inclusion of austenite using the spearhead configurations shown in \autoref{fig:spearhead}. The black arrows indicate the direction in which the two spearheads must move in order to enclose the volume of austenite in between the spearheads. (b) A \textit{rhombic disphenoid} shaped austenite inclusion formed by the construction shown in (a). }
    \label{fig:inclusion1}
\end{figure}

\begin{figure}[hbt!]
    \centering
    \begin{minipage}[b]{0.6\textwidth}
    \centering
    \includegraphics[scale=0.585]{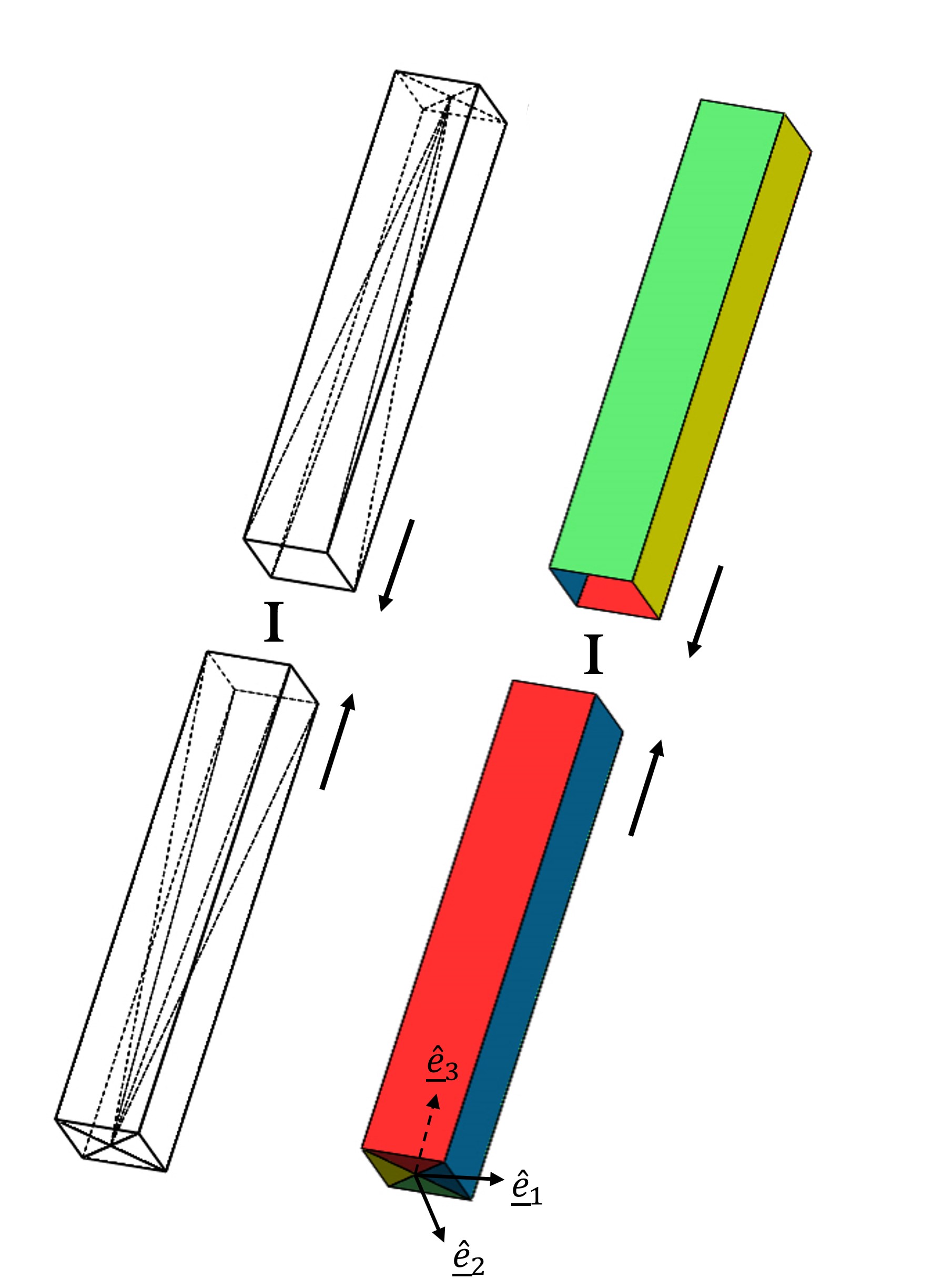}
    \subcaption{}
        \label{fig:inclusion2_a}
    \end{minipage}%
    \begin{minipage}[b]{0.35\textwidth}
    \centering
    \raisebox{1 em}{
    \includegraphics[scale=0.475]{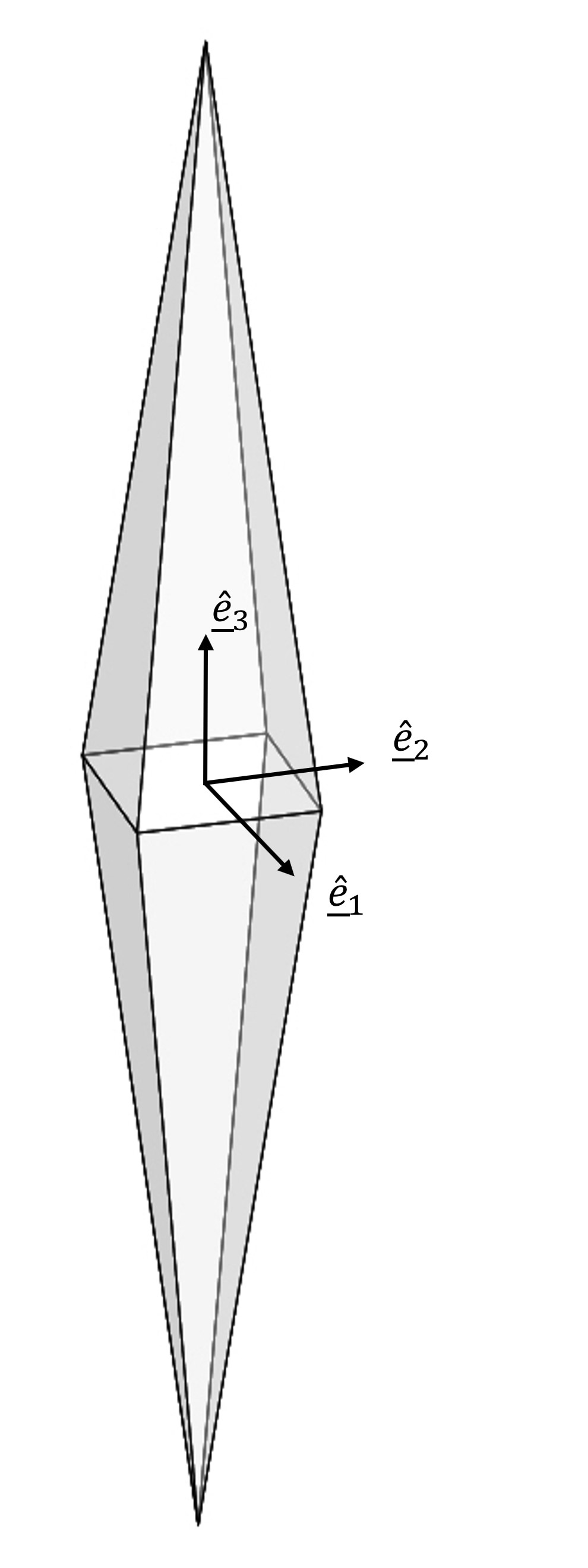}
    }
    \subcaption{}
        \label{fig:inclusion2_b}
    \end{minipage}%
   \caption{(a) Construction of an inclusion of austenite using the inverted form of spearhead configuration shown in \autoref{fig:spearhead_a}. The inverted spearhead is formed by interchanging the austenite and martensite regions across every phase interface, while maintaining compatibility. The wire-frame representation shows the internal structure of the inverted spearheads. The black arrows indicate the direction in which the two inverted spearheads must move in order to enclose a finite volume of austenite between them (b) A \textit{rhombic dipyramid} shaped austenite inclusion formed by the construction shown in (a).}
       \label{fig:inclusion2}
\end{figure}

The crystallography of the inclusions shown in \autoref{fig:inclusion1_b} and \autoref{fig:inclusion2_b} is particularly interesting as they are members of the family of 47 crystallographic forms \cite{graef012structure}. The inclusion \autoref{fig:inclusion1_b} has a \textit{rhombic disphenoid} shape, characterized by the point group $222\:(D_2)$, while the inclusion \autoref{fig:inclusion2_b} is a \textit{rhombic dipyramid} shape, having a point group of $mmm\:(D_{2h})$. Together, these two geometries are the only members of the crystallographic form family that possess an orthorhombic point group and enclose a finite volume. This is especially significant for martensitic transformations, as they represent the first known configurations that enclose a finite region of austenite within the martensitic phase through geometrically exact interfaces.\\[5 pt]
The extreme compatibility conditions in cubic to orthorhombic transformations, not only improve the austenite-martensite compatibility, but also enhance the martensite-martensite compatibility. To understand this, consider again, quartet 1 from \autoref{tab:Quartets}. Then by using \eqref{eq:92} and \eqref{eq:93}, we can write,

\begin{equation}
    R^+_iU_i-R^+_jU_j\;=\;\sqrt{2}C\begin{pmatrix}
        0\\0\\1
    \end{pmatrix}\otimes\:(\hat{m}^+_i-\hat{m}^+_j)\label{eq:94}
\end{equation}

where $i=1,2,3$ and $j=i+1,\;i+2,\;...,\;4$. A direct computation shows,\\[5 pt]
\begin{equation}
    \hat{m}^+_1-\hat{m}^+_2\;=\;2\beta\begin{pmatrix}
        0\\1\\0
    \end{pmatrix},\qquad
    \hat{m}^+_1-\hat{m}^+_3\;=\;\beta\begin{pmatrix}
        1\\1\\0
    \end{pmatrix},\qquad
   \hat{m}^+_1-\hat{m}^+_4\;=\;\beta\begin{pmatrix}
        -1\\1\\0
    \end{pmatrix},
    \notag
\end{equation}

\begin{equation}
    \hat{m}^+_2-\hat{m}^+_3\;=\;\beta\begin{pmatrix}
        1\\-1\\0
    \end{pmatrix},\qquad
    \hat{m}^+_2-\hat{m}^+_4\;=\;\beta\begin{pmatrix}
        -1\\-1\\0
    \end{pmatrix},\qquad
    \hat{m}^+_3-\hat{m}^+_4\;=\;2\beta\begin{pmatrix}
        -1\\0\\0
    \end{pmatrix},\label{eq:95}
\end{equation}

where $\beta=-\frac{\lambda _1^2-1}{\rho\sqrt{2\lambda _1^2-1}}.$\\[5pt]
\begin{figure}[hbt!]
    \centering
    \includegraphics[scale=0.60]{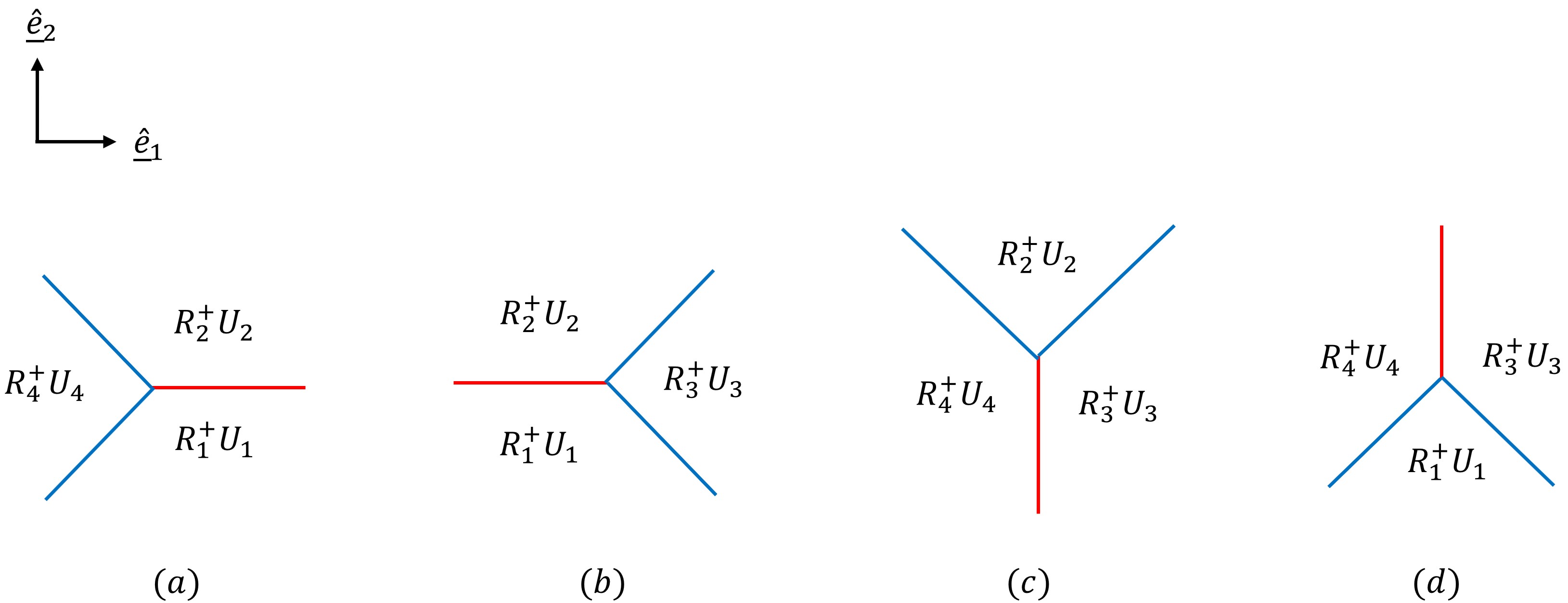}
    \caption{Martensitic triplets for the deformations belonging to quartet 1, when $d=1$ and \eqref{eq:27} is satisfied. Blue lines represent Type-I interfaces and red lines represent compound interfaces.}
    \label{fig:triplets_27}
\end{figure}

\begin{figure}[hbt!]
    \centering
    \includegraphics[scale=0.58]{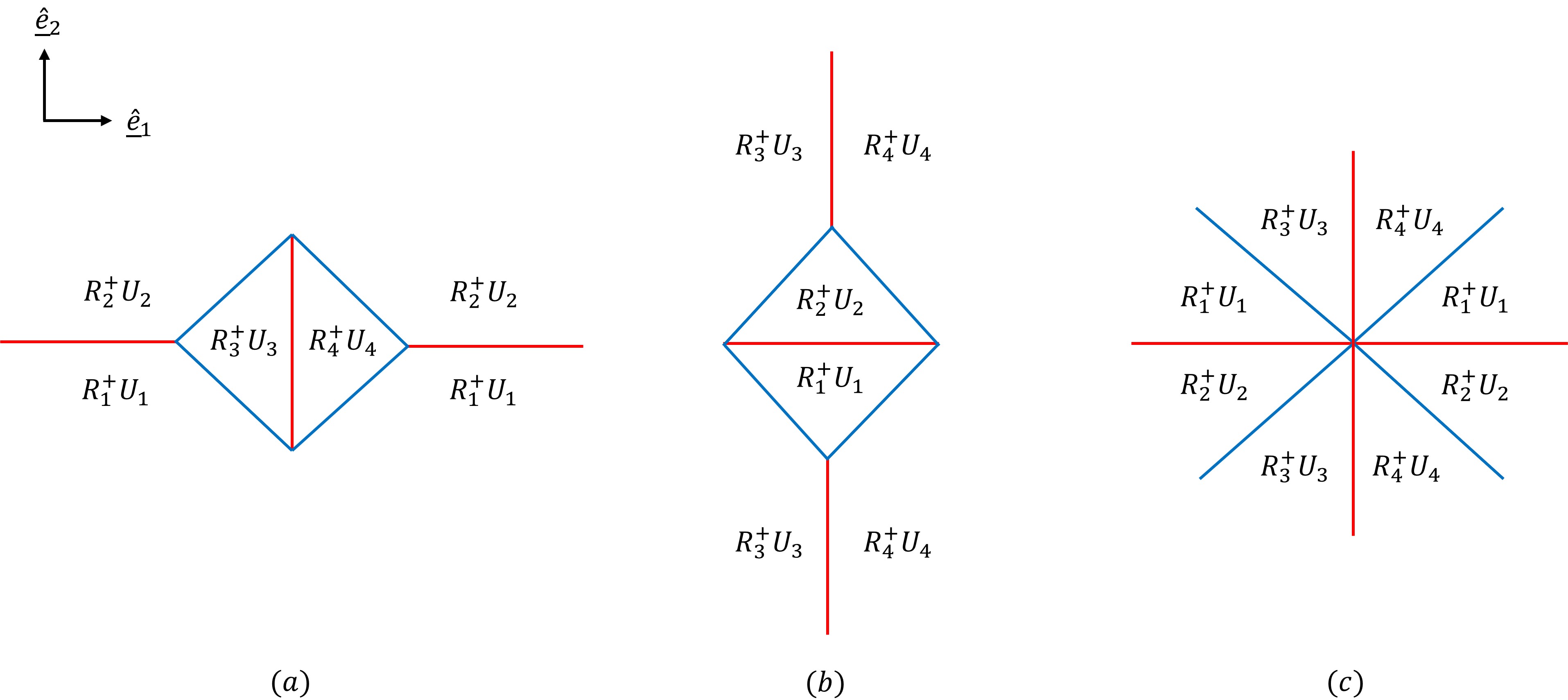}
    \caption{More intricate microstructures constructed from the martensitic triplets shown in \autoref{fig:triplets_27}. }
    \label{fig:more_triplets}
\end{figure}

Using \eqref{eq:94} and \eqref{eq:95}, we can construct several microstructures using three martensitic variants, as illustrated in \autoref{fig:triplets_27}. These three-fold microstructures, known as triplets, differ fundamentally from the triple clusters discussed earlier in this paper. While the triple clusters discussed previously involved two martensitic variants and austenite—highlighting supercompatibility between the two phases, the triplets shown here lie entirely within the martensitic phase. They represent a condition of supercompatibility among martensitic variants, known as the triplet condition, which was introduced by Della Porta \cite{DellaPorta2022Triplet}. The ability to construct martensitic triplets here suggests that the triplet condition must hold under the extreme compatibility conditions for cubic-to-orthorhombic transformations. In what follows, we demonstrate that this is indeed the case.
\begin{table}[hbt!]
    \centering
\begin{tabular}{ |c|c|c|} 
 \hline
 Group 1 & Group 2 & Group 3\\
 \hline
 &&\\
 $(U_1,U_3,U_5)$ & $(U_2,U_3,U_5)$ & $(U_1,U_2,U_3), \:(U_1,U_3,U_4),\;(U_1,U_5,U_6)$\\
 &&\\
 $(U_1,U_4,U_6)$ & $(U_2,U_4,U_6)$ & $(U_1,U_2,U_4),\:(U_2,U_3,U_4),\: (U_3,U_5,U_6)$\\
 &&\\
  $(U_2,U_3,U_6)$ & $(U_1,U_3,U_6)$ & $(U_1,U_2,U_5),\:(U_3,U_4,U_5),\;(U_3,U_5,U_6) $\\
 &&\\
 $(U_2,U_4,U_5)$ & $(U_1,U_4,U_5)$ & $(U_1,U_2,U_6),\:(U_3,U_4,U_6),\;(U_4,U_5,U_6)$\\
  &&\\
 \hline
\end{tabular}
\caption{Classification of all possible variant triples by Della Porta \cite{DellaPorta2022Triplet}. If any one triple from any group satisfies the triplet conditions, then all the triples in that group also satisfy the triplet conditions.}
\label{tab:groups_TC}
\end{table}
Let $U_1$ to $U_6$ as defined in \eqref{eq:19} be the six stretch tensors corresponding to orthorhombic variants. Let $\lambda_1$, $\lambda_2=d=1$ and $\lambda_3$ be the eigenvalues of $U_1$. Based on the relative symmetry of the variants, Della Porta et al. \cite{DellaPorta2022Triplet} have classified all possible variant triples into three groups. This grouping is shown here for reference; see \autoref{tab:groups_TC}. If a triple from any of the three groups satisfies the triplet conditions, then so do all the triples in that group. \\[5 pt]
According to Della Porta et al \cite{DellaPorta2022Triplet}, there are six ways in which triplet conditions can be satisfied in cubic to orthorhombic transformations. These can be expressed by the following six algebraic equations in terms of $\lambda_1$ and $\lambda_3$.

\begin{enumerate}[label=\Alph*)]
    \item $\lambda_3^2\;+2\lambda_1^2\lambda_3^2\;-3\lambda_1^2=0$,
    \item $\lambda_1^2\;+2\lambda_1^2\lambda_3^2\;-3\lambda_3^2=0$,
    \item $2+\lambda_1^2\;-3\lambda_3^2=0$,
    \item $2+\lambda_3^2\;-3\lambda_1^2=0$,
    \item $\lambda_1^2\;-2\lambda_1^2\lambda_3^2\;+\lambda_3^2=0$,
    \item $2-\lambda_1^2\;-\lambda_3^2=0$.
\end{enumerate}

Equations (A) and (B) allow the formation of martensitic triplets for groups 1 and 2, respectively, such that all interfaces are Type-I. Similarly, equations (C) and (D) enable the formation of martensitic triplets for groups 1 and 2, respectively, such that all interfaces are Type-II. Equation (E) allows martensitic triplets in group 3, where two interfaces are Type-I and one interface is compound. Likewise, equation (F) allows martensitic triplets in group 3, with two Type-II interfaces and one compound interface. \\[5 pt]
It is easy to see that equations (E) and (F) can be obtained by squaring equations \eqref{eq:27} and \eqref{eq:38} respectively. This shows that when $d=1$ and either \eqref{eq:27} or \eqref{eq:38} holds, i.e. when the extreme compatibility conditions are satisfied, the triplet conditions for group 3 are automatically achieved.\\[5 pt]
Another interesting feature of the microstructures shown in \Cref{fig:spearhead_a,fig:inclusion1_a,fig:inclusion2_a} is the emergence of a four-fold microstructure composed of four martensitic variants. In this configuration, all interfaces between variants are Type-I interfaces. This is in contrast to the well-known crossing twin structure studied in \cite{Bhattacharya1997}, where the interfaces alternate between Type I and Type II solutions. The significance of this microstructure lies in the fact that, unlike crossing twins which arise naturally from the symmetry of martensite phase, the formation of this four-fold structure requires additional compatibility conditions, which are met under the extreme compatibility conditions for cubic-to-orthorhombic transformations. To show this, we carry out an analysis similar to Bhattacharya's treatment of crossing twins (Theorem 2, \cite{Bhattacharya1997}), in the following theorem. \\[5pt]
\begin{theorem}
    \label{thm:crossing}
    Given $V_m$ $\in \mathbb{R}^{3\times3}_{sym+}  \; and \; Q_m\in SO(3) ,\;m=1,2,3,4,\:$ $V_i\neq V_j\;\forall\;i\neq j$, suppose the following conditions hold:\\[5pt]
    \begin{equation}
        Q_1V_2\:-\:V_1=\;a_1\:\otimes\:\hat{n}_1
        \tag{FF1}\label{eq:FF1}
    \end{equation}
    \begin{equation}
        Q_2V_3\:-\:V_2=\;a_2\:\otimes\:\hat{n}_2
        \tag{FF2}\label{eq:FF2}
    \end{equation}
    \begin{equation}
        Q_3V_4\:-\:V_3=\;a_3\:\otimes\:\hat{n}_3
        \tag{FF3}\label{eq:FF3}
    \end{equation}
    \begin{equation}
        Q_4V_1\:-\:V_4=\;a_4\:\otimes\:\hat{n}_4
        \tag{FF4}\label{eq:FF4}
    \end{equation}
    \begin{equation}
        Q_1Q_2Q_3Q_4=I
        \tag{FF5}\label{eq:FF5}
    \end{equation}
     \begin{equation}
        \hat{n}_1,\:\hat{n}_2,\:\hat{n}_3,\:\hat{n}_4 \quad lie \; on\; a \; plane.
        \tag{FF6}\label{eq:FF6}
    \end{equation}
    \begin{equation}
        V_2=R_1V_1R_1^T,\;V_3=R_2V_2R_2^T,\;V_4=R_2V_1R_2^T 
        \tag{FF7}\label{eq:FF7}
    \end{equation}

    for non-zero vectors $a_m$, unit vectors $\hat{n}_m,\;\hat{e}_1,\: \hat{e}_2$ and rotations  $R_1=R_{\pi}[\hat{e}_1]$, $R_2=R_{\pi}[\hat{e}_2]$, such that $\hat{e}_1.\hat{e}_2=0$.  Then, for equations \eqref{eq:FF1}$\:-\:$\eqref{eq:FF4} to all be Type-I solutions, the following must hold:\\[5pt]
\begin{equation}
    \sum_{i=1}^3 \frac{1}{\lambda_i^2}(\hat{v}_i \cdot \hat{e}_1)(\hat{v}_i \cdot \hat{e}_2) = 0.
    \label{eq:alltype1}
\end{equation}

where $\lambda_i$ are the eigenvalues of $V_1$ and $\hat{v}_i$ are the corresponding eigenvectors. In particular if $d=1$ and \eqref{eq:27} holds for matrices in \eqref{eq:19}, then, \eqref{eq:alltype1} is automatically satisfied for cubic to orthorhombic transformations.\\[5pt]
Similarly, for equations \eqref{eq:FF1}$\:-\:$\eqref{eq:FF4} to all be Type-II solutions, the following must hold:\\[5pt]
\begin{equation}
    \sum_{i=1}^3 {\lambda_i^2}(\hat{v}_i \cdot \hat{e}_1)(\hat{v}_i \cdot \hat{e}_2) = 0.
    \label{eq:alltype2}
\end{equation}
In particular if $d=1$ and \eqref{eq:38} holds for matrices in \eqref{eq:19}, then, \eqref{eq:alltype2} is automatically satisfied for cubic to orthorhombic transformations.
\end{theorem}
\begin{proof}
Let $R_1=(-I+2\:\hat{e}_1\otimes\hat{e}_1)$ and $R_2=(-I+2\:\hat{e}_2\otimes\hat{e}_2)$, where $(\hat{e}_1 \perp\hat{e}_2)$. Then, $R_1R_2\:=\:R_2R_1\:= I \:-2\: \hat{e_1}\otimes\hat{e_1}-2\: \hat{e_2}\otimes\hat{e_2}\:=\: -I \:+2\:\hat{e_3}\:\otimes\:\hat{e_3} $, where $\hat{e}_3 \perp(\hat{e}_1, \hat{e}_2)$. Then, $R_3=R_1R_2=R_2R_1$ is a rotation of $180^\circ$ about an axis perpendicular to both $\hat{e}_1$ and $\hat{e}_2$. \\[5pt]
From \eqref{eq:FF7}, we have $V_3= R_1V_4R_1^T=R_3V_1R_3^T$ and $V_4=R_3V_2R_3^T.$ Then, the Type-I solutions to \cref{eq:FF1,eq:FF2,eq:FF3,eq:FF4} are written as: \\[5 pt]
\begin{equation}
\begin{aligned}
    a_1 &= 2\left(\frac{V_1^{-1} \hat{e}_1}{|V_1^{-1} \hat{e}_1|^2} - V_1 \hat{e}_1\right),\qquad
    \hat{n}_1 = \hat{e}_1,\qquad
    Q_1 = \left(-I + 2 \frac{V_1^{-1} \hat{e}_1 \otimes V_1^{-1} \hat{e}_1}{|V_1^{-1} \hat{e}_1|^2} \right) R_1, \\[4pt]
    a_2 &= 2\left(\frac{V_2^{-1} \hat{e}_2}{|V_2^{-1} \hat{e}_2|^2} - V_2 \hat{e}_2\right),\qquad
    \hat{n}_2 = \hat{e}_2,\qquad
    Q_2 = \left(-I + 2 \frac{V_2^{-1} \hat{e}_2 \otimes V_2^{-1} \hat{e}_2}{|V_2^{-1} \hat{e}_2|^2} \right) R_2, \\[4pt]
    a_3 &= 2\left(\frac{V_3^{-1} \hat{e}_1}{|V_3^{-1} \hat{e}_1|^2} - V_3 \hat{e}_1\right),\qquad
    \hat{n}_3 = \hat{e}_1,\qquad
    Q_3 = \left(-I + 2 \frac{V_3^{-1} \hat{e}_1 \otimes V_3^{-1} \hat{e}_1}{|V_3^{-1} \hat{e}_1|^2} \right) R_1, \\[4pt]
    a_4 &= 2\left(\frac{V_4^{-1} \hat{e}_2}{|V_4^{-1} \hat{e}_2|^2} - V_4 \hat{e}_2\right),\qquad
    \hat{n}_4 = \hat{e}_2,\qquad
    Q_4 = \left(-I + 2 \frac{V_4^{-1} \hat{e}_2 \otimes V_4^{-1} \hat{e}_2}{|V_4^{-1} \hat{e}_2|^2} \right) R_2.
\end{aligned}
\label{eq:type1solutions}
\end{equation}
    
Expressing \cref{eq:type1solutions} in terms of $V_1$ and using the relations $R_1\hat{e}_1=\hat{e_1},\;\:R_2\hat{e}_1= -\hat{e}_2$, and  $R_3\hat{e}_1=-\hat{e_1}$, we have the following: 
\\[5 pt]
\begin{equation}
\begin{aligned}
    a_1 = 2\left(\frac{V_1^{-1} \hat{e}_1}{|V_1^{-1} \hat{e}_1|^2} - V_1 \hat{e}_1\right),\qquad
    \hat{n}_1 = \hat{e}_1,\qquad
    Q_1 &= \left(-I + 2 \frac{V_1^{-1} \hat{e}_1 \otimes V_1^{-1} \hat{e}_1}{|V_1^{-1} \hat{e}_1|^2} \right) R_1, \\[4pt]
    a_2 = 2R_1\left(-\frac{V_1^{-1} \hat{e}_2}{|V_1^{-1} \hat{e}_2|^2} + V_1 \hat{e}_2\right),\qquad
    \hat{n}_2 = \hat{e}_2,\qquad
    Q_2 &= R_1\left(-I + 2 \frac{V_1^{-1} \hat{e}_2 \otimes V_1^{-1} \hat{e}_2}{|V_1^{-1} \hat{e}_2|^2} \right)R_1 R_2, \\[4pt]
    a_3 = 2R_3\left(-\frac{V_1^{-1} \hat{e}_1}{|V_1^{-1} \hat{e}_1|^2} + V_1 \hat{e}_1\right),\qquad
    \hat{n}_3 = \hat{e}_1,\qquad
    Q_3 &= R_3\left(-I + 2 \frac{V_1^{-1} \hat{e}_1 \otimes V_1^{-1} \hat{e}_1}{|V_1^{-1} \hat{e}_1|^2} \right) R_3R_1, \\[4pt]
    a_4 = 2R_2\left(\frac{V_1^{-1} \hat{e}_2}{|V_1^{-1} \hat{e}_2|^2} - V_1 \hat{e}_2\right),\qquad
    \hat{n}_4 = \hat{e}_2,\qquad
    Q_4 &= R_2\left(-I + 2 \frac{V_1^{-1} \hat{e}_2 \otimes V_1^{-1} \hat{e}_2}{|V_1^{-1} \hat{e}_2|^2} \right) .
\end{aligned}
\label{eq:type1solutions_substd}
\end{equation}

$\hat{n}_1,\;\hat{n}_2,\;\hat{n}_3$ and $\hat{n}_4$ are all perpendicular to $\hat{e}_3$, so \eqref{eq:FF6} holds for this choice of Type-I solutions. Then, this choice of solutions must also satisfy \eqref{eq:FF5}. Imposing this requirement yields the sufficient conditions for the existence of a four-fold cluster with all Type-I interfaces.\\[7 pt]
From \eqref{eq:type1solutions_substd}, $Q_1Q_2Q_3Q_4= R_\pi[V_1^{-1}\hat{e}_1]\;R_\pi[V_1^{-1}\hat{e}_2]\;R_\pi[V_1^{-1}\hat{e}_1]\;R_\pi[V_1^{-1}\hat{e}_2]$ $= \left(R_\pi[V_1^{-1}\hat{e}_1]\;R_\pi[V_1^{-1}\hat{e}_2]\right)^2$.\\[7pt]
Then, for \eqref{eq:FF5} to hold, either
\[
R_\pi[V_1^{-1}\hat{e}_1]\;R_\pi[V_1^{-1}\hat{e}_2] = I \quad \text{or} \quad R_\pi[V_1^{-1}\hat{e}_1]\;R_\pi[V_1^{-1}\hat{e}_2] = -I + 2\:\hat{c}\otimes\hat{c},
\]
for some unit vector $\hat{c}$. If the former holds, then $V_1^{-1}\hat{e}_1 \parallel V_1^{-1}\hat{e}_2$, which implies
\[
V_1^{-1}\hat{e}_1 \times V_1^{-1}\hat{e}_2 = 0 \;\Longrightarrow\; \det(V_1^{-1})\,V_1(\hat{e}_1 \times \hat{e}_2) = 0.
\]
Since $V_1 \in \mathbb{R}^{3\times3}_{\text{sym}+}$, this implies $\hat{e}_1 \parallel \hat{e}_2$, which is a contradiction. If the latter holds, then the composition of two distinct $180^\circ$ rotations results in another $180^\circ$ rotation about axis the $\hat{c}$. This is possible only if the axes of the two composing rotations are orthogonal. So, we must have
\[
V_1^{-1}\hat{e}_1.V_1^{-1}\hat{e}_2=0 \implies\hat{e}_1.V_1^{-2}\hat{e}_2=0.
\]
$V_1 \in \mathbb{R}^{3\times3}_{\text{sym}+} \implies V_1^{-1} \in \mathbb{R}^{3\times3}_{\text{sym}+}$ and therefore $V_1^{-1}$ admits a polar decomposition, the use of which in the above equation yields \eqref{eq:alltype1}. \\[5pt]
Let, $V_1=U_1$ as in \eqref{eq:19}, and assume that $d=1$ and that \eqref{eq:27} holds. Then, it can be easily verified that \eqref{eq:alltype1} is identically satisfied. Then, for $\hat{e}_1=\frac{1}{\sqrt{2}}\{1,-1,0\}$ and $\hat{e}_2=\frac{1}{\sqrt{2}}\{1,1,0\}$, we have $V_1=U_1$, $V_2=U_4$, $V_3=U_2$ and $V_4=U_3$, which can form a four-fold microstructure composed entirely of Type-I interfaces, see \autoref{fig:spearhead}\textit{(i)}.\\[5pt]
Similarly, if we assume Type-II solutions to \cref{eq:FF1,eq:FF2,eq:FF3,eq:FF4} that satisfy \eqref{eq:FF7}, we can show that, for \eqref{eq:FF5} and \eqref{eq:FF6} to hold, it is necessary that $\hat{e}_1 \cdot V_1^{2} \hat{e}_2 = 0$. This condition, after applying the polar decomposition of $V_1$, leads to \eqref{eq:alltype2}. Furthermore, if $U_1$ satisfies $d=1$ and \eqref{eq:38}, then \eqref{eq:alltype2} is satisfied and the variants $U_1,\:U_4,\:U_2,\:U_3$ can form a four-fold microstructure composed entirely of Type-II interfaces.\\[5pt]
\end{proof}
The significance of extreme compatibility conditions is underlined by their ability to not only facilitate phase transformations but also enhance compatibility among martensitic variants. This is achieved through the formation of martensitic triplets, regular crossing twins, and special crossing twins, as established in \autoref{thm:crossing}. Such configurational flexibility enables considerable strain accommodation. As a result, the material is expected to suppress functional fatigue more effectively and to promote variant switching without generating dislocations under mechanical loading. These attributes can collectively improve the functional performance of the alloy under both thermal and mechanical actuation.
\\[5pt]

\begin{figure}[hbt!]
    \centering
    \includegraphics[scale=0.58]{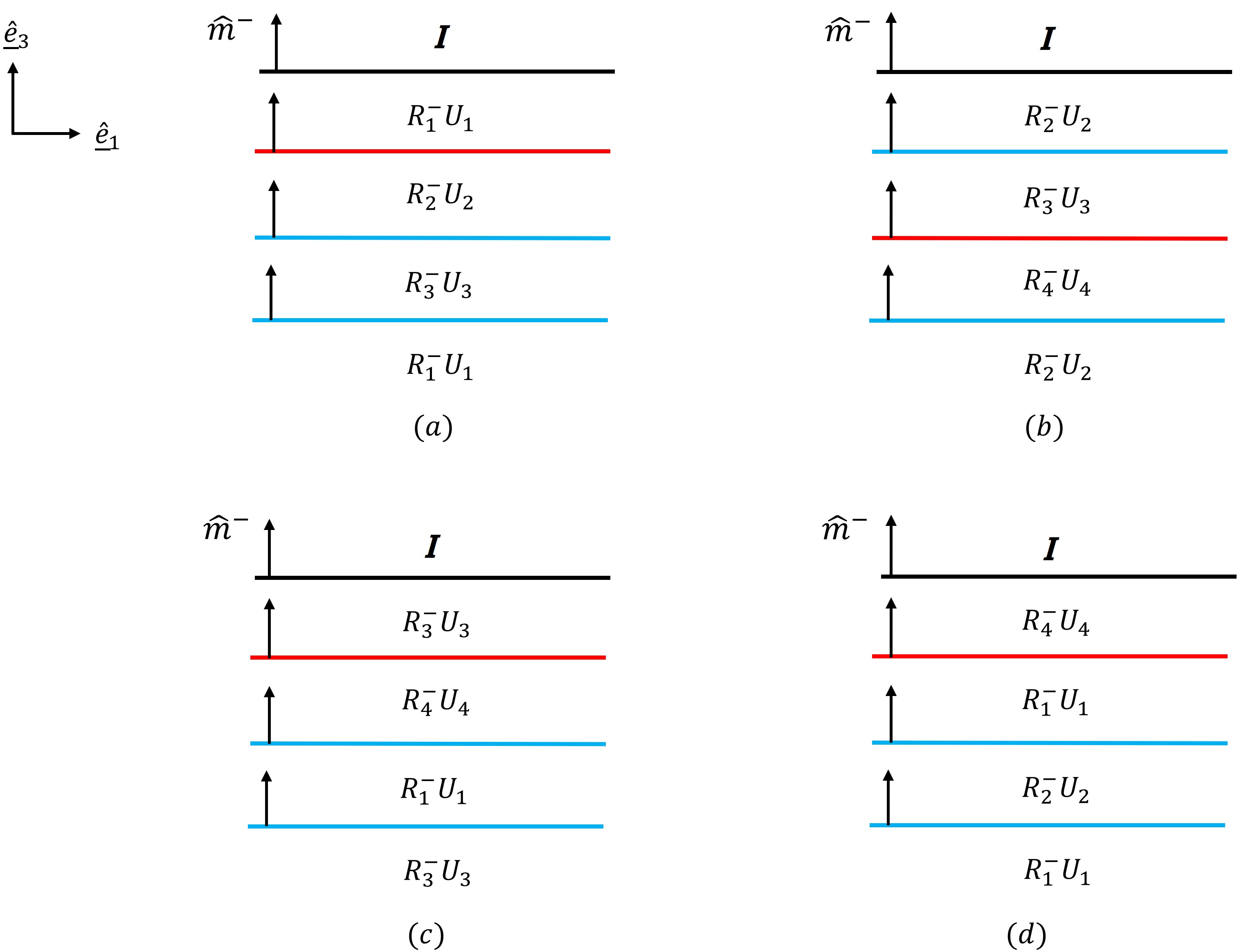}
    \caption{Microstructures possible when $d=1$ and \eqref{eq:38} is satisfied. Blue lines represent Type-II interfaces, red lines represent compound interfaces and black lines represent austenite-martensite interfaces.}
    \label{fig:triplets_38}
\end{figure}

The microstructures shown in \Cref{fig:spearhead,fig:inclusion1,fig:inclusion2,fig:triplets_27,fig:more_triplets} were constructed by assuming $d=1$ and substituting \eqref{eq:27} into \eqref{eq:92}. If, instead, we substitute \eqref{eq:38} and repeat the analysis, we can construct additional microstructures, as illustrated in \autoref{fig:triplets_38}. For brevity, the complete calculations are omitted, but the construction method follows the same steps as in the previous case. In this scenario, equation (E) for martensite–martensite compatibility is satisfied, along with Type-II triple clusters for Type I/II twin pairs and triple clusters of the second kind for compound twin pairs. Additionally, \eqref{eq:alltype2} is satisfied, enabling the formation of special crossing twins involving only Type-II interfaces. As in the previous case, the enhancement of austenite-martensite compatibility as well as the martensite-martensite compatibility improve the overall functional behavior of the material.

\section{Relevance to Experiments and Conclusion}
\label{sec:conclusion}
In this paper, we have extended the theory of zero elastic-energy phase interfaces—originally established through the cofactor conditions for laminates of Type-I/II domains—to include laminates of compound domains. By introducing the notion of extreme compatibility, we have identified conditions under which both Type-I/II and compound laminates can simultaneously form fully compatible, stress-free interfaces with austenite across all volume fractions, without the need for any transition layer. These results optimize the framework of geometric phase compatibility, illustrating the potential of extreme compatibility conditions to promote zero-elastic energy configurations that persist despite defects and heterogeneities, leading to highly reversible phase transformations.\\[5 pt]
A central result of our work is the identification of a connection between the commutation property of martensitic variants and their underlying symmetry. We have shown that compatible variants that commute necessarily form compound domains. Moreover, under specific conditions, such commuting pairs can each give rise to two distinct triple clusters with austenite, in contrast to Type-I/II systems, which allow at most one such cluster per pair, thereby offering a significant advantage over Type-I/II systems in terms of microstructural flexibility. A representative example of this situation is the case of compound twins in cubic to orthorhombic transformations. This commutation framework also offers a natural classification of non-conventional twins, such as those arising in cubic to monoclinic-II transformations, as compound domains, and enables the construction of stress-free interfaces between their laminates and austenite.\\[5 pt]
We have also analyzed cases involving compound domains that do not commute, such as the compound twins in cubic-to-monoclinic transformations. To address such cases, we have derived a general set of conditions under which triple clusters between a pair of compound domains and austenite are possible. Specifically, for any pair of compound domains—commuting or not—to form triple clusters with austenite, at least one of \cref{eq:58,eq:59,eq:61,eq:63} must be satisfied. The commutation-enabled triple clusters are recovered as a special and particularly important case of the general conditions. In all cases, \eqref{eq:CC3} provides the sufficient condition for compatibility between the resulting laminate and austenite. Interestingly, not all compound domains exhibit the same tendency to form triple clusters with austenite. This tendency depends on the symmetry of the martensitic phase. Taken together, these general conditions, along with the well-established cofactor conditions, offer a complete characterization of stress-free interfaces for laminates formed from all types of twin systems. \\[5 pt]
\begin{figure}[hbt!]
    \centering
    \includegraphics[scale=0.53]{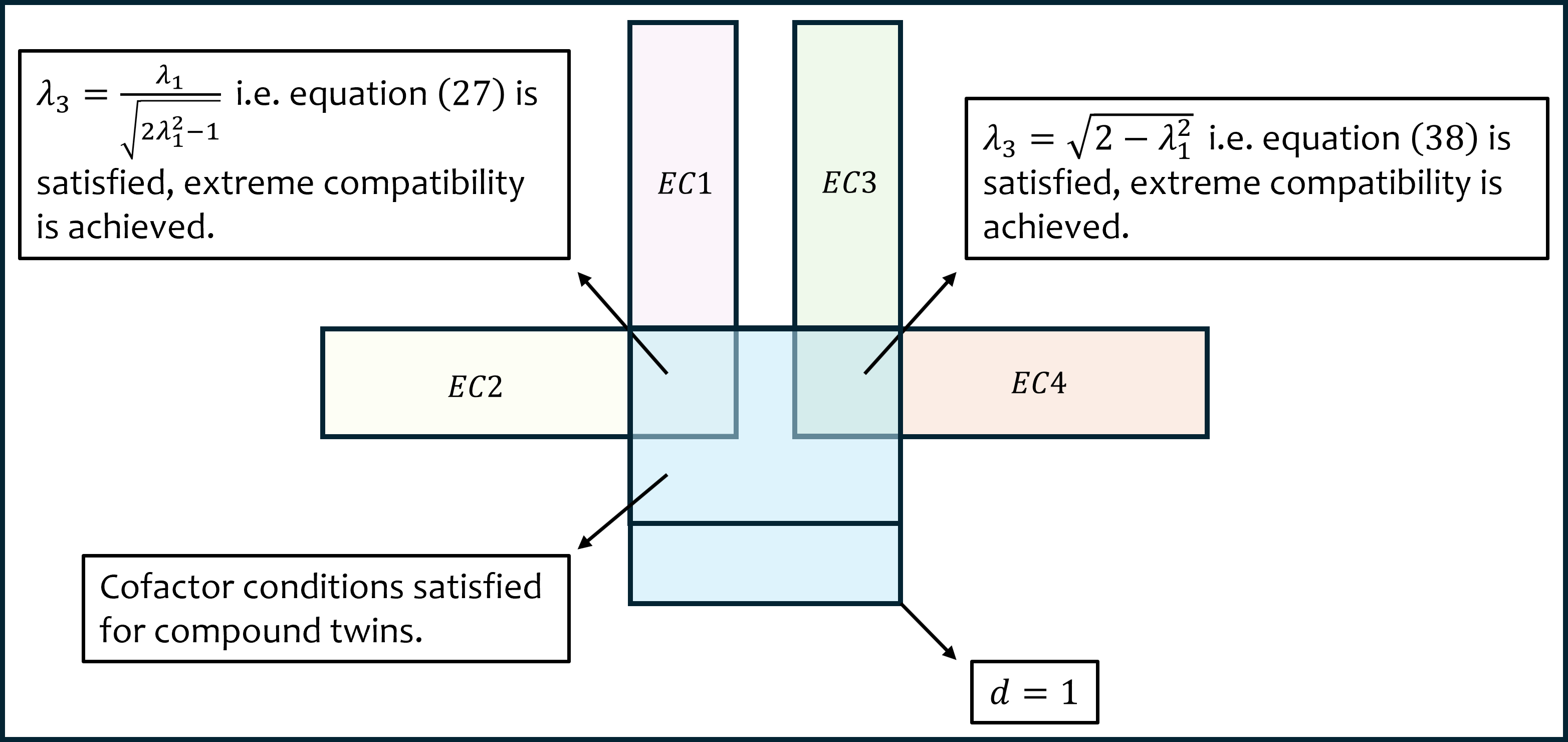}
   \caption{Summary of extreme compatibility conditions for cubic to orthorhombic transformations and relationship with cofactor conditions.}
       \label{fig:5ortho}
\end{figure}

\begin{figure}[hbt!]
    \centering
    \includegraphics[scale=0.55]{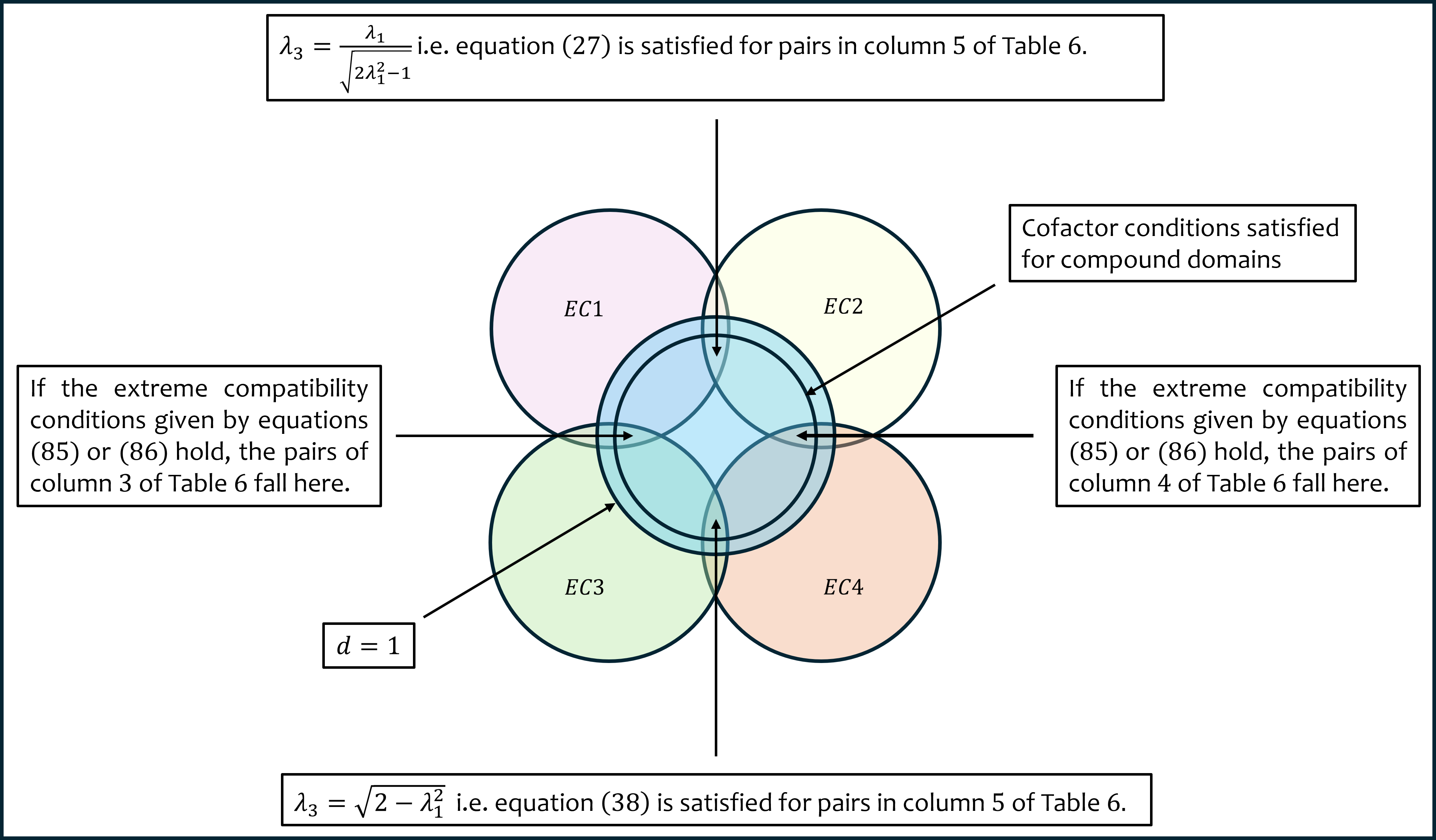}
   \caption{Summary of extreme compatibility conditions for cubic to monoclinic-II transformations and relationship with cofactor conditions.}
       \label{fig:5mono}
\end{figure}
We have studied the extreme compatibility conditions for two cases, \textit{(i)} cubic to orthorhombic transformations and \textit{(ii)} cubic to monoclinic-II transformations. For cubic to orthorhombic transformations, the extreme compatibility conditions are attained when $d=1$ and either \eqref{eq:27} or \eqref{eq:38} is satisfied. These conditions are both necessary and sufficient for laminates of Type-I/II and compound twins to be compatible with austenite across all volume fractions, without requiring any transition layer between the phases. For cubic-to-monoclinic-II transformations, extreme compatibility is achieved when laminates of all conventional twins—that is, the pairs in columns 1 to 4 of \autoref{tab:CM}—are simultaneously compatible with austenite across all volume fractions, again without any need for a transition layer. This is accomplished when the stretch tensor takes the form of either \eqref{eq:89} or \eqref{eq:90}, which imposes strict constraints on the lattice parameters and determines them uniquely. However, this restriction can be slightly relaxed by compromising compatibility in one or more of the columns 1 to 4 of \autoref{tab:CM}. An example of such a case will be discussed in our forthcoming work.\\[5pt]
\Cref{fig:5ortho,fig:5mono} provide a summary of the extreme compatibility conditions for cubic to orthorhombic and cubic to monoclinic-II transformations, respectively. The relationship between extreme compatibility and the cofactor conditions is also highlighted in these figures. For both the transformations, \eqref{eq:CC1} and \eqref{eq:CC2} are satisfied for compound domains whenever $d=1$ holds. Within the set defined by $d=1$ there exists a subset for which \eqref{eq:CC3} is also satisfied, implying complete satisfaction of cofactor conditions. However, in general, the elimination of the transition layer between compound domains and austenite is not guaranteed throughout this region. It is only in the intersection of this region with the extreme compatibility equations that the transition layer is eliminated.
\\[5 pt]
In cubic to orthorhombic transformations, the extreme compatibility conditions give rise to novel microstructures. This includes a flexible stress-free spearhead martensitic nucleus, stress-free inclusions of austenite within martensite, and a special four-fold compatible martensitic cluster. These microstructures not only illustrate the geometric richness made possible by extreme compatibility but also point toward new transformation pathways and improved reversibility in shape memory alloys.\\[5 pt]
We present an example illustrating how the microstructures predicted by our framework correspond to experimental observations reported in the earlier literature. In their investigation on CuAu alloys, Smith and Bowles reported a relevant micrograph in 1960 \cite{Smith1960}. The micrograph shows a diamond-shaped martensitic nucleus embedded within the austenite matrix, see \autoref{fig:AuCu_a}. This nucleus, formed by four twinned laminates of orthorhombic martensite arranged in the form of a pyramid, provides the transformation mechanism in this material.
\begin{figure}[hbt!]
    \centering
    \begin{minipage}[b]{0.5\textwidth}
    \centering
    \includegraphics[scale=0.65]{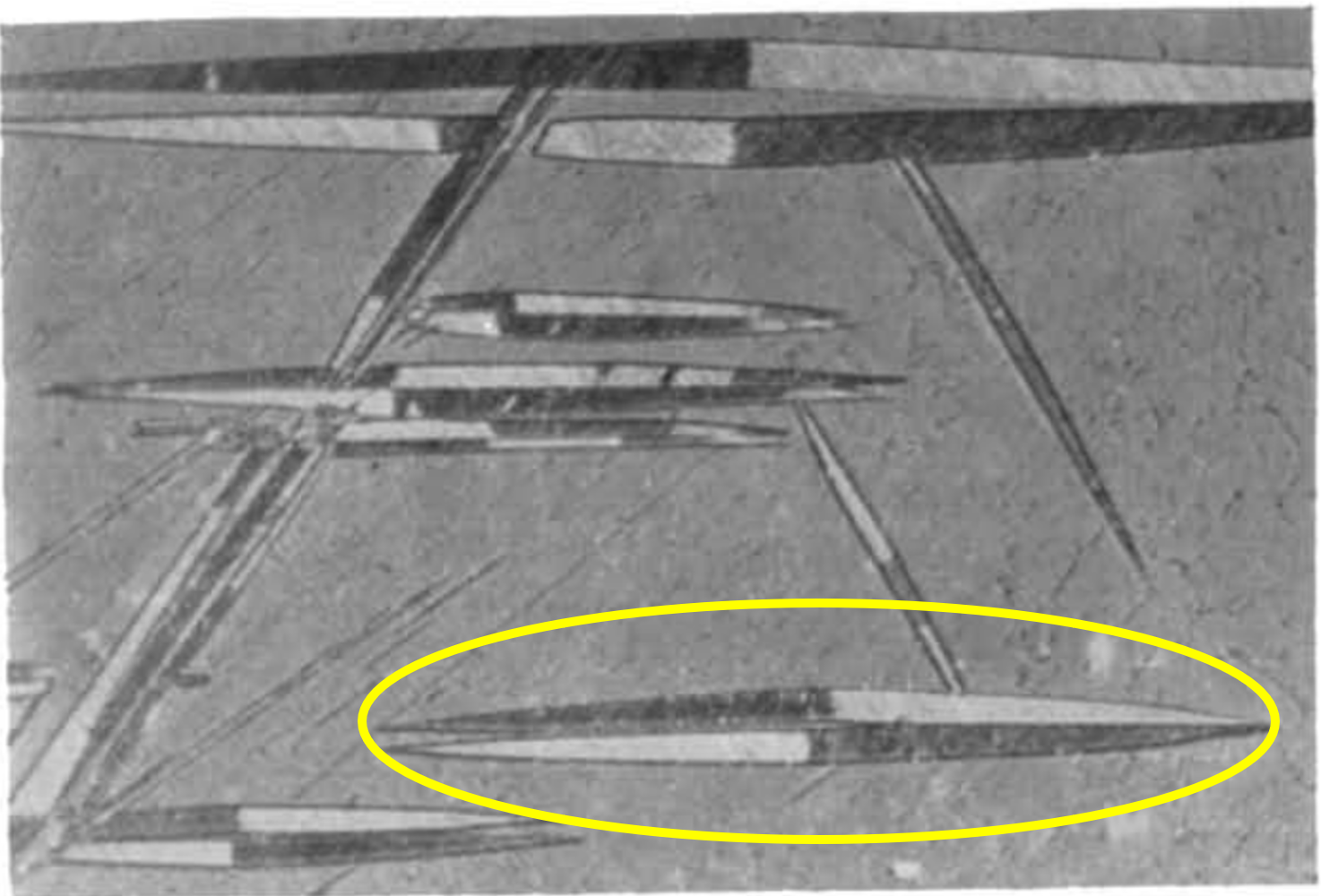}
    \subcaption{}
        \label{fig:AuCu_a}
    \end{minipage}%
    \begin{minipage}[b]{0.5\textwidth}
    \centering
    \includegraphics[scale=0.65]{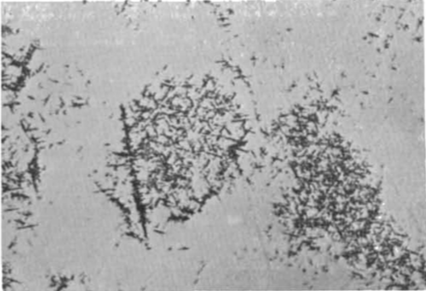}
    \subcaption{}
        \label{fig:AuCu_b}
    \end{minipage}%
   \caption{Micrograph reported in \cite{Smith1960}. (a) A diamond-shaped nucleus composed of martensitic laminates embedded within the austenite matrix, formed during cooling from the austenite phase. (b) A wider view of the specimen showing the repeated formation of similar nuclei throughout the sample.}
       \label{fig:AuCu}
\end{figure}
From the data provided in the original paper, we compute the shape strain associated with the laminate adjacent to the habit plane. Upon extracting the stretch part of the average deformation gradient via polar decomposition given in \eqref{eq:stretch}, we observe that the eigenvalues of these stretch tensor approximately satisfy two key compatibility conditions of our framework: the middle eigenvalue is unity and the other two eigenvalues obey the relationship in \eqref{eq:27}. However, the stretch tensor itself is close to but not strictly orthorhombic, suggesting that the underlying structure does not fall exactly within the traditional classification of martensitic symmetry classes.\\[5pt]
\begin{equation}
\label{eq:stretch}
    U=\left(
\begin{array}{ccc}
 0.9999 & 3.0689\times10^{-6} & -0.0013 \\
 3.0689\times10^{-6} & 1.0221 & 4.1782\times10^{-6} \\
 -0.0013 & 4.1782\times10^{-6} & 0.9734 \\
\end{array}
\right)
\end{equation}

This micrograph is intriguing because it suggests that the microstructure depicted in this micrograph may be, in some sense, an approximate realization of spearhead nucleus discussed in \autoref{sec:new_micro} with a critical distinction that our spearhead nucleus is constructed entirely from single-variant deformations without any transition layers. Further insight is gained by considering the role of this diamond-shaped structure in the transformation process. As shown in \autoref{fig:AuCu_b}, similar nuclei appear repeatedly throughout the microstructure and serve as initiation points for martensitic growth. This promotes a uniform transformation behavior across the sample, potentially enhancing the functional performance of the material. By allowing the formation of such nuclei entirely within a stress-free interface setting, without the need for transition layers, extreme compatibility offers avenues for designing microstructures that promote uniform, reversible and energetically favorable transformation.\\[5pt]
Looking ahead, the experimental realization of extreme compatibility conditions remains a challenge. This requires identifying material systems, both metallic and ceramic, that can meet these stringent geometric criteria. For materials undergoing cubic to orthorhombic transformation, the conditions $d=1$ along with either \eqref{eq:27} or \eqref{eq:38} define the criteria for extreme compatibility. In the case of cubic to monoclinic-II transformations, extreme compatibility is achieved when the stretch tensor takes the form given in \eqref{eq:89} or \eqref{eq:90}. As established in our analysis, $d = 1$ acts as an initial screening criterion for extreme compatibility, though it can be difficult to achieve. We point to two materials from the literature that approximately meet this condition and may provide preliminary directions for crystallographic tuning. \\[5pt]
The first is vanadium dioxide (VO\textsubscript{2}), which undergoes a metal-to-insulator transition at 340 K, accompanied by a structural change from tetragonal to monoclinic phase. In their study, Liang et al. \cite{VO2} demonstrated that substituting tungsten for vanadium in this material allows the condition $d=1$ to be approximately satisfied at a composition of (V\textsubscript{$1–x$}W\textsubscript{$x$}O\textsubscript{2}) with $x=0.025$. The authors also reported a reduced thermal hysteresis for this composition. The second example is a compound system of barium titanate (BaTiO\textsubscript{3}) and barium calcium titanate (Ba\textsubscript{$1–x$}Ca\textsubscript{$x$}TiO\textsubscript{3}), studied by Chen et al. \cite{Chen2023}. This system undergoes two successive phase transformations: a high temperature cubic to tetragonal transition at 388 K, which is ferroelectric to paraelectric and marks the Curie point, followed by a low temperature tetragonal to orthorhombic transformation at 273 K. In the latter transition, the condition $d=1$ is approximately satisfied, as one of the basal axes of the tetragonal unit cell transforms into an orthorhombic axis with minimal dilatation. These examples offer useful reference points for future efforts aimed at tuning lattice parameters to satisfy extreme compatibility conditions.\\[5pt]
In a forthcoming work \cite{forthcoming}, we aim to advance the understanding of martensitic phase transformations by drawing analogies with deformation modes commonly encountered in physical systems. These analogies may offer intuitive frameworks to
understand complex transformation mechanisms and the resulting microstructures. In particular, they facilitate the extension of extreme compatibility to transformations such as cubic to monoclinic-II and even cubic to triclinic, without relying on cumbersome algebraic treatments. They also point toward a potential self-accommodation mechanism and help to interpret certain microstructures that are frequently observed in martensitic transformations.

\section*{CRediT authorship contribution statement}
\textbf{Mohd Tahseen:} Conceptualization, Methodology, Investigation, Formal analysis, Writing – original draft, Writing – review $\&$ editing. \textbf{Vivekanand Dabade:} Conceptualization, Formal analysis, Writing – review $\&$ editing, Validation, Funding acquisition, Supervision.

\section*{Declaration of competing interest}

The authors declare that they have no known competing financial interests or personal relationships that could have appeared to influence the work reported in this paper.

\section*{Acknowledgments}
We sincerely thank Prof. John M Ball for taking an interest in this work during ESOMAT 2024 and offering encouraging remarks. We are also grateful to Prof. Kaushik Bhattacharya for generously sharing a scanned copy of the manuscript Kinematics of Crossing Twins \cite{Bhattacharya1997}, which contributed to the development of our analysis presented in section \ref{sec:new_micro}. Furthermore, we would like to acknowledge the Indian Institute of Science Startup Grant and the Prime Minister’s Research Fellowship (PMRF ID: 0203695) for providing financial support for this research.

\bibliography{references}
\end{document}